\documentclass[12pt]{article}
\usepackage[utf8]{inputenc}
\usepackage{mathrsfs}
\usepackage{natbib}

\usepackage{amsmath,amssymb,amsthm,amsfonts,setspace,fullpage,graphicx,xcolor,mathabx,paralist,paralist,nicefrac,comment,booktabs,multirow, float, caption}
\usepackage[multiple]{footmisc}
\usepackage{thmtools}
\usepackage{thm-restate}
\newtheorem{assumption}{Assumption}

\newtheorem{proposition}{Proposition}

\usepackage{caption}
\usepackage{array}
\usepackage{bbm}
\usepackage[multiple]{footmisc}
\usepackage{centernot}
\usepackage{apptools}
\usepackage[many]{tcolorbox}
\usepackage{longtable}
\usepackage{placeins}
\usepackage{amsthm}
\usepackage{amsmath}
\usepackage{sectsty}

\usepackage[margin=1.25in]{geometry}
\allowdisplaybreaks

 \newtcolorbox{recipebox}[2][]{
  breakable,
  colback=gray!5!white,
  colframe=gray!60!black,
  fonttitle=\bfseries,
  title=#2,
  arc=2mm,
  boxrule=0.8pt,
  #1
}

\makeatletter

\DeclareTextSymbolDefault{\textquotedbl}{T1}

\theoremstyle{definition}
\newtheorem{defn}{\protect\definitionname}
\theoremstyle{definition}
 
\subsectionfont{\fontsize{12}{16}\selectfont\bfseries}
\usepackage{graphicx}
\usepackage{tikz}
\usepackage{algorithm}
\usepackage{algorithmic}
\usepackage{pgfplots}
\pgfplotsset{compat=1.17}

\providecommand{\definitionname}{Definition}
\providecommand{\examplename}{Example}

\begin{document}
\global\long\def\a{\alpha}%
 
\global\long\def\b{\beta}%
 
\global\long\def\g{\gamma}%
 
\global\long\def\d{\delta}%
 
\global\long\def\e{\epsilon}%
 
\global\long\def\l{\lambda}%
 
\global\long\def\t{\theta}%
 
\global\long\def\o{\omega}%
 
\global\long\def\s{\sigma}%

\global\long\def\G{\Gamma}%
 
\global\long\def\D{\Delta}%
 
\global\long\def\L{\Lambda}%
 
\global\long\def\T{\Theta}%
 
\global\long\def\O{\Omega}%
 
\global\long\def\R{\mathbb{R}}%
 
\global\long\def\N{\mathbb{N}}%
 
\global\long\def\Q{\mathbb{Q}}%
 
\global\long\def\I{\mathbb{I}}%
 
\global\long\def\P{\mathbb{P}}%
 
\global\long\def\E{\mathbb{E}}%
\global\long\def\B{\mathbb{\mathbb{B}}}%
\global\long\def\S{\mathbb{\mathbb{S}}}%
\global\long\def\V{\mathbb{\mathbb{V}}\text{ar}}%
 
\global\long\def\GG{\mathbb{G}}%
\global\long\def\TT{\mathbb{T}}%

\global\long\def\X{{\bf X}}%
\global\long\def\cX{\mathcal{X}}%
 
\global\long\def\cY{\mathcal{Y}}%
 
\global\long\def\cA{\mathscr{A}}%
 
\global\long\def\cB{\mathscr{B}}%
\global\long\def\cF{\mathcal{F}}%
 
\global\long\def\cM{\mathscr{M}}%
\global\long\def\cN{\mathcal{N}}%
\global\long\def\cG{\mathcal{G}}%
\global\long\def\cC{\mathcal{C}}%
\global\long\def\sp{\,}%

\global\long\def\es{\emptyset}%
 
\global\long\def\mc#1{\mathscr{#1}}%
 
\global\long\def\ind{\mathbf{\mathbbm1}}%
\global\long\def\indep{\perp}%

\global\long\def\any{\forall}%
 
\global\long\def\ex{\exists}%
 
\global\long\def\p{\partial}%
 
\global\long\def\cd{\cdot}%
 
\global\long\def\Dif{\nabla}%
 
\global\long\def\imp{\Rightarrow}%
 
\global\long\def\iff{\Leftrightarrow}%

\global\long\def\up{\uparrow}%
 
\global\long\def\down{\downarrow}%
 
\global\long\def\arrow{\rightarrow}%
 
\global\long\def\rlarrow{\leftrightarrow}%
 
\global\long\def\lrarrow{\leftrightarrow}%

\global\long\def\abs#1{\left|#1\right|}%
 
\global\long\def\norm#1{\left\Vert #1\right\Vert }%
 
\global\long\def\rest#1{\left.#1\right|}%

\global\long\def\bracket#1#2{\left\langle #1\middle\vert#2\right\rangle }%
 
\global\long\def\sandvich#1#2#3{\left\langle #1\middle\vert#2\middle\vert#3\right\rangle }%
 
\global\long\def\turd#1{\frac{#1}{3}}%
 
\global\long\def\ellipsis{\textellipsis}%
 
\global\long\def\sand#1{\left\lceil #1\right\vert }%
 
\global\long\def\wich#1{\left\vert #1\right\rfloor }%
 
\global\long\def\sandwich#1#2#3{\left\lceil #1\middle\vert#2\middle\vert#3\right\rfloor }%

\global\long\def\abs#1{\left|#1\right|}%
 
\global\long\def\norm#1{\left\Vert #1\right\Vert }%
 
\global\long\def\rest#1{\left.#1\right|}%
 
\global\long\def\inprod#1{\left\langle #1\right\rangle }%
 
\global\long\def\ol#1{\overline{#1}}%
 
\global\long\def\ul#1{\underline{#1}}%
 
\global\long\def\td#1{\tilde{#1}}%
\global\long\def\bs#1{\boldsymbol{#1}}%

\global\long\def\upto{\nearrow}%
 
\global\long\def\downto{\searrow}%
 
\global\long\def\pto{\overset{p}{\longrightarrow}}%
 
\global\long\def\dto{\overset{d}{\longrightarrow}}%
 
\global\long\def\asto{\overset{a.s.}{\longrightarrow}}%

\setlength{\abovedisplayskip}{6pt} \setlength{\belowdisplayskip}{6pt}
\title{Could It Have Fit Anything? Restrictiveness of Economic Models, LLMs, and AI Workflows\thanks{We thank Karun Adusumilli, Isaiah Andrews, Tim Christensen, Annie Liang, Frank Schorfheide, Kevin  Song, and four anonymous reviewers at ACM EC'26 for helpful comments. Fudenberg acknowledges NSF grant SES-2417162 for financial support. This paper was previously circulated under the title ``Model Restrictiveness in Functional and Structural Settings''.}}
\author{Drew Fudenberg\thanks{Department of Economics, Massachusetts Institute of Technology, drewf@mit.edu}, Wayne Yuan Gao\thanks{Department of Economics, University of Pennsylvania,
 waynegao@upenn.edu.}, and Zhiheng You\thanks{Department of Economics, University of Pennsylvania,
 zhyou@sas.upenn.edu.}}
\maketitle

\begin{abstract}

\noindent
\singlespacing

\noindent How much does a prediction procedure rule out when its model class is unspecified? We extend restrictiveness from economic models to black-box procedures by measuring how well they approximate simulated prediction rules satisfying specified economic restrictions. The measure applies to continuum domains, and converges to model-class restrictiveness under risk consistency. We provide simulation-based estimators and inference, and apply the measure alongside empirical completeness in two settings. For lottery certainty equivalents, LLM completion and a code-execution agent attain nearly identical completeness, but the agent is substantially less restrictive. A two-parameter cumulative prospect theory specification achieves completeness close to the LLM procedures and tabular foundation models (TFM) while being substantially more restrictive. For consumer choice over products, an LLM agent and a TFM are less restrictive than fitted multinomial, nested, and mixed logit models in both shares and implied price elasticities. Despite allowing more flexible substitution patterns, nested and mixed logit fit simulated shares better than multinomial logit without producing more accurate price elasticities.

\vspace{1em}

\noindent\textbf{Keywords:} restrictiveness, completeness, black-box procedures, large language models, LLM agents, tabular foundation models, AI, Gaussian process, learning curve, choice under risk, discrete choice.

\end{abstract}
\thispagestyle{empty} \vspace{-1em}

\newpage \setcounter{page}{1}

\section{Introduction}
\label{sec:intro}

Economic models are restrictive by design: they rule out patterns in the data that conflict with economic theories, and this structure is what makes them useful for interpretation and decision making. However, a model's good empirical fit may reflect either the regularities it captures or its flexibility to accommodate many different patterns. Evaluating a model thus requires assessing both how well it fits the observed data and how much structure it imposes.

This question becomes more pressing as prediction rules are increasingly produced by black-box \emph{procedures}, such as machine-learning pipelines, automated model search, tabular foundation models, and workflows built on large language models (LLMs) that read a dataset, write code, and return a fitted prediction rule \citep{breiman2001statistical,korinek2023generative,ludwig2025large}. AI-assisted tools make it easy to estimate vast numbers of candidate models on any given dataset \citep{projectAPE}, exacerbating the concern that good fit may simply reflect enough flexibility to fit anything. 

\citet*[FGL]{fudenberg2023flexible} propose a measure of the restrictiveness of an economic model by its average approximation error on eligible prediction rules drawn from a specified evaluation distribution, normalized by a baseline. Restrictiveness complements empirical completeness, which measures how much of a flexible benchmark's predictive improvement a model achieves on observed data \citep{FKLM}. Models with similar empirical performance (completeness) can differ substantially in how much they rule out (restrictiveness).

We extend restrictiveness from specified model classes to prediction procedures, including black-box procedures whose available prediction rules cannot be explicitly described. To evaluate a procedure, we draw eligible prediction rules from a function space that satisfies given economic restrictions, supply noise-free observations from each rule, and measure prediction error according to a given discrepancy function. Procedural restrictiveness is the average discrepancy relative to a baseline at a given training sample size. It evaluates the procedure as implemented, including model selection, estimation, and computation limits as integral parts of the given procedure. Thus, two workflows using the same language model can differ in restrictiveness, as can two estimators of the same economic model. For a specified model class, we show that procedural restrictiveness converges to model restrictiveness under a mild risk-consistency condition on an estimation procedure for the given model class.

We describe how to compute procedural restrictiveness, including settings  with continuous covariates, where we must sample from an infinite-dimensional functional space. Gaussian-process evaluation distributions provide a practical way to generate eligible prediction rules while incorporating economic shape restrictions. The evaluation distribution specifies how these rules are weighted, and the discrepancy specifies which differences in predictions matter. We provide Monte Carlo estimators and inference at a fixed training sample size.

Procedural restrictiveness can be  computed  using the same pipeline that produces predictions on empirical data. Running this pipeline on simulated data measures restrictiveness, while running it on observed data measures completeness. Together, the two measures assess empirical performance and flexibility without requiring an explicit description of the procedure's model class.

The first application studies certainty equivalents for binary lotteries. We find substantial differences in restrictiveness among both AI procedures and economic models with similar empirical completeness. An LLM completion and a code-execution agent attain nearly identical completeness on experimental lottery responses, but the completion is substantially more restrictive. A two-parameter version of Cumulative Prospect Theory likewise attains nearly the same completeness as the full model while ruling out substantially more. These comparisons show what empirical fit alone leaves unresolved: procedures that predict the observed responses similarly can have very different abilities to approximate other eligible prediction rules. We also examine how these comparisons vary with sample size, language model, and prompt.

The second application studies differentiated-product demand, measuring restrictiveness separately for market shares and the price elasticities implied by fitted share functions. In simulated demand systems, richer structural specifications improve share fit without improving implied price elasticities. The code-execution agent and tabular foundation model are less restrictive than the logit, nested-logit, and mixed-logit specifications we consider on both dimensions. They also achieve higher completeness in predicting observed cereal-market shares. In that comparison, however, the AI procedures can use product-specific effects omitted from our structural specifications; the observed data do not directly reveal the price elasticities.

\subsection*{Related literature}

 \cite*{FKLM}'s    \emph{completeness} measures the fraction of the predictable variation in an outcome that is captured by a given model, relative to a flexible statistical benchmark. It is implemented using machine-learning methods to approximate the best possible prediction given observables, and has been applied to models of choice under risk and other domains. \citet*{fudenberg2023flexible} propose a  measure of the \emph{restrictiveness} of a model, which is also the object of interest in our current paper. \citet*{fudenberg2023flexible}  also propose  evaluating models by  comparing  their  restrictiveness together with their completeness, which produces an empirical Pareto frontier that trades off fit on real data against the regularities ruled out by the model.%
 \citet{EllisNeff} connect sample analogs of model restrictiveness to Rademacher and offset Rademacher complexity and derive a generalization-error bound expressed in terms of finite-sample restrictiveness. They point out the value of
economically meaningful eligible sets, which can distinguish
models with the same (and potentially degenerate) limiting complexity measures. 

\citet*{crawford2025demand} uses Gaussian process priors to represent demand functions satisfying the restrictions of rational choice in consumer settings, focusing on a binary (0-1) distance metric  \citep*{Selten} to evaluate whether observed demand can be rationalized \emph{exactly}, as opposed to our use of a continuous discrepancy.  Continuous discrepancy functions  help circumvent the degeneracy problem for conditional rationality that \citet*{crawford2025demand} encounters under continuous conditioning variables.

More broadly, our framework is related to the statistical learning theory literature on model complexity and capacity, where notions such as VC dimension, Rademacher complexity, and  metric entropy have been used to establish upper bounds on generalization errors under empirical risk minimization. Although conceptually related, our restrictiveness measure is an \emph{average-case} approximation measure, defined with respect to a context-specific user-chosen discrepancy function (under an evaluation distribution on an economically meaningful eligible set), and our restrictiveness measure produces an interpretable number in the unit interval rather than rate bounds. See, also, Section 3.4 of \cite*{fudenberg2023flexible} for related discussions and references. In this paper, we further stress that the choice of the discrepancy function should be treated as an important modeling decision for restrictiveness to be interpretable in a context-specific manner.

A recent and growing literature explores the use of LLMs and AI tools for economic prediction and research \citep{korinek2023generative}. Applications include simulating economic decisions \citep{horton2023large}, predicting experimental treatment effects \citep{ashokkumar2026experiments}, and forecasting inflation \citep{fariaeCastroLeibovici2024AIInflation} and stock returns from news \citep{lopezlira2026chatgpt}. LLMs can also organize research workflows: \citet{manning2024automated} develop a system that generates and tests social-scientific hypotheses using simulated agents. Beyond language models, tabular foundation models such as TabPFN \citep{hollmann2025tabpfn} and TabICL \citep{qu2026tabicl} learn from synthetic datasets to make predictions. These developments make AI prediction procedures increasingly relevant to economic applications.

Despite these developments, it remains  unclear   when and why AI procedures perform well in specific economic settings. There has been work that documents LLMs' promising predictive results, as well as failures to reproduce human decisions and economic preference heterogeneity \citep[see, e.g.,][]{gao2025take,chen2026preferences}. \citet{ludwig2025large} clarify the econometric requirements for using LLM predictions and measurements, including the roles of training-data leakage and validation. Testing on a benchmark set of tasks can help assess predictive performance of a procedure, but usually does not reveal how much structure a workflow imposes or how that structure changes with its model, prompt, and agentic design. We assess these restrictions by evaluating the same procedure on simulated data generated from eligible prediction rules, using a common scale for AI procedures and economic models.

~

The rest of the paper is organized as follows. Section~\ref{sec:framework} defines procedural restrictiveness,  relates it to model restrictiveness, and describes its computation and inference. Section~\ref{sec:apps} applies the framework to certainty equivalents, and Section~\ref{sec:app-dc} applies it to differentiated-product demand.

\section{Model Restrictiveness and Procedural Restrictiveness}\label{sec:framework}
Model restrictiveness measures the best approximation error within a specified model class, averaged over eligible prediction rules drawn from an evaluation distribution and normalized by a baseline. Procedural restrictiveness instead evaluates the fitted rules returned by a procedure using synthetic data generated from those eligible rules. We define both measures and show that procedural restrictiveness converges to model restrictiveness when the procedure and its baseline are risk-consistent for their respective model classes.
\subsection{Model Restrictiveness}\label{subsec:Setup}

We begin with the basic setup and review the notion of model restrictiveness in \citet*[hereafter FGL]{fudenberg2023flexible}, which is defined for an explicitly specified model class.\footnote{The notion of restrictiveness is essentially the same as in FGL, but here we do not require the covariate domain $\mathcal X$ to be finite as in FGL. Section~\ref{sec:implementation} discusses how to implement restrictiveness with continuous covariates.} 

Let $X$ be a \emph{covariate vector} with support $\mathcal X$ and marginal distribution $P_X$, and let $Y\in\mathcal Y$ be an \emph{outcome}. We assume that $\mathcal X$ is a compact subset of $\R^d$; $P_X$ may be chosen by the researcher, known a priori, or estimated from data. A \emph{prediction rule} is a function $f:\mathcal X\to\mathcal Y$. We denote the set of all such functions by $\ol{\mathcal F}\equiv\mathcal Y^{\mathcal X}$, which is assumed to be a metric space.

We take as a primitive a \emph{discrepancy} function $d:\ol{\mathcal{F}}\times\ol{\mathcal{F}}\rightarrow\mathbb{R}_{+}$
where $d(f,f')$ measures how different the two prediction rules $f$
and $f'$ are. For example, if $Y$ is a vector,
a natural choice for $d$ is the expected mean-squared distance between
the prediction rules (with respect to $P_{X}$), and if $Y$ is a distribution
a natural choice for $d$ is the expected KL-divergence (again with
respect to $P_{X}$). We allow for functions $d$ that are not distances
(such as KL-divergence), but require that $d(f,f)=0$. We also assume that $d$ is uniformly bounded, and that
$d(\cd,f)$ and $d(f,\cd)$ are continuous.

We will evaluate the restrictiveness of a specific model class $\mathcal{F}_{\Theta}:=\{f_{\theta}\}_{\theta\in\Theta}\subseteq\overline{\mathcal{F}}$,
where the prediction rules $f_{\theta}$ depend continuously on a
parameter $\theta$ from a parameter set $\Theta$, which can be finite
or infinite dimensional. Restrictiveness is defined relative to a
compact set of eligible prediction rules $\mathcal{F}\subseteq\ol{\mathcal{F}}$
that reflect any constraints the model is known to have. For example,
if a model is known to imply that choices respect first-order stochastic
dominance, we can define $\mathcal{F}$ to be all rules with this
property, and measure the model's additional restrictiveness beyond
this. In general, the eligible set $\mathcal{F}$ consists of all
prediction rules that satisfy user-specified background constraints,
where the special case of $\mathcal{F}=\ol{\mathcal{F}}$ corresponds
to the question of whether $\mathcal{F}_{\Theta}$ imposes any restrictions
at all.

Let $\lambda_{\mathcal{F}}$ denote a chosen evaluation distribution on $\mathcal{F}$. We define the restrictiveness of a model to be its expected discrepancy to the set $\cF_\Theta$, normalized with respect to the expected discrepancy of a baseline 
prediction rule $\cF_{\text{base}}$. This 
baseline  prediction rule is chosen to suit the setting, and we interpret its performance as a lower bound that any sensible model should outperform: for example, in some scenarios a natural baseline is the constant model $\cF_{\text{base}} = \{c:c\in\R\}$, while in others it may be a singleton set $\mathcal{F}_{\text {base }}=\left\{f_{\theta_0}\right\}$, where $f_{\theta_0}$ is the model in $\mathcal{F}_{\Theta}$ evaluated at baseline parameter $\theta_0$.

\begin{assumption}[Nondegeneracy]\label{assu:nondegen}
    $\mathbb{E}_{\lambda_{\mathcal{F}}}[d(\cF_{\text{base}},f)] >0$.
\end{assumption}

\begin{defn}[Model Restrictiveness]\label{def:rest}
The restrictiveness of model $\mathcal{F}_{\Theta}$ with respect
to eligible set $\mathcal{F}$ is 
\begin{equation}
r(\mathcal{F}_{\Theta};\mathcal{F},d)=\frac{\mathbb{E}_{\lambda_{\mathcal{F}}}[d(\mathcal{F}_{\Theta},f)]}{\mathbb{E}_{\lambda_{\mathcal{F}}}[d(\cF_{\text{base}},f)]}\label{eq:r}
\end{equation}
where $
d(\mathcal{F}_{\Theta},f):=\inf_{f_{\theta}\in\mathcal{F}_{\Theta}}d(f_{\theta},f)$.
\end{defn}

Note that equation \eqref{eq:r} implies that restrictiveness is invariant to positive multiplication  of the discrepancy function $d$ by the linearity of the expectation operator. 

\subsection{Procedural Restrictiveness}\label{subsec:procedures}

Model restrictiveness in Definition~\ref{def:rest} evaluates the best approximation available within a   specific model class. We now extend the concept to a \emph{prediction procedure}, which takes a dataset as input and returns a fitted prediction rule. This allows us to evaluate both estimators of economic models and black-box workflows whose available rules have no tractable explicit description. The object of evaluation is the procedure as implemented, including its ability to select and fit a prediction rule.

Consider the same setup in   Section~\ref{subsec:Setup} with covariate distribution $P_X$, eligible set $\cF$, evaluation distribution $\l_\cF$, and discrepancy $d$. We first formalize the meaning of a general prediction procedure. 

\begin{defn}[Prediction Procedure]\label{def:procedure}
A \emph{prediction procedure} is a sequence $\cA=(\cA_n)_{n\in\N}$ of measurable mappings
\[
\cA_n:\left(\cX\times\cY\right)^{n}\times\O_{\cA}\to\ol\cF,
\quad
(S_n,\o)\mapsto\hat f_n=\cA_n(S_n,\o),
\]
where $(\O_{\cA},P_{\o})$ is a probability space whose element $\o$ captures all internal randomness of the procedure and is assumed to be independent of the data. The procedure is \emph{deterministic} if $\cA_n$ does not depend on $\o$.
\end{defn}

For example, least squares in an economic model defines a prediction procedure once the fitting algorithm, optimization method, and stopping rule are specified. A machine-learning procedure also includes any hyperparameter selection and cross-validation. For an LLM workflow, the specification includes the model version, prompt, computation budget, failure protocol and other design choices. Everything affecting the returned rule is part of the procedure, so two estimators of the same model class, or two workflows differing only in their prompts, can have different restrictiveness at a given sample size.

For an eligible prediction rule $f\in\cF$, treated as the
\emph{pseudo-truth} generating the simulated observations below,
define the noise-free sample
\begin{equation}\label{eq:Sample}
    S_n(f) := \{(X_i,Y_i)\}_{i=1}^n,
    \quad
    X_i \stackrel{\text{i.i.d.}}{\sim} P_X,
    \quad
    Y_i = f(X_i),
\end{equation}
and write $\hat f_n:=\cA_n(S_n(f),\o)$ for the fitted rule the procedure returns on this sample. 

For each $n\ge1$, define the procedure's \emph{expected discrepancy} by
\begin{equation}
    L_n(\cA)
    :=
    \E_{f\sim\l_{\cF}}\,
    \E_{S_n(f),\,\o}
    \left[
        d(\hat f_n,f)
    \right].
    \label{eq:learning-curve}
\end{equation}
This averages the discrepancy between the returned rule and the pseudo-truth over pseudo-truth draws, training samples, and internal randomness. The maintained boundedness of $d$ makes $L_n(\cA)$ uniformly bounded in $n$. The noise-free design removes observation noise from the evaluation, while retaining the effects of limited training data, estimation, and internal optimization.

We normalize $L_n(\cA)$ by the expected discrepancy of a baseline procedure $\cA_{\mathrm{base}}$. The baseline classes in Section~\ref{subsec:Setup} provide natural choices: a singleton class induces a procedure that always returns its sole rule, while least squares over the class of constant rules returns the sample mean.

\begin{defn}[Procedural Restrictiveness]\label{def:proc-rest}
The \emph{procedural restrictiveness profile} of a procedure $\cA$ (relative to $\cF$, $\l_\cF$, $d$, and $\cA_{\mathrm{base}}$) is the sequence
\[
r_n(\cA):=\frac{L_n(\cA)}{L_n(\cA_{\mathrm{base}})},
\quad n\ge1,
\]
where the baseline is chosen so that $L_n(\cA_{\mathrm{base}})>0$ for every $n\ge1$.
\end{defn}

Procedural restrictiveness $r_n(\cA)$ measures the procedure's expected approximation error relative to the baseline after $n$ noise-free observations generated by an eligible prediction rule drawn from $\l_\cF$. Low values indicate that it approximates eligible prediction rules well on average. The measure evaluates the rule actually returned, so it reflects both restrictions on possible outputs and the procedure's ability to learn from the supplied sample. Reporting $r_n(\cA)$ at several sample sizes reveals how this ability changes with more observations, including whether approximation appears to plateau at sample sizes relevant to the application.

By evaluating returned rules without requiring an explicit model class, procedural restrictiveness can be applied to black-box prediction procedures. An LLM, for example, may receive a training sample in its prompt and return predictions or code implementing a prediction rule. A tabular foundation model (TFM) instead receives the sample in tabular form and uses it as context to predict at new inputs. TabPFN and TabICL are examples pretrained on synthetic datasets that make predictions with fixed weights \citep{hollmann2025tabpfn,qu2026tabicl}. The same approach also applies to AI workflows that combine such models with an \emph{agentic harness} that controls tool use and the sequence of model calls. An LLM agent may write and execute code, compare specifications, call a TFM, and revise its predictor before returning it. Treating the complete workflow as the prediction procedure allows us to compare how alternative prompts, tools, and computation budgets affect its ability to approximate eligible prediction rules. With a common evaluation design and baseline, these workflows and estimators of economic models can be evaluated on the same scale.

\subsection{Relationship between Model and Procedural Restrictiveness}\label{subsec:learning-curve}

We now show that the two notions of restrictiveness above, model restrictiveness and procedural restrictiveness, agree asymptotically when the prediction procedure consistently attains the best approximation available within a specified model class and the baseline procedure satisfies the same condition for its class. In this sense, procedural restrictiveness in Definition~\ref{def:proc-rest} generalizes model restrictiveness in Definition~\ref{def:rest}.

To formalize this connection, fix a model class $\cF_\T$. An estimator $\hat\t_n(S_n,\o)\in\T$ induces the prediction procedure $\cA_n(S_n,\o)=f_{\hat\t_n(S_n,\o)}$.  We require the following risk consistency condition with respect to $d$.

\begin{assumption}[Risk-Consistency]\label{assu:risk-consistency}
The procedure $\cA$ is risk-consistent for the model class $\cF_\T$ (relative to $\lambda_{\cF}$ and $d$): for $\lambda_{\cF}$-almost every $f\in\cF$,
\begin{equation}
    \lim_{n\to\infty}
    \E_{S_n(f),\,\o}
    \left[
        d(\hat f_n,f)
    \right]
    =
    \inf_{f_\theta\in\cF_\T} d(f_\theta,f)
    =
    d(\cF_\T,f).
    \label{eq:risk-consistency}
\end{equation}
\end{assumption}

For discrepancies $d$ given by expected pointwise losses, the risk consistency condition above coincides with the standard notion of risk consistency in statistical learning, which follows from a uniform law of large numbers and convergence of the fitting algorithm to the empirical minimum. Under squared-error discrepancy, this covers least squares over compact parametric classes with continuous predictions, subject to suitable boundedness conditions. See also \citet[Lemma~4]{EllisNeff} for a formal result that establishes this risk-consistency condition.

\begin{proposition}[Reconciliation with Model-Class Restrictiveness]
\label{prop:learning-curve}
Suppose that $\cA$ is risk-consistent for $\cF_\T$ and $\cA_{\mathrm{base}}$ is risk-consistent for $\cF_{\mathrm{base}}$. Then
\[
    \lim_{n\to\infty} L_n(\cA)
    =
    \E_{\l_{\cF}}\left[d(\cF_\T,f)\right],
\]
and the limit restrictiveness $r(\cA):=\lim_{n\to\infty}r_n(\cA)$ exists, with
\[
    r(\cA)
    =
    \frac{\displaystyle\lim_{n\to\infty} L_n(\cA)}
         {\displaystyle\lim_{n\to\infty} L_n(\cA_{\mathrm{base}})}
    =
    \frac{\E_{\l_{\cF}}\left[d(\cF_\T,f)\right]}{\E_{\l_{\cF}}\left[d(\cF_{\mathrm{base}},f)\right]}
    =
    r(\cF_\T;\cF,d).
\]
\end{proposition}

Proposition~\ref{prop:learning-curve} shows that, given a specific model class, a risk-consistent estimator, and a  baseline, procedural restrictiveness coincides with model restrictiveness in the limit. Meanwhile, procedural restrictiveness  remains meaningful when the model class is not explicitly specified or risk-consistency is unavailable. Hence, procedural restrictiveness can be viewed as a generalization of model restrictiveness.

When $d(h,f)=\E_{P_X}[\ell(h(X),f(X))]$ for a nonnegative pointwise loss $\ell$, the discrepancy is the prediction risk against noise-free labels. In this case, $L_n(\cA)$ in \eqref{eq:learning-curve} can be viewed as a  noise-free version of the average-case learning curve, and Proposition~\ref{prop:learning-curve} identifies its normalized limit with model restrictiveness under risk-consistency. This provides a learning-curve interpretation of procedural restrictiveness without requiring an expected-loss representation for its definition. Online Appendix~\ref{sec:oa-gp-learning-curve} compares this special case with noisy learning curves in the Gaussian-process literature.

For LLMs, TFMs, and more complex AI workflows, risk-consistency guarantees are generally unavailable. Even without risk-consistency, a procedure's possible outputs define a \emph{reachable class}. Under the baseline conditions below, the restrictiveness of this class provides a lower bound on procedural restrictiveness, whether or not the class has a tractable explicit description.

\begin{defn}[Reachable Class]\label{def:reachable}
The \emph{reachable class} of a procedure $\cA$ is
\[
\cF_{\cA}
:=\bigcup_{n\in\N}\cA_n\Big(\left(\cX\times\cY\right)^{n}\times\O_{\cA}\Big)
\subseteq\ol\cF,
\]
the set of prediction rules that $\cA$ can output at any sample size.
\end{defn}

\begin{proposition}[Reachable-Class Lower Bound]\label{prop:reachable}$\quad$
\begin{enumerate} 
\item[(i)] For every $n\ge1$,
$L_n(\cA)\;\ge\;\E_{\l_\cF}\big[d(\cF_{\cA},f)\big]$.
\item[(ii)] If $\cA_{\mathrm{base}}$ is risk-consistent for $\cF_{\mathrm{base}}$, then 
$\liminf_{n\to\infty}r_n(\cA)\;\ge\; r\left(\cF_{\cA};\cF,d\right)$.
\item[(iii)] Under the singleton baseline, i.e.,  $\cF_{\mathrm{base}}=\{f_{\mathrm{base}}\}$ and $\cA_{\mathrm{base},n} \equiv f_{\mathrm{base}}$, then $r_n(\cA)\;\ge\;r\left(\cF_{\cA};\cF,d\right)$ for every $n\ge1$.
\end{enumerate}
\end{proposition}

Proposition~\ref{prop:reachable} relates procedural restrictiveness to the restrictions implicit in a black-box procedure without requiring the procedure to be risk-consistent. With a singleton baseline, low procedural restrictiveness implies low restrictiveness of its reachable class: the procedure's outputs already approximate eligible prediction rules well on average relative to the baseline. High procedural restrictiveness is less conclusive, since it may reflect limitations of learning or search as well as restrictions on possible outputs. With an estimated, risk-consistent baseline, part (ii) gives the corresponding bound for the lower limit of procedural restrictiveness.

\subsection{Computation and Inference}\label{sec:implementation}\label{subsec:Implement}\label{subsec:proc-inference}

 This subsection briefly discusses computation and inference for procedural restrictiveness. The required primitives are the covariate and outcome spaces $(\cX,\cY)$, the covariate distribution $P_X$, the eligible set $\cF$ and its evaluation distribution $\l_\cF$, the discrepancy $d$, the prediction procedure $\cA$, and the baseline procedure $\cA_{\mathrm{base}}$. We evaluate the procedure at a given training sample size $n$.

\paragraph{Evaluation Distributions}  The choice of $\l_\cF$ is domain-dependent. When $\cX$ is finite, as in FGL, rules with finite-dimensional outcomes can be represented by finite vectors, making $\l_\cF$ relatively straightforward to specify. With continuous $\cX$, however, eligible rules $\cF$ may form an infinite-dimensional function space. Bayesian nonparametric priors provide a way to specify and sample distributions over such spaces. 
Gaussian processes are a natural choice for regression settings, where a rule maps covariates to a real-valued prediction. Their mean and covariance functions control the level, variation, and smoothness of the sampled functions \citep{williams2006gaussian}. For prediction rules that assign probability distributions, Dirichlet-process priors and their mixture extensions provide alternatives \citep{Ferguson1973,quintana2022dependent}. Gaussian processes can also be combined with Dirichlet-process mixing to accommodate heterogeneous regression functions \citep{rasmussen2001infinite}. More general generative priors, such as those used to construct synthetic pretraining tasks for TFMs, could likewise be adapted to generate eligible prediction rules \citep{hollmann2025tabpfn,qu2026tabicl}. Each construction must induce a distribution over rules in $\cF$, from which we generate the noise-free observations.

In our applications, we use Gaussian processes to construct certainty-equivalent functions and the utility components underlying demand systems, incorporating economic shape restrictions through constraints or transformations. Appendix~\ref{sec:gp.details} discusses Gaussian-process specifications and methods for imposing these restrictions, with application-specific details given in Online Appendix Sections~\ref{subsec:evaluation.example1} and~\ref{subsec:evaluation.example2}.

\paragraph{Estimation and Inference} Once the evaluation distribution is specified, we estimate procedural restrictiveness by averaging discrepancies across simulation replications and use the variation across replications to quantify Monte Carlo uncertainty.

Fix a training sample size $n$ and draw $f_1,\dots,f_M\overset{\text{i.i.d.}}{\sim}\l_\cF$. For each replication, independently generate the noise-free sample $S_n(f_m)$ and run the procedure and baseline with their respective internal randomness to obtain
\[
D_m:=d\big(\cA_n(S_n(f_m),\o_m),f_m\big),
\quad
B_m:=d\big(\cA_{\mathrm{base},n}(S_n(f_m),\o_m'),f_m\big).
\]
For a singleton baseline, $B_m=d(f_{\mathrm{base}},f_m)$. When discrepancies are evaluated numerically, the fitted rules are scored on a separate sample of inputs from $P_X$.

Writing $\bar D_M$ and $\bar B_M$ for the sample means, the estimator and its Monte Carlo standard error are
\[
\widehat r_{n,M}(\cA):=\frac{\bar D_M}{\bar B_M},
\quad
\widehat{\operatorname{se}}\big(\widehat r_{n,M}(\cA)\big)
:=\frac{\widehat\s_n}{\sqrt M},
\]
where
\[
\widehat\s_n^2
:=\frac{1}{(M-1)\bar B_M^2}
\sum_{m=1}^M\left(D_m-\widehat r_{n,M}(\cA)B_m\right)^2.
\]
At fixed $n$, asymptotic normality follows from the standard central limit theorem and delta method as $M$ grows: see Proposition~\ref{prop:proc-clt} in Appendix~\ref{sec:proofs} for details. The standard error accounts for variation across pseudo-truths, training samples, and internal randomness.

For model restrictiveness, the same formulas apply with $D_m=d(\cF_\T,f_m)$ and $B_m=d(\cF_{\mathrm{base}},f_m)$. In our applications, we evaluate procedures on common pseudo-truths and training samples, allowing paired comparisons that account for the covariance of their discrepancies.

\section{Application 1: Certainty Equivalents}\label{sec:apps}

\begingroup
\makeatletter
\@ifundefined{FloatBarrier}{%
  \gdef\FloatBarrier{%
    \par\begingroup\let\@elt\relax
    \edef\@tempa{\@botlist\@deferlist\@dbldeferlist}%
    \ifx\@tempa\@empty\else
      \ifx\@fltovf\relax
        \if@firstcolumn\clearpage\else\null\newpage\FloatBarrier\fi
      \else
        \newpage\let\@fltovf\relax\FloatBarrier
      \fi
    \fi
    \endgroup\suppressfloats[t]%
  }%
}{}
\endgroup
\newcounter{empiricalpanel}[figure]
\renewcommand{\theempiricalpanel}{\thefigure\alph{empiricalpanel}}

We compare economic models of risk preferences with LLM workflows, TFMs, and other statistical learners in predicting certainty equivalents, the sure amounts valued the same as lotteries. We measure restrictiveness on simulated certainty-equivalent mappings and completeness \citep{FKLM} on experimental lottery valuations from \citet{Bruhin}.  An LLM completion and a code-execution agent predict observed valuations almost equally well, but the agent is substantially less restrictive. A two-parameter specification of cumulative prospect theory attains nearly the same completeness as the full model while imposing much stronger restrictions.
\subsection{Setting, Evaluation Distribution, and Discrepancy}
\label{subsec:app1-setting}

In this application, we predict the certainty equivalent of a binary lottery from its configurations (prize amounts and probabilities). Specifically, each lottery is characterized by a tuple $x=(\bar{z}, \underline{z}, p)$, where $\bar{z}\geq\underline{z} \geq 0$ are the possible prizes, and $p$ is the probability of the larger prize. We construct the lottery domain $\mathcal D_B$ by allowing
$\bar z$, $\underline z$, and $p$ to vary continuously over their
respective observed ranges in the experimental data of \citet{Bruhin}, with prizes
normalized by the maximum observed prize:
\begin{equation}
\mathcal D_B=\left\{(\bar z,\underline z,p):
\bar z\in[1/15,1],\ \underline z\in[0,1/3],\
p\in[0.05,0.95],\ \bar z\geq\underline z\right\}.
\label{eq:cpt-rectangular-support}
\end{equation}
A prediction rule assigns each $x\in\mathcal D_B$
a certainty equivalent in $[\underline z,\bar z]$.

The eligible set $\mathcal F$ on $\mathcal D_B$ imposes two economic restrictions.
First, the certainty equivalent must lie between the lottery's lowest and highest prizes:
$\underline z \leq f(\bar z,\underline z,p) \leq \bar z$.
Second, the certainty equivalent cannot decrease when either prize or the probability of the higher prize increases:
$f(x_1)\geq f(x_2)$ whenever $\bar z_1\geq\bar z_2$,
$\underline z_1\geq\underline z_2$, and $p_1\geq p_2$.\footnote{The importance of these restrictions is also discussed in \citet{EllisNeff}, which uses  the 
monotonicity of eligible functions and economic model classes to  establish a bound on sample-to-population discrepancy gaps.} 

We construct an evaluation distribution over certainty-equivalent functions using a constrained Gaussian process that imposes these restrictions on the sampled lotteries. 
For each of $M=100$ replications and each $n\in\{50,100,200\}$, we draw $n$ training lotteries and 1,000 evaluation lotteries uniformly over $\mathcal D_B$, and generate their certainty equivalents from this distribution. Each procedure observes only the training lotteries and their certainty equivalents. 

We measure discrepancy by the mean squared difference between predicted and pseudo-true certainty equivalents over the evaluation lotteries:
\[
d(f,g)
=
\E_{X\sim P_X}\left[(f(X)-g(X))^2\right],
\]
where $P_X$ is the uniform distribution over $\mathcal D_B$.
We approximate this discrepancy by the MSE over the 1,000 evaluation lotteries in each replication.
For both restrictiveness and completeness, the baseline model is the expected-value rule $f_{\mathrm{EV}}(x)=p\bar z+(1-p)\underline z$.
Restrictiveness is estimated as the ratio of the average MSE across replications to the corresponding average for the expected-value rule.
All models and procedures use common draws within each sample size so that we can make paired comparisons.
Online Appendix Section~\ref{subsec:evaluation.example1}
provides details on the evaluation distribution and
numerical construction.

To measure completeness, we use the 4,459 certainty equivalents reported by 179 participants for the 25 lotteries in the experimental data of \cite{Bruhin}.\footnote{16 of the \(179 \times 25 = 4{,}475\) possible participant–lottery observations are missing from the data.} We use ten cross-validation folds common to all procedures. Each procedure receives the 25 lottery means calculated from the training responses and predicts held-out individual responses. We evaluate predictions by their mean squared error on these held-out responses. TabPFN, a TFM, fitted to the same training lottery means, serves as the flexible benchmark. Completeness measures each procedure's reduction in mean squared prediction error relative to expected value as a fraction of the reduction achieved by TabPFN. Since the experimental data is public, language models' completeness may reflect prior exposure to the data as well as estimation from the supplied observations.

\subsection{Economic Models and Procedures}\label{subsec:cpt-procedures}

Rank-dependent utility (RDU) \citep{Quiggin} combines utility over prizes with probability weighting. Outcomes are ranked by their prizes, and decision weights are obtained from their cumulative probabilities. For the binary lotteries considered here, the certainty-equivalent rule is
\begin{equation}
f_{u,w}(\bar z,\underline z,p)
=u^{-1}\!\left(w(p)u(\bar z)+[1-w(p)]u(\underline z)\right),
\label{eq:app1-rdu}
\end{equation}
where $u$ is a strictly increasing utility function and $w$ is an increasing probability-weighting function with $w(0)=0$ and $w(1)=1$. Expected utility sets $w(p)=p$; the expected-value baseline additionally sets $u(z)=z$.

We evaluate cumulative prospect theory (CPT) and its submodels, which are parametric special cases of RDU. 
Our CPT specification uses
\[
    u(z)=z^\alpha,\quad
    w(p)=\frac{\delta p^\gamma}
    {\delta p^\gamma+(1-p)^\gamma},
    \quad
    (\alpha,\gamma,\delta)\in[0.01,1]^2\times[0.01,30].
\]
The parameter $\alpha$ governs utility curvature, $\gamma$ governs the shape of probability weighting, and $\delta$ captures optimism or pessimism by shifting the weight assigned to the larger prize.
The submodels fix one or more parameters at the expected-value baseline values $(\alpha,\gamma,\delta)=(1,1,1)$; for example, $\mathrm{CPT}(\alpha,\gamma)$ fixes $\delta=1$.
For binary lotteries, the submodels $\mathrm{CPT}(\alpha,\delta)$ and $\mathrm{CPT}(\delta)$ are equivalent to the disappointment-aversion (DA) specifications of \citet{Gul1991} under a reparameterization.\footnote{Online Appendix Section~\ref{subsec:procedures.example1} provides the DA specifications and their parameter correspondence with CPT.}

We compare these economic models with several statistical and AI-based procedures. Gaussian-process regression uses a Mat\'ern-$3/2$ covariance function with fixed hyperparameters, with predictions clipped to the prize interval. Two tabular foundation models, TabPFN v2 \citep{hollmann2025tabpfn} and TabICL v2 \citep{qu2026tabicl}, use in-context learning with fixed pretrained weights. AutoML selects machine-learning models and tunes their hyperparameters automatically. 

We also evaluate three LLM-based procedures: AIDE, Completion, and Agent. AIDE uses an LLM to generate and select machine-learning scripts through tree search based on cross-validated fit. For Completion and Agent, our main comparisons use Claude Sonnet 5.5 with the common neutral task in the following box, which asks for a prediction program without prescribing a model or estimation method. We examine alternative language model versions and prompt designs in Section~\ref{subsec:cpt-prompts}.

\begin{recipebox}{Neutral Task for Completion and Agent}
\label{box: neutral.task}
\small
This is an economic prediction problem concerning choice under risk.

The supplied dataset contains binary lotteries and their associated certainty equivalents. Each row has columns z\_hi, z\_lo, p, and ce. A lottery pays z\_hi with probability p and z\_lo with probability $1 - p$, where $0 \leq$ z\_lo $\leq$ z\_hi $\leq 1$ and $0 \leq$ p $\leq 1$. The certainty equivalent ce is the sure amount valued the same as the lottery. Prizes and certainty equivalents are expressed in normalized units.

Using the supplied data, develop the best prediction algorithm you can for certainty equivalents, and return Python code implementing predict\_ce(z\_hi, z\_lo, p), which returns a finite float throughout the stated input domain. Include any supporting code, fitting steps, or retained data needed to execute the prediction function. Use only the supplied data for fitting and model selection. Once constructed, the predictor must produce reproducible outputs independently of previous prediction queries.

Return one fenced Python code block containing the prediction function and its supporting code. Use only numpy, math, and the Python standard library as imported dependencies. The returned code must be self-contained and perform no file or network I/O. Supplied data and fitted state may be embedded in the code.
\end{recipebox}

Completion generates its prediction program without executing code, with thinking and effort set to their lowest available settings.\footnote{Sonnet 5.5 does not allow its thinking mode to be set as \texttt{disabled}. We set its thinking mode to be \texttt{between\_tools}, its lowest thinking setting, with \texttt{low} effort. This suppresses up-front thinking while allowing thinking blocks between tool calls. Since Completion has no tool access, this setting effectively turns off thinking.} Agent uses high effort and can execute Python on the supplied data for up to 20 rounds before returning its program. Either program may also estimate parameters from the training observations when used for prediction. Predictions are constrained to lie between the lottery's prizes, and failed fits or invalid predictions are replaced by the expected-value rule. Online Appendix~\ref{sec:detail.example1} provides the implementation details and failure counts.

\subsection{Main Results}\label{subsec:cpt-blackbox}\label{subsec:cpt-econ}

Table~\ref{tab:app1-profiles} reports restrictiveness profiles across training sample size $n$. Figure~\ref{fig:app1-plane} summarizes restrictiveness at $n=100$ and completeness in predicting participants' lottery responses. Online Appendix Table~\ref{tab:app1-main} provides the estimates and standard errors.

\begin{table}[t!]
  \caption{Restrictiveness profiles over $n$}
  \label{tab:app1-profiles}
  \centering
  \setlength{\tabcolsep}{6pt}%
  \begin{tabular}{@{}lccc@{}}\hline
   & $n=50$ & $n=100$ & $n=200$ \\ \hline
  CPT$(\alpha,\gamma,\delta)$ & 0.127 & 0.123 & 0.123 \\
  CPT$(\alpha,\gamma)$ & 0.586 & 0.584 & 0.582 \\
  GP regression & 0.021 & 0.017 & 0.014 \\
  TabPFN v2 & 0.018 & 0.014 & 0.012 \\
  TabICL v2 & 0.020 & 0.015 & 0.012 \\
  AutoML (FLAML) & 0.132 & 0.069 & 0.039 \\
  AIDE & 0.075 & 0.044 & 0.057 \\
  Completion & 0.167 & 0.328 & 0.389 \\
  Agent & 0.073 & 0.025 & 0.016 \\ \hline
  \end{tabular}
  \par\smallskip
{\footnotesize\raggedright\emph{Notes:} $M=100$ draws per column. Procedures are paired within columns; designs are separately realized across sample sizes. Monte Carlo standard errors range from $0.002$ to $0.06$.\par}
\end{table}

\begin{figure}[t!]
  \vspace*{8pt}
  \caption{Restrictiveness and completeness ($n=100$)}
  \label{fig:app1-plane}

  {\centering

  \begin{tikzpicture}
  \begin{axis}[
      width=14cm, height=9.5cm,
      xlabel={Completeness},
      ylabel={Restrictiveness},
      ymin=0, ymax=1.0, xmin=0, xmax=1.25,
      xtick={0,0.2,0.4,0.6,0.8,1.0},
      ytick={0,0.2,0.4,0.6,0.8,1.0},
      grid=major, grid style={dashed, gray!30},
      tick align=outside, axis lines=left, clip=false,
      every axis plot/.append style={thick},
  ]
  
  \addplot[only marks, mark=*, mark size=2.5pt, color=black]
    coordinates {(0.976288,0.123169) (0.950833,0.148310) (0.265532,0.253237) (0.265532,0.269215) (0.976414,0.584265) (0.685161,0.737135) (0.248542,0.851943)};
  \addplot[only marks, mark=diamond*, mark size=3.5pt, color=gray]
    coordinates {(1.000000,0.014421) (1.001890,0.014675) (0.922470,0.068997)};
  \addplot[only marks, mark=*, mark size=3pt, color=black, mark options={fill=white}]
    coordinates {(1.001612,0.016713)};
  \addplot[only marks, mark=triangle*, mark size=3.5pt, color=black]
    coordinates {(0.741567,0.044062)};
  \addplot[only marks, mark=star, mark size=4.5pt, color=black]
    coordinates {(0.988496,0.328211) (0.988619,0.024900)};
  \node[anchor=west, font=\small] at (axis cs:1.025,0.18) {CPT$(\alpha,\gamma,\delta)$};
  \draw[gray, thin] (axis cs:1.02,0.18) -- (axis cs:0.98,0.124);
  \node[anchor=east, font=\small] at (axis cs:0.925,0.205) {CPT$(\gamma,\delta)$};
  \draw[gray, thin] (axis cs:0.93,0.205) -- (axis cs:0.949,0.151);
  \node[anchor=east, font=\small] at (axis cs:0.948,0.584) {CPT$(\alpha,\gamma)$};
  \node[anchor=west, font=\small] at (axis cs:0.320,0.325) {CPT$(\delta)$};
  \draw[gray, thin] (axis cs:0.315,0.325) -- (axis cs:0.273,0.274);
  \node[anchor=west, font=\small] at (axis cs:0.320,0.210) {CPT$(\alpha,\delta)$};
  \draw[gray, thin] (axis cs:0.315,0.210) -- (axis cs:0.273,0.249);
  \node[anchor=east, font=\small] at (axis cs:0.655,0.737) {CPT$(\gamma)$};
  \node[anchor=west, font=\small] at (axis cs:0.269,0.88) {CPT$(\alpha)$};
  \node[anchor=east, font=\small] at (axis cs:0.89,0.105) {AutoML};
  \draw[gray, thin] (axis cs:0.895,0.105) -- (axis cs:0.921,0.070);
  \node[anchor=east, font=\small] at (axis cs:0.875,0.035) {GP};
  \draw[gray, thin] (axis cs:0.880,0.035) -- (axis cs:1.001,0.017);
  \node[anchor=west, font=\small] at (axis cs:1.040,0.025) {TabPFN, TabICL};
  \draw[gray, thin] (axis cs:1.035,0.025) -- (axis cs:1.002,0.014);
  \node[anchor=east, font=\small] at (axis cs:0.705,0.044) {AIDE};
  \node[anchor=west, font=\small] at (axis cs:1.034,0.328) {Completion};
  \draw[gray, thin] (axis cs:1.032,0.328) -- (axis cs:0.994,0.328);
  \node[anchor=west, font=\small] at (axis cs:1.055,0.080) {Agent};
  \draw[gray, thin] (axis cs:1.05,0.080) -- (axis cs:0.993,0.026);
  \end{axis}
  \end{tikzpicture}
  \par}
  \smallskip

  {\footnotesize \emph{Notes}: Restrictiveness at $n=100$ (vertical)
  against completeness (horizontal), from Online Appendix
  Table~\ref{tab:app1-main}. Completeness uses TabPFN v2 as the
  flexible benchmark. Filled circles: CPT classes (DA specifications
  coincide with their CPT counterparts); gray diamonds: TabPFN,
  TabICL, and AutoML; open circle: GP; filled triangle: AIDE;
  stars: Completion and Agent using Claude Sonnet 5.5 with the
  shared neutral task.\par}
\end{figure}

Within CPT, the full model and $\mathrm{CPT}(\alpha,\gamma)$ attain nearly identical completeness but differ substantially in restrictiveness:  At $n=100$, fixing $\delta$ at one raises restrictiveness from $0.123$ to $0.584$.  The probability-weighting parameters $\gamma$ and $\delta$ each contribute more flexibility than the utility-curvature parameter $\alpha$, consistent with the ordering found by FGL. Note also that as in FGL, CPT($\gamma)$ is both more complete and more restrictive than CPT($\alpha, \delta)$:   greater flexibility need not yield better fit, and the number of parameters does not determine relative restrictiveness when the models in question are not nested.\footnote{FGL implement restrictiveness on a finite lottery domain using incremental uniform draws so that monotonicity constraints are satisfied. Here we instead consider a continuous lottery domain and use a Mat\'ern-$3/2$ Gaussian process with monotonicity constraints imposed at the sampled points. The key qualitative CPT findings are nevertheless robust across these two evaluation designs.}

Completion and Agent  attain nearly identical completeness ($0.988$ and $0.989$), but differ substantially in restrictiveness. At $n=100$, Completion's restrictiveness is $0.328$, above full CPT's $0.123$, whereas Agent's is $0.025$, or $0.098$ below full CPT (paired s.e.\ $0.015$). This ordering holds at all three sample sizes.\footnote{Appendix~C  shows that at $n=100$ the greater contributions of $\gamma$ and $\delta$ relative to $\alpha$ and the restrictiveness ordering of Completion, full CPT, and Agent are all robust to two alternate specifications of the evaluation distribution: bounded reweighting of the GP-based distribution and adaptation of the tabular foundation model prior. }
With adaptive thinking at high effort, Completion's restrictiveness falls to $0.087$ at $n=100$, while its completeness remains $0.988$. Online Appendix~\ref{subsec:procedures.example1} reports the full results.

TFMs and GP regression are less restrictive than the fitted economic models and attain similar completeness. Their restrictiveness is $0.014$--$0.017$ at $n=100$, falling to $0.012$--$0.014$ at $n=200$. 
AutoML and AIDE are less restrictive than full CPT at $n=100$ ($0.069$ and $0.044$, respectively), but attain lower completeness ($0.922$ and $0.742$, respectively): once again greater flexibility did not translate into better empirical fit. Under the configurations and computation budgets considered here, these procedures offer no predictive advantage over the Sonnet 5.5 Agent, TFMs, or the best-performing CPT specifications.

The Sonnet 5.5  Agent  attains  low restrictiveness along with high completeness, and lies close to the TFMs on both dimensions. The two-parameter economic model $\mathrm{CPT}(\alpha,\gamma)$ occupies a distinctive position near the upper-right corner of Figure~\ref{fig:app1-plane}: it achieves completeness close to the main LLM procedures and TFMs, while remaining substantially more restrictive.

\subsection{Specification Choices by LLMs}\label{subsec:app1-specifications}

To better understand the differences in restrictiveness between Completion and Agent, we examine the specifications each procedure selects to make predictions. Figure~\ref{fig:app1-neutral-model-choices} reports these choices for synthetic data and observed data.

\begin{figure}[t!]
\centering
\setstretch{1}
\caption{Selected prediction specifications in the lottery application}
\label{fig:app1-neutral-model-choices}
\begin{minipage}[t]{0.48\textwidth}
  \centering
  \refstepcounter{empiricalpanel}
  \label{fig:app1-neutral-model-choices-restrictiveness-pure_completion}
  \captionsetup{font=small,position=top,skip=6pt}
  \caption*{(\alph{empiricalpanel}) Completion: restrictiveness}
  \includegraphics[width=\linewidth]{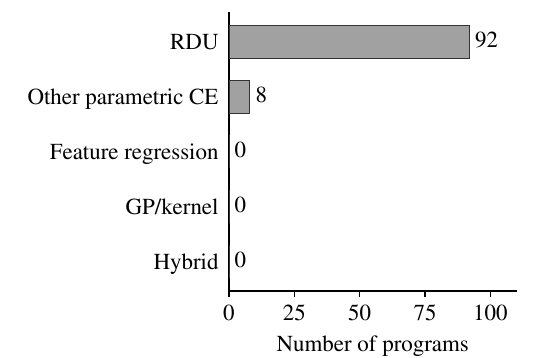}
\end{minipage}\hfill
\begin{minipage}[t]{0.48\textwidth}
  \centering
  \refstepcounter{empiricalpanel}
  \label{fig:app1-neutral-model-choices-completeness-pure_completion}
  \captionsetup{font=small,position=top,skip=6pt}
  \caption*{(\alph{empiricalpanel}) Completion: completeness}
  \includegraphics[width=\linewidth]{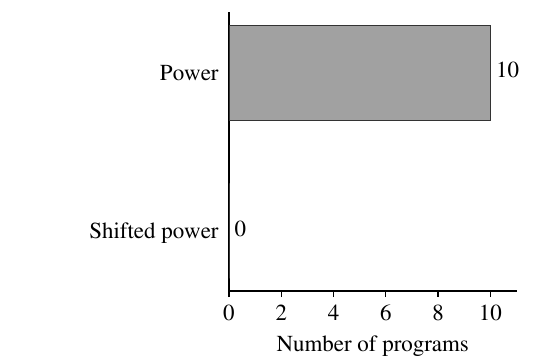}
\end{minipage}
\par\smallskip
\begin{minipage}[t]{0.48\textwidth}
  \centering
  \refstepcounter{empiricalpanel}
  \label{fig:app1-neutral-model-choices-restrictiveness-agent}
  \captionsetup{font=small,position=top,skip=6pt}
  \caption*{(\alph{empiricalpanel}) Agent: restrictiveness}
  \includegraphics[width=\linewidth]{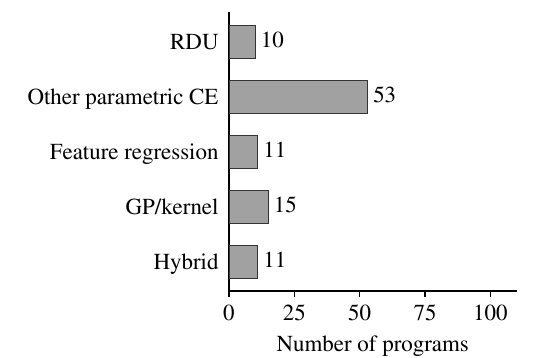}
\end{minipage}\hfill
\begin{minipage}[t]{0.48\textwidth}
  \centering
  \refstepcounter{empiricalpanel}
  \label{fig:app1-neutral-model-choices-completeness-agent}
  \captionsetup{font=small,position=top,skip=6pt}
  \caption*{(\alph{empiricalpanel}) Agent: completeness}
  \includegraphics[width=\linewidth]{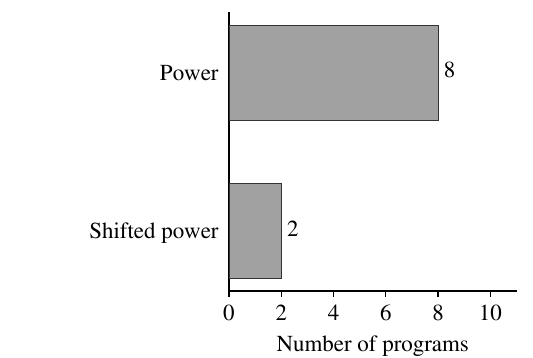}
\end{minipage}
\par\smallskip

{\footnotesize\raggedright\emph{Notes:} Left panels: 100 fits to simulated data per procedure at $n=100$. Right panels: utility specifications in ten fits to experimental lottery responses per procedure, all RDU. Both procedures use Claude Sonnet 5.5 and the prediction task described in Section~\ref{subsec:cpt-procedures}. Definitions of the specification categories are in Online Appendix Section~\ref{subsec:specifications.example1}.\par}
\end{figure}

Panels (a) and (c) suggest that Completion's higher restrictiveness accompanies greater reliance on RDU specifications when fitting simulated certainty equivalents. At $n=100$, RDU accounts for $92\%$ of Completion fits, compared with $10\%$ for Agent. Agent more often fits certainty equivalents directly to prizes and probabilities, using regressions, kernel methods, or combinations of these approaches. Its most frequent category, ``Other parametric CE,'' includes formulas that gain flexibility by relaxing some RDU restrictions. For example, some of the alternative parametric formulas
allow probability weighting to depend on prize amounts as well as
probabilities, whereas RDU imposes a common weighting function $w(p)$ across lotteries. Agent therefore relies less often on the restrictions imposed by a common utility function and
probability-weighting rule.

On the empirical data, both procedures select RDU specifications in all ten folds, as shown in panels (b) and (d). Completion always uses power utility, while Agent uses power utility in eight folds and shifted-power utility in two. Both use least-squares estimation and attain similar completeness with few parameters. A procedure that accommodates diverse simulated certainty-equivalent functions can therefore select a parsimonious economic specification for the observed data.

\FloatBarrier
\subsection{LLM Versions and Prompt Design}\label{subsec:cpt-prompts}

The language model and the prompt are components of a procedure, and both can affect its restrictiveness. We consider these effects separately, holding the prompt fixed
when comparing model versions and Claude Sonnet 5.5 fixed when comparing prompts. All comparisons use the same 100 simulated samples of size $n=100$ to measure restrictiveness and the same ten folds of experimental responses to measure completeness.
\par

\paragraph{Model Versions} We first compare Completion and Agent across four Anthropic models (the flagship Claude Opus 5.5, the mid-tier Sonnet 5.5 and its predecessor Sonnet 5, and the smaller, faster Haiku 5.5) and two OpenAI models (GPT-6 Luna and the more capable GPT-6 Sol). Here Completion uses each model’s lowest available reasoning and effort setting,\footnote{Opus 5.5 only accepts ``adaptive thinking'' mode.} while Agent uses each provider’s high-effort setting. The task, supplied data, and maximum computation budgets are the same across model versions.

\begin{figure}[t!]
\centering
\caption{Model versions: restrictiveness and completeness ($n=100$)}
\label{fig:app1-model-versions-plane}
\includegraphics[width=\textwidth]{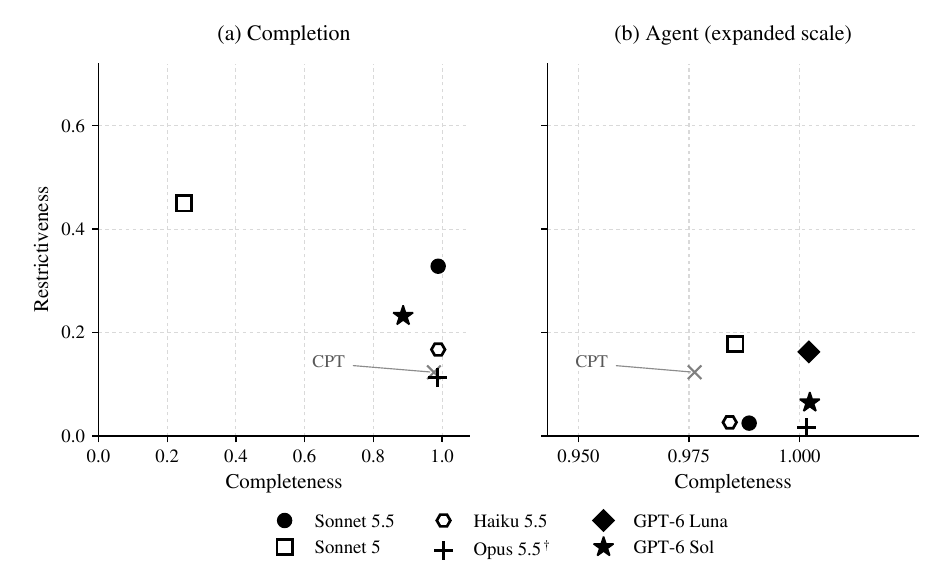}
{\footnotesize\raggedright\emph{Notes:} Both panels use the same prompt, pseudo-truths, training samples, and empirical cross-validation folds. The Agent panel expands the completeness scale. Crosses show full CPT (Online Appendix Table~\ref{tab:app1-main}). $^\dagger$Completion with Opus 5.5 uses adaptive thinking at low effort. Completion with GPT-6 Luna is omitted (completeness $-1.386$, restrictiveness $0.379$). Online Appendix~\ref{subsec:app1-version-comparison} reports the estimates and their standard errors.\par}
\end{figure}

Figure~\ref{fig:app1-model-versions-plane} presents the results
for Completion in panel (a) and Agent in panel (b). Agent is less restrictive than Completion for every language model considered. With Opus 5.5, Completion is close to full CPT in restrictiveness, whereas its other configurations are more restrictive. Agent is more restrictive than full CPT with Sonnet 5 or GPT-6 Luna, but less restrictive with the other language models. Agent has the lowest restrictiveness with Opus 5.5, at $0.016$. Haiku 5.5 and Sonnet 5.5 produce similar Agent results despite their different Completion results.

The returned programs also differ in their use of the supplied data for estimation. Numerical parameter estimation is included in $72\%$ of Completion fits with Sonnet 5.5, compared with $99\%$ with Haiku 5.5, $91\%$ with GPT-6 Sol, and all fits with Opus 5.5. Many Completion fits with Sonnet 5.5 use only a subset of the 100 observations. Agent can estimate parameters and compare specifications before returning its prediction rule, although unsuccessful calculations consume part of its fixed computation budget. Agent with Sonnet 5 reaches the 20-round limit in 44 of 100 replications, compared with none with Sonnet 5.5.

All six Agent configurations predict the observed responses about as accurately as TabPFN. Completion's performance varies more across language models. With Sonnet 5, Completion fails to produce a valid prediction rule in eight of the ten folds. With GPT-6 Luna, one fit fails and another has large prediction errors, producing negative completeness. These outcomes reflect differences in both the returned specifications and their successful estimation and execution.

\paragraph{Prompt Design}

\begin{table}[t!]
\centering
\caption{Role messages and prediction instructions across prompt designs}
\label{tab:app1-prompt-sentences}
\small\singlespacing\setlength{\tabcolsep}{4pt}
\renewcommand{\arraystretch}{1.15}
\begin{tabular}{@{}>{\raggedright\arraybackslash}p{0.225\textwidth}>{\raggedright\arraybackslash}p{0.325\textwidth}>{\raggedright\arraybackslash}p{0.40\textwidth}@{}}
\toprule
Prompt & Role message & Prediction instruction \\
\midrule
Neutral & None & ``Using the supplied data, develop the best prediction algorithm you can for certainty equivalents.'' \\[5pt]
Econometrician role & ``You are an econometrician specializing in decision theory and choice under risk.'' & Same as neutral \\[5pt]
Economic-model instruction & Same as econometrician role & ``Using the supplied data, select an appropriate economic model of risk preferences, estimate its parameters, and use the estimated model to predict certainty equivalents.'' \\[5pt]
Data-scientist role & ``You are a data scientist specializing in machine learning and predictive modeling.'' & Same as neutral \\[5pt]
Machine-learning instruction & Same as data-scientist role & ``Using the supplied data, select an appropriate machine-learning method, train it, and use the trained model to predict certainty equivalents.'' \\[5pt]
\bottomrule
\end{tabular}
\par\smallskip
{\footnotesize\raggedright\emph{Notes:} Role messages precede the task as system messages; prediction instructions appear within the task. All other prompt text is identical across the five designs. Completion and Agent receive the same task and role messages; Agent also receives instructions for using its execution tool.\par}
\end{table}

Our second comparison examines prompt design using Claude Sonnet 5.5. We compare the neutral prompt, which specifies the prediction task without assigning a professional role or prescribing a modeling approach, with prompts that assign the LLM the role of an econometrician or data scientist. For each of the two specific roles, we also consider an additional instruction to select and estimate an economic model of risk preferences or to select and train a machine-learning method, respectively.
Table~\ref{tab:app1-prompt-sentences} quotes the role messages and prediction instructions. The remaining prompt text, supplied data, prediction requirements, and computation budgets are unchanged across prompt designs.

Figure~\ref{fig:app1-prompt-comparison} reports the restrictiveness of Completion and Agent under different prompt designs.
For Completion, assigning the data-scientist role reduces restrictiveness. Relative to the neutral prompt, the role alone reduces restrictiveness by $0.145$ (paired s.e.\ $0.052$), and the role combined with the machine-learning instruction reduces it by $0.190$ ($0.048$). 

The returned programs show that the data-scientist prompts also change how Completion uses the supplied observations in estimation. Under the neutral prompt, 72 of the 100 programs specify numerical estimation, including 36 that use the full sample. The corresponding counts rise to 99 and 72 under the data-scientist role, and to 100 and 86 when the role is combined with the machine-learning instruction. Appendix Table~\ref{tab:app1-prompt-estimation} reports these frequencies alongside restrictiveness. More frequent estimation and greater use of the supplied sample help explain Completion's lower restrictiveness under these prompts.

\begin{figure}[t!]
\centering
\caption{Prompt design and restrictiveness ($n=100$)}
\label{fig:app1-prompt-comparison}
\includegraphics[width=\textwidth]{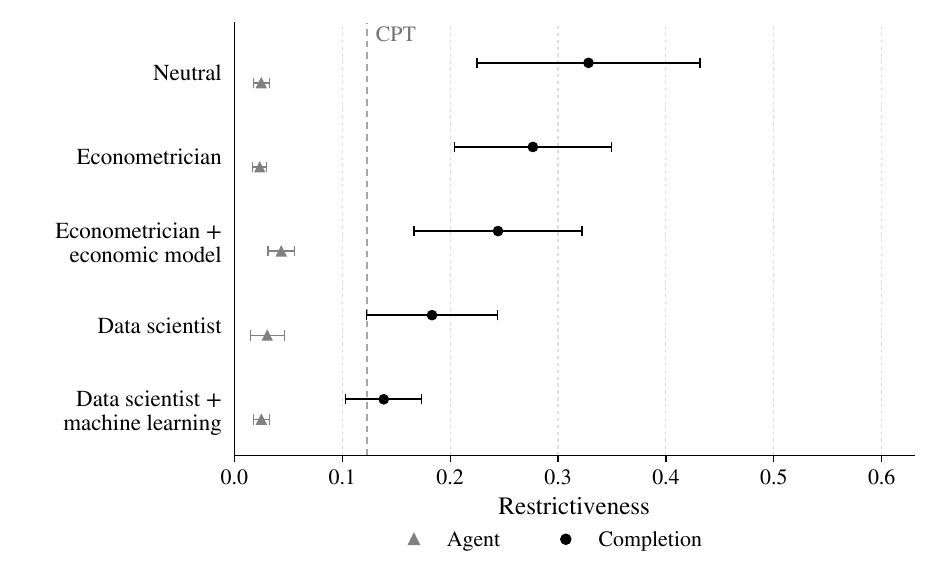}
{\footnotesize\raggedright\emph{Notes:} All prompts use Claude Sonnet 5.5 and the same 100 synthetic samples of size $n=100$. Bars are 95\% confidence intervals. The dashed line is CPT full model. Online Appendix~\ref{subsec:app1-prompt-design} reports the numerical results.\par}
\end{figure}

Agent's restrictiveness changes little under either professional role or the machine-learning instruction. The economic-model instruction increases restrictiveness by $0.019$ (s.e.\ $0.004$), to $0.043$, while leaving it below that of CPT full model. This instruction also shifts specification choice toward economic models of risk preferences: Agent selects RDU specifications in $58\%$ of the simulated samples under the economic-model prompt, compared with $10\%$ under the neutral prompt. 

Completeness remains between $0.985$ and $0.994$ under all five prompts for both procedures. Prompt design therefore changes restrictiveness while leaving empirical fit nearly unchanged. For Completion, the data-scientist prompts are accompanied by more extensive use of the supplied observations in estimation; for Agent, the economic-model instruction shifts specification choice toward RDU.
\par

\FloatBarrier

\section{Application 2: Discrete Choice}\label{sec:app-dc}

Our second application concerns consumer choice among differentiated products, another important problem in economics. Economists develop various discrete-choice models to describe how product characteristics affect market shares of products. Beyond predicting market shares, economists also use such models to study substitution patterns across products. Below, we compare three workhorse discrete-choice models, logit, nested logit, and mixed logit, along with the LLM and TFM procedures.

We measure restrictiveness on simulated demand systems, separately for market shares and their implied own- and cross-price elasticities, and completeness for share predictions using the cereal data of \citet{nevo2000practitioner}. The exercise illustrates how restrictiveness can evaluate the economic implications of fitted prediction rules beyond their fit to the outcomes used for estimation. Richer structural specifications approximate simulated shares better without improving implied price elasticities, while Agent and TabPFN are substantially less restrictive on both dimensions.

\subsection{Setting, Evaluation Distribution, and Discrepancy}

We use the cereal-demand teaching data of \citet{nevo2000practitioner}, which contain 24 products across 94 markets, to specify the number of products and calibrate the support and units of the synthetic exercise. Accordingly, a market $X=(x_1,\ldots,x_{24})$ contains $J=24$ products, each described by its price $p_j$, a binary characteristic $m_j$ (``mushy''), and sugar content $s_j$. A prediction rule returns 25 strictly positive shares that sum to one, $f(X)=(f_0(X),\ldots,f_{24}(X))$, where $f_0$ is the outside share.

For each mushy-status category $m\in\{0,1\}$, the support of price and sugar content $(p,s)$ is the convex hull $\mathcal H_m$ of the observed price--sugar pairs in that category.  To draw a synthetic product, we first draw $m\sim\mathrm{Bernoulli}(1/3)$, matching the fraction of mushy products, and then draw $(p,s)$ uniformly over $\mathcal H_m$. For each market, we independently draw 24 products, so that each market contains a new set of products.

In this application, eligible prediction rules are market-share functions generated by random-utility demand systems with i.i.d.\ type-I extreme-value shocks, consumer heterogeneity, and utility strictly decreasing
in own price. To construct a draw $f$ from the evaluation distribution $\lambda_{\mathcal F}$, we specify the systematic utility of product $j$ for a consumer of type $\nu\sim N(0,I)$ as
\[
V_j(\nu)=-\log 24+L(z^s_j,m_j,\nu)-\int_0^{z^p_j}a(v,z^s_j,m_j,\nu)\,dv,
\quad a=\log(1+\exp\Lambda)/c,
\]
where $z^p$ and $z^s$ denote standardized price and sugar content variables. The function $L$ captures preferences for nonprice characteristics, while $a$ represents price sensitivity. Both $L$ and $\Lambda$ are specified as sums of two Gaussian processes: one common to all consumers and another that varies with consumer characteristics $\nu$. We use Mat\'ern-$3/2$ covariance for the continuous inputs and correlation $0.5$ across mushy-status categories. The constant $c$ normalizes mean price sensitivity, and the positivity of $a$ ensures that utility is strictly decreasing in price. The term $-\log 24$ adjusts for the number of inside products, so that the aggregate inside share and the outside share are each $1/2$ when the other utility terms are zero.

Normalizing the outside utility to $V_0(\nu)=0$, the population market shares are
\[
f_j(X) =
\E_{\nu}\left[
\frac{\exp(V_j(\nu))}
{1+\sum_{k=1}^{24}\exp(V_k(\nu))}
\right],
\quad j=0,\ldots,24.
\]
These demand systems imply negative own-price and positive cross-price elasticities, so that products are gross substitutes. Consumer heterogeneity allows substitution patterns that need not satisfy the independence of irrelevant alternatives (IIA) property of multinomial logit.

We draw $M=200$ eligible market-share functions independently from $\lambda_{\mathcal F}$. For each draw, we generate a sequence of 200 markets and supply the first $n\in\{50,100,200\}$ markets to every procedure. We measure restrictiveness using a separate evaluation sample of 100 markets drawn from the same distribution as the training samples.

All models and procedures are estimated using market-share observations. We evaluate their restrictiveness using three discrepancies: one for market shares and two for the own- and cross-price elasticities implied by the fitted share functions. For a differentiable share function $g$, define the price elasticity of product $j$ with respect to the price of product $k$ as\footnote{We compute the implied price elasticities using symmetric finite differences in log prices, as described in Online Appendix
Section~\ref{subsec:evaluation.example2}.}
\[ 
\mathcal E_{jk}[g](X) = \frac{\partial \log g_j(X)}{\partial \log p_k}, \quad j,k=1,\ldots,24. 
\] 
The share discrepancy is the mean squared error of the 24 inside shares. The own-price elasticity discrepancy averages the squared errors of the 24 own-price elasticities, while the cross-price elasticity discrepancy averages those of the 552 cross-price elasticities.

The baseline model class is the common-intercept logit
$b_j=e^\gamma/(1+24e^\gamma)$, and its price elasticities are zero. We estimate each restrictiveness measure as the ratio of the average procedure loss to the average baseline loss across replications. All procedures use common draws for paired comparisons.

For completeness, we use the cereal data of \citet{nevo2000practitioner}, which contain 94 markets across 47 cities. We evaluate completeness only for market-share predictions, since the data do not directly reveal structural counterfactuals like price elasticities. The loss is the mean squared prediction error of the 24 inside shares. All procedures use the same five repetitions of five-fold cross-validation over cities, with both markets in a city assigned to the same fold. The completeness measure $\kappa$ expresses the improvement over the common-intercept logit as a fraction of the improvement achieved by TabPFN, which serves as the flexible benchmark.

\subsection{Economic Models and Procedures}
We now describe the economic models and black-box procedures used in the comparison. We consider three discrete choice models. For each inside product $j=1,\ldots,24$ in market $X$, let $w_j=(1,p_j,s_j,m_j)^\top$ collect an intercept, price, sugar content, and mushy status. 
Multinomial logit (MNL) specifies the systematic utility of product $j$ as $\delta_j=w_j^\top\beta$, with outside-option utility normalized to zero. 
Its market share is 
\[ 
S_j^{\mathrm{MNL}}(X) = \frac{\exp(\delta_j)} {1+\sum_{k=1}^{24}\exp(\delta_k)}, \quad j=1,\ldots,24. 
\] 
Nested logit (NL) groups products into two nests based on mushy status. 
Let $D_m=\sum_{k:m_k=m}\exp(\delta_k/(1-\rho_m))$. The market share of product $j$ is 
\[ 
S_j^{\mathrm{NL}}(X) = \frac{\exp(\delta_j/(1-\rho_{m_j})) D_{m_j}^{-\rho_{m_j}}} {1+\sum_{m=0}^1 D_m^{1-\rho_m}}, 
\]
where $\rho_0,\rho_1\in[0,1)$ govern within-nest substitution. 
Mixed logit (MXL) allows the utility coefficients to vary across consumers according to $b\sim N(\beta,\Sigma)$, yielding 
\[
S_j^{\mathrm{MXL}}(X) = \E_b\left[ \frac{\exp(w_j^\top b)} {1+\sum_{k=1}^{24}\exp(w_k^\top b)} \right]. 
\] 
MNL, NL, and MXL have four, six, and fourteen parameters, respectively. All three models are fitted by least squares to the inside market shares in the training sample.

We compare these structural models with three AI-based procedures. TabPFN predicts each product's log share relative to the outside option, $\log S_j-\log S_0$, using the price, sugar content, and mushy status of all 24 products. We estimate 24 regressions, one for each product, and convert the predicted log share ratios to shares using the logit formula. Dependence on the characteristics of the whole market allows flexible substitution patterns. We also evaluate Completion and Agent as in Application 1, both using Claude Sonnet 5.5. We adapt the neutral prompt to the market-share prediction task as follows:
\begin{recipebox}[breakable=false]{Task for Completion and Agent}
\small
This is an economic prediction problem concerning demand for differentiated products.

Each observation in the supplied dataset is a complete market containing 24 products. Its input is a 24-by-3 array with columns price, mushy, and sugar. These represent the product's price, a binary product characteristic, and sugar content.

Each target contains 25 market shares: the outside share---the proportion choosing none of these products---followed by the share corresponding to each product row in the same order.

Data: \emph{[the $n$ supplied markets]}

Using the supplied data, develop the best prediction algorithm you can for market shares.

Return Python code implementing predict\_shares(market), which accepts a 24-by-3 array and returns 25 finite, strictly positive shares that sum to one, in the stated order.

Include any supporting code, fitting steps, or retained data needed to execute the prediction function. Use only the supplied data for fitting and model selection. Once constructed, the predictor must produce reproducible outputs independently of previous prediction queries.
\end{recipebox}

Completion produces its prediction program without execution feedback, while Agent can run Python on the supplied data for up to 20 rounds before returning its program. The returned programs may estimate parameters from the supplied training data when initialized or first queried. Invalid share predictions receive the fitted baseline model for both shares and price elasticities.

\subsection{Main Results}

\begin{table}[t!]
  \caption{Discrete choice: restrictiveness  at $n=100$ and completeness}
  \label{tab:app2-main}
  \centering
  \small\setlength{\tabcolsep}{3pt}%
  \begin{tabular}{@{}lcccc@{}}\hline
   & \multicolumn{3}{c}{Restrictiveness} & Completeness \\
   \cmidrule(lr){2-4}\cmidrule(lr){5-5}
   & Shares & Own-price & Cross-price & $\kappa$ \\ \hline
  MNL & 0.183 (0.009) & 0.299 (0.011) & 0.267 (0.009) & 0.165 (0.044) \\
  NL & 0.173 (0.009) & 0.409 (0.028) & 0.392 (0.031) & 0.162 (0.043) \\
  MXL & 0.140 (0.009) & 0.310 (0.011) & 0.416 (0.017) & 0.167 (0.049) \\
  TabPFN & 0.046 (0.003) & 0.075 (0.002) & 0.172 (0.007) & 1 (reference) \\
  Completion & 0.871 (0.012) & 0.653 (0.025) & 0.884 (0.008) & $-$0.802 (0.246) \\
  Agent & 0.026 (0.010) & 0.082 (0.009) & 0.106 (0.019) & 1.149 (0.044) \\ \hline
  \end{tabular}
  \par\smallskip
{\footnotesize\raggedright\emph{Notes:} The Own-price and Cross-price columns report restrictiveness for own-price and cross-price elasticities, respectively. Parentheses report standard errors.\par}
\end{table}

\begin{table}[t!]
  \caption{Discrete choice: restrictiveness profiles over $n$ (shares / own / cross)}
  \label{tab:app2-profiles}
  \centering
  \small\setlength{\tabcolsep}{4pt}%
  \begin{tabular}{@{}lccc@{}}\hline
   & $n=50$ & $n=100$ & $n=200$ \\ \hline
  MNL & 0.184 / 0.297 / 0.268 & 0.183 / 0.299 / 0.267 & 0.182 / 0.299 / 0.265 \\
  NL & 0.175 / 0.408 / 0.395 & 0.173 / 0.409 / 0.392 & 0.172 / 0.406 / 0.381 \\
  MXL & 0.142 / 0.314 / 0.425 & 0.140 / 0.310 / 0.416 & 0.139 / 0.314 / 0.412 \\
  TabPFN & 0.161 / 0.135 / 0.332 & 0.046 / 0.075 / 0.172 & 0.016 / 0.042 / 0.099 \\
  Completion & 0.880 / 0.648 / 0.888 & 0.871 / 0.653 / 0.884 & 0.907 / 0.660 / 0.879 \\
  Agent & 0.081 / 0.180 / 0.199 & 0.026 / 0.082 / 0.106 & 0.024 / 0.096 / 0.090 \\ \hline
  \end{tabular}
  \par\smallskip
{\footnotesize\raggedright\emph{Notes:} Entries report share / own-price elasticity / cross-price elasticity restrictiveness.\par}
\end{table}

Table~\ref{tab:app2-main} reports restrictiveness at $n=100$ and completeness, and Table~\ref{tab:app2-profiles} reports the restrictiveness profiles over $n$. Three findings emerge.

First, richer structural specifications improve share fit without improving the implied price elasticities. At $n=100$, share restrictiveness declines from $0.183$ for MNL to $0.173$ for NL and $0.140$ for MXL.
\footnote{Appendix~C shows that the lower demand-share restrictiveness of NL and MXL relative to MNL is robust to reweighting the evaluation distribution. Some elasticity comparisons are more sensitive.}
 MXL's cross-price elasticity restrictiveness nevertheless exceeds MNL's by $0.149$ (paired s.e.\ $0.015$), while their own-price elasticity measures are similar. NL performs worse than MNL on both elasticity measures. These patterns change little with sample size. Thus, when estimated from market shares, as is common in empirical demand analysis, NL and MXL can yield substitution patterns further from the truth than MNL, despite relaxing its IIA restriction. This distinction matters because substitution patterns are typically inferred from fitted demand models rather than directly targeted in estimation. The finding therefore concerns the fitted procedures, not the restrictiveness of their underlying model classes under the price-elasticity discrepancies.

Second, Agent and TabPFN are substantially less restrictive than the fitted structural models in both shares and price elasticities at $n=100$. Agent has  lower price-elasticity restrictiveness than the structural models at every sample size, while TabPFN's approximation improves considerably with more data. At $n=200$, TabPFN's share, own-price elasticity, and cross-price elasticity restrictiveness are $0.016$, $0.042$, and $0.099$, respectively, below MNL's by $0.166$, $0.257$, and $0.166$ (standard errors about $0.01$). Agent's restrictiveness also falls substantially between $n=50$ and $n=100$, and its cross-price elasticity restrictiveness at $n=100$ is below TabPFN's by $0.066$ (s.e.\ $0.021$). In contrast, Completion is the most restrictive procedure: at $n=100$, its share, own-price elasticity, and cross-price elasticity restrictiveness are $0.871$, $0.653$, and $0.884$, respectively. These values correspond to approximately $65\%$--$88\%$ of baseline loss across the three criteria.

Third, the procedures also differ substantially in their completeness on the observed cereal data. Unlike the synthetic markets, the cereal data preserve product identities across markets. This allows Agent and TabPFN to estimate product-specific demand patterns, whereas the structural
specifications considered here omit product fixed effects. Agent reduces out-of-sample mean squared prediction error by $30\%$ relative to the common-intercept logit, compared with $26\%$ for TabPFN. The structural models capture about one sixth of TabPFN's improvement over the baseline. Completion performs worse than the baseline ($\kappa=-0.802$, s.e.\ $0.246$), with invalid share predictions in 14 of the 25 folds. As with the experimental lottery responses, language models' completeness may also reflect prior exposure to these public observations.

\subsection{Specification Choices and Estimation}
\label{subsec:app2-specifications}

To better understand the differences in restrictiveness and completeness between Completion and Agent, we examine their specification choices and use of the supplied data in estimation of the chosen specifications.
Figure~\ref{fig:app2-neutral-model-choices} reports the prediction rules and execution outcomes for synthetic data and
observed cereal data.

\begin{figure}[t!]
\centering
\setstretch{1}
\caption{Prediction rules and execution outcomes in the demand application}
\label{fig:app2-neutral-model-choices}
\begin{minipage}[t]{0.48\textwidth}
  \centering
  \refstepcounter{empiricalpanel}
  \label{fig:app2-neutral-model-choices-restrictiveness-completion_reasoning_disabled}
  \captionsetup{font=small,position=top,skip=6pt}
  \caption*{(\alph{empiricalpanel}) Completion: restrictiveness}
  \includegraphics[width=\linewidth]{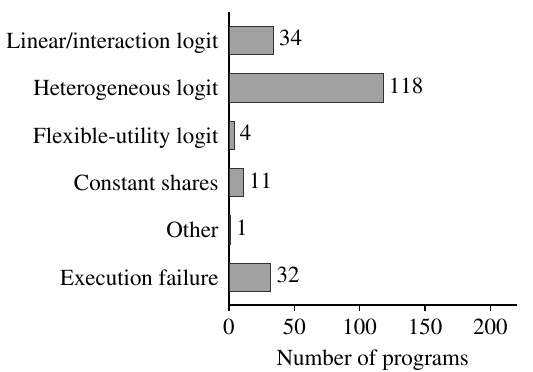}
\end{minipage}\hfill
\begin{minipage}[t]{0.48\textwidth}
  \centering
  \refstepcounter{empiricalpanel}
  \label{fig:app2-neutral-model-choices-completeness-completion_reasoning_disabled}
  \captionsetup{font=small,position=top,skip=6pt}
  \caption*{(\alph{empiricalpanel}) Completion: completeness}
  \includegraphics[width=\linewidth]{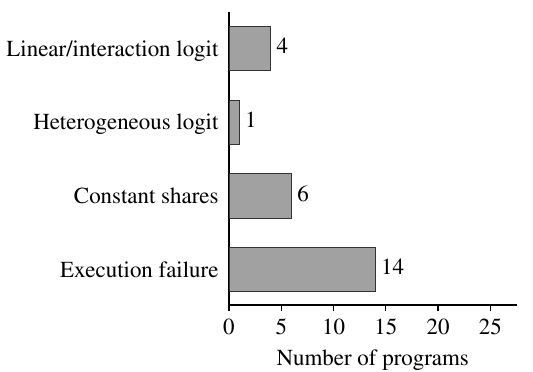}
\end{minipage}
\par\medskip
\begin{minipage}[t]{0.48\textwidth}
  \centering
  \refstepcounter{empiricalpanel}
  \label{fig:app2-neutral-model-choices-restrictiveness-agent}
  \captionsetup{font=small,position=top,skip=6pt}
  \caption*{(\alph{empiricalpanel}) Agent: restrictiveness}
  \includegraphics[width=\linewidth]{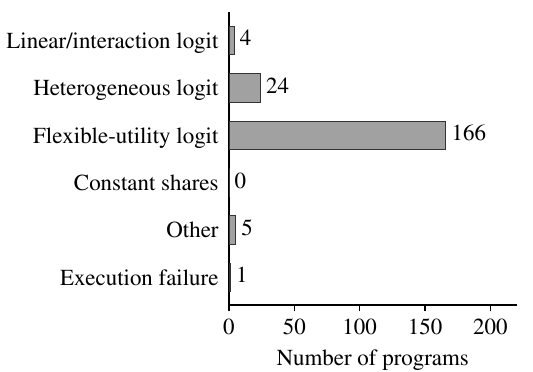}
\end{minipage}\hfill
\begin{minipage}[t]{0.48\textwidth}
  \centering
  \refstepcounter{empiricalpanel}
  \label{fig:app2-neutral-model-choices-completeness-agent}
  \captionsetup{font=small,position=top,skip=6pt}
  \caption*{(\alph{empiricalpanel}) Agent: completeness}
  \includegraphics[width=\linewidth]{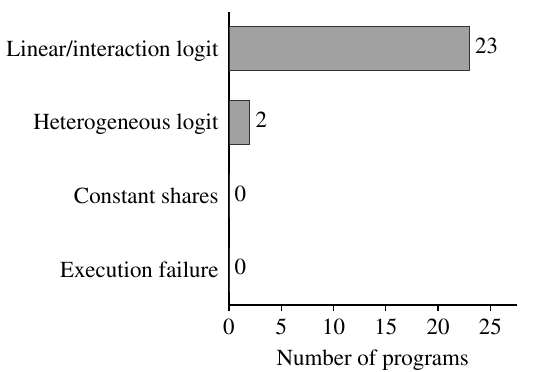}
\end{minipage}
\par\medskip

{\footnotesize\raggedright\emph{Notes:} Left panels: 200 fits to simulated data per procedure at $n=100$. Right panels: 25 fits to cereal market shares per procedure from five repetitions of five city folds. Category definitions are in Online Appendix Section~\ref{subsec:specifications.example2}.\par}
\end{figure}

On the synthetic demand systems, Completion primarily selects heterogeneous logits ($59\%$ of fits), which average conditional-logit choice probabilities over unobserved consumer types or utility shocks. Agent predominantly uses flexible-utility logits ($83\%$ of fits), which allow utility indices to depend nonlinearly on market characteristics (e.g., GP and kernel regressions of log share ratios). Agent can execute code and compare fitted specifications before returning its prediction rule but  Completion must select and implement a specification without execution feedback. This difference may help explain Agent's greater reliance on flexible-utility specifications that better capture the nonlinearities in the synthetic demand systems, contributing to its lower restrictiveness.

The procedures also differ substantially in how they estimate their chosen specifications. At $n=100$, only one of the 200 Completion programs successfully estimates parameters numerically from the supplied sample, 167 produce predictions using fixed or default coefficients, and the remaining 32 fail and receive baseline predictions. Some programs define fitting routines that are never invoked or leave training-data placeholders empty. Without execution feedback, Completion cannot detect and correct these problems. Agent instead estimates its prediction rules during tool use or when the returned program runs, with fitted parameters retained for prediction. Among the 168 replications where both procedures produce predictions without fallback, share restrictiveness remains $0.845$ for Completion and $0.019$ for Agent. The gap in restrictiveness therefore extends beyond execution failures.

On the observed data, Agent selects relatively simple choice specifications with product fixed effects. It uses ordinary logit in 23 of the 25 fits and mixed logit with a normally distributed random coefficient on mushy status in the remaining two. Agent successfully estimates parameters numerically during tool use in all 25 fits and retains the fitted coefficients in its returned prediction programs. Product intercepts capture persistent differences in shares across the 24 cereal products. Completion does not successfully estimate parameters numerically in any of its 25 fits: 11 programs produce predictions using fixed or default coefficients, including six constant-share rules, while the remaining 14 fail and receive baseline predictions. These results suggest that Agent's higher completeness reflects its ability to select and estimate specifications suited to the observed data, rather than simply choosing more elaborate models.

\FloatBarrier

\section{Conclusion}\label{sec:conc}
We extend the notion of restrictiveness in \cite{fudenberg2023flexible} to functional settings with continuum domains and to prediction procedures, including LLMs, TFMs, and AI workflows. Procedural restrictiveness evaluates how well a procedure approximates eligible prediction rules at a given sample size, without requiring an explicit description of its model class. Together with completeness \citep{FKLM}, it allows economic models and AI procedures to be compared by both their predictive performance and the structure they impose.

In the lottery application, similar completeness accompanies substantial differences in restrictiveness across economic models and LLM workflows. In differentiated-product demand, richer structural specifications improve share fit without improving implied price elasticities. These findings illustrate why predictive performance alone gives an incomplete account of a model or procedure.

Further work could extend these comparisons across tasks and workflow designs, improve computation in high-dimensional settings, and examine whether restrictiveness can inform model selection or regularization.

\setstretch{1}

\bibliography{library}

@misc{projectAPE,
  title={Project {APE}: {A}utonomous Policy Evaluation},
  author={Yanagizawa-Drott, David},
  year={2026},
  howpublished={\texttt{https://ape.socialcatalystlab.org/}, Social 
  Catalyst Lab,  University of Zurich},
  institution={University of Zurich, Department of Economics},
  note={Accessed: February 2026}
}

@article{crawford2025demand,
  title={Demand Systems Without Errors},
  author={Crawford, Ian},
  journal={International Economic Review},
  volume={66},
  number={4},
  pages={1599--1617},
  year={2025},
  publisher={Wiley Online Library}
}

@article{EllisNeff,
  title={Model Complexity and Restrictiveness},
  author={Keaton Ellis and Sara Neff},
  year={2026},
  journal={arXiv:2610.00784}
}

@article{legratietgarnier2015,
  title={Asymptotic analysis of the learning curve for {Gaussian} process regression},
  author={Le Gratiet, Loic and Garnier, Josselin},
  journal={Machine Learning},
  volume={98},
  number={3},
  pages={407--433},
  year={2015},
  publisher={Springer}
}

@article{Selten,
	Author = {Reinhard Selten },
	Journal = {Mathematical Social Sciences},
	Issue = {2},
	Pages = {153-167},
	Title = {Properties for a Measure of Predictive Success},
	Volume = {21},
	Year = {1991}}

@article{FKLM,
author = {Fudenberg, Drew and Kleinberg, Jon and Liang, Annie and Mullainathan, Sendhil},
title = {Measuring the Completeness of Economic Models},
journal = {Journal of Political Economy},
volume = {130},
number = {4},
pages = {956-990},
year = {2022}}

@article{Bruhin,
	Author = {Adrian Bruhin and Fehr-Duda, Helga and Epper, Thomas},
	Journal = {Econometrica},
	Title = {Risk and Rationality: Uncovering Heterogeneity in Probability Distortion},
	Year = {2010},
	volume={78},
	number={4},
	pages={1375--1412} 
}

@article{Quiggin,
	Author = {John Quiggin},
	Journal = {Journal of Economic Behavior and Organization},
	Pages = {323-343},
	Title = {A Theory of Anticipated Utility},
	Volume = {3},
	Number={4},
	Year = {1982}}

@article{CPT,
	Author = {Amos Tversky and Daniel Kahneman},
	Journal = {Journal of Risk and Uncertainty},
	Number = {4},
	Pages = {297-323},
	Title = {Advances in Prospect Theory: Cumulative Representation of Uncertainty},
	Volume = {5},
	Year = {1992}}

@article{gul1991,
  title={A Theory of Disappointment Aversion},
  author={Gul, Faruk},
  journal={Econometrica},
  volume={59},
  number={3},
  pages={667--686},
  year={1991},
  publisher={JSTOR}
}

@article{fudenberg2023flexible,
    author = {Fudenberg, Drew and Gao, Wayne and Liang, Annie},
    title = {How Flexible Is that Functional Form? {Quantifying} the Restrictiveness of Theories},
    journal = {Review of Economics and Statistics},
    volume = {108},
    number = {1},
    pages = {194-209},
    year = {2026},
    month = {01}  
}

@article{quintana2022dependent,
  title={The dependent {Dirichlet} process and related models},
  author={Quintana, Fernando A and M{\"u}ller, Peter and Jara, Alejandro and MacEachern, Steven N},
  journal={Statistical Science},
  volume={37},
  number={1},
  pages={24--41},
  year={2022},
  publisher={Institute of Mathematical Statistics}
}

@book{williams2006gaussian,
  title={Gaussian processes for machine learning},
  author={Williams, Christopher KI and Rasmussen, Carl Edward},
  volume={2},
  year={2006},
  publisher={MIT Press Cambridge, MA}
}

@article{swiler2020survey,
  title={A survey of constrained {Gaussian} process regression: Approaches and implementation challenges},
  author={Swiler, Laura P and Gulian, Mamikon and Frankel, Ari L and Safta, Cosmin and Jakeman, John D},
  journal={Journal of Machine Learning for Modeling and Computing},
  volume={1},
  number={2},
  year={2020},
  publisher={Begel House Inc.}
}

@article{maatouk2017gaussian,
  title={Gaussian process emulators for computer experiments with inequality constraints},
  author={Maatouk, Hassan and Bay, Xavier},
  journal={Mathematical Geosciences},
  volume={49},
  number={5},
  pages={557--582},
  year={2017},
  publisher={Springer}
}

@article{snelson2003warped,
  title={Warped {Gaussian} processes},
  author={Snelson, Edward and Ghahramani, Zoubin and Rasmussen, Carl},
  journal={Advances in Neural Information Processing Systems},
  volume={16},
  year={2003}
}

@inproceedings{riihimaki2010gaussian,
  title={Gaussian processes with monotonicity information},
  author={Riihim{\"a}ki, Jaakko and Vehtari, Aki},
  booktitle={Proceedings of the thirteenth international conference on artificial intelligence and statistics},
  pages={645--652},
  year={2010},
  organization={JMLR Workshop and Conference Proceedings}
}

@article{shively2009bayesian,
  title={A {Bayesian} approach to non-parametric monotone function estimation},
  author={Shively, Thomas S and Sager, Thomas W and Walker, Stephen G},
  journal={Journal of the Royal Statistical Society Series B: Statistical Methodology},
  volume={71},
  number={1},
  pages={159--175},
  year={2009},
  publisher={Oxford University Press}
}

@article{breiman2001statistical,
  author  = {Breiman, Leo},
  title   = {Statistical Modeling: The Two Cultures},
  journal = {Statistical Science},
  year    = {2001},
  volume  = {16},
  number  = {3},
  pages   = {199--231}
}

@article{korinek2023generative,
  author  = {Korinek, Anton},
  title   = {Generative {AI} for Economic Research: Use Cases and Implications for Economists},
  journal = {Journal of Economic Literature},
  year    = {2023},
  volume  = {61},
  number  = {4},
  pages   = {1281--1317}
}

@article{ludwig2025large,
  author  = {Ludwig, Jens and Mullainathan, Sendhil and Rambachan, Ashesh},
  title   = {Large Language Models: An Applied Econometric Framework},
  journal = {Annual Review of Economics},
  volume  = {18},
  pages   = {283--316},
  year    = {2026},
  doi     = {10.1146/annurev-economics-120925-105620}
}

@techreport{horton2023large,
  author      = {Horton, John J.},
  title       = {Large Language Models as Simulated Economic Agents: What Can We Learn from {Homo Silicus}?},
  institution = {National Bureau of Economic Research},
  type        = {Working Paper},
  number      = {31122},
  year        = {2023},
  doi         = {10.3386/w31122},
  url         = {https://www.nber.org/system/files/working_papers/w31122/revisions/w31122.rev0.pdf}
}

@article{ashokkumar2026experiments,
  author  = {Ashokkumar, Ashwini and Hewitt, Luke and Ghezae, Isaias and Willer, Robb},
  title   = {Large Language Models Can Predict the Results of Social Science Experiments},
  journal = {Nature},
  volume  = {656},
  pages   = {115--122},
  year    = {2026},
  doi     = {10.1038/s41586-026-10742-x}
}

@article{fariaeCastroLeibovici2024AIInflation,
  author  = {Faria-e-Castro, Miguel and Leibovici, Fernando},
  title   = {Artificial Intelligence and Inflation Forecasts},
  journal = {Federal Reserve Bank of St. Louis Review},
  volume  = {106},
  number  = {12},
  pages   = {1--14},
  year    = {2024},
  doi     = {10.20955/r.2024.12}
}

@article{lopezlira2026chatgpt,
  author  = {Lopez-Lira, Alejandro and Tang, Yuehua},
  title   = {Can {ChatGPT} Forecast Stock Price Movements? {Return} Predictability and Large Language Models},
  journal = {Journal of Financial Economics},
  volume  = {184},
  pages   = {104335},
  year    = {2026},
  doi     = {10.1016/j.jfineco.2026.104335}
}

@techreport{manning2024automated,
  author      = {Manning, Benjamin S. and Zhu, Kehang and Horton, John J.},
  title       = {Automated Social Science: Language Models as Scientist and Subjects},
  institution = {National Bureau of Economic Research},
  type        = {Working Paper},
  number      = {32381},
  year        = {2024},
  doi         = {10.3386/w32381}
}

@article{gao2025take,
  author  = {Gao, Yuan and Lee, Dokyun and Burtch, Gordon and Fazelpour, Sina},
  title   = {Take Caution in Using {LLM}s as Human Surrogates},
  journal = {Proceedings of the National Academy of Sciences},
  volume  = {122},
  number  = {24},
  pages   = {e2501660122},
  year    = {2025},
  doi     = {10.1073/pnas.2501660122}
}

@article{chen2026preferences,
  author  = {Chen, Fadong and Wang, Qijun and Wang, Lei and Bhatia, Sudeep},
  title   = {Simulated Responses from Large Language Models Are Weak Proxies for Human Economic Preferences across 74 Countries},
  journal = {Nature Human Behaviour},
  year    = {2026},
  note    = {Advance online publication, October 12},
  doi     = {10.1038/s41562-026-02611-4}
}

@article{hollmann2025tabpfn,
  author  = {Hollmann, Noah and M{\"u}ller, Samuel and Purucker, Lennart and Krishnakumar, Arjun and K{\"o}rfer, Max and Hoo, Shi Bin and Schirrmeister, Robin Tibor and Hutter, Frank},
  title   = {Accurate Predictions on Small Data with a Tabular Foundation Model},
  journal = {Nature},
  year    = {2025},
  volume  = {637},
  number  = {8045},
  pages   = {319--326}
}

@article{qu2026tabicl,
  author  = {Qu, Jingang and Holzm{\"u}ller, David and Varoquaux, Ga{\"e}l and Le Morvan, Marine},
  title   = {{TabICLv2}: A Better, Faster, Scalable, and Open Tabular Foundation Model},
  journal = {arXiv preprint arXiv:2602.11139},
  year    = {2026}
}

@article{nevo2000practitioner,
  title={A Practitioner's Guide to Estimation of Random-Coefficients Logit Models of Demand},
  author={Nevo, Aviv},
  journal={Journal of Economics \& Management Strategy},
  volume={9},
  number={4},
  pages={513--548},
  year={2000}
}

@article{WassermanKadane1992,
  author  = {Wasserman, Larry and Kadane, Joseph B.},
  title   = {Computing Bounds on Expectations},
  journal = {Journal of the American Statistical Association},
  year    = {1992},
  volume  = {87},
  number  = {418},
  pages   = {516--522},
  doi     = {10.1080/01621459.1992.10475234}
}

@article{GewekePetrella1998,
  author  = {Geweke, John and Petrella, Lea},
  title   = {Prior Density-Ratio Class Robustness in Econometrics},
  journal = {Journal of Business \& Economic Statistics},
  year    = {1998},
  volume  = {16},
  number  = {4},
  pages   = {469--478},
  doi     = {10.1080/07350015.1998.10524786}
}

@article{Ferguson1973,
  author  = {Ferguson, Thomas S.},
  title   = {A {Bayesian} Analysis of Some Nonparametric Problems},
  journal = {The Annals of Statistics},
  volume  = {1},
  number  = {2},
  pages   = {209--230},
  year    = {1973},
  doi     = {10.1214/aos/1176342360}
}

@inproceedings{rasmussen2001infinite,
  author    = {Rasmussen, Carl Edward and Ghahramani, Zoubin},
  title     = {Infinite Mixtures of {Gaussian} Process Experts},
  booktitle = {Advances in Neural Information Processing Systems},
  volume    = {14},
  year      = {2001},
  url       = {https://proceedings.neurips.cc/paper_files/paper/2001/hash/9afefc52942cb83c7c1f14b2139b09ba-Abstract.html}
}
\bibliographystyle{ecca}
\onehalfspacing

\vspace{1em}
\appendix\setcounter{figure}{0}
\setcounter{table}{0}
\renewcommand{\thefigure}{A.\arabic{figure}}
\renewcommand{\thetable}{A.\arabic{table}}
\renewcommand{\thesubsection}{\Alph{subsection}}
\setcounter{subsection}{0}
\section*{Appendix}

\subsection{Proofs}\label{sec:proofs}
\begin{proof}[Proof of Proposition~\ref{prop:learning-curve}]
By Assumption~\ref{assu:risk-consistency}, $\E_{S_n(f),\o}[d(\hat f_n,f)]\to d(\cF_\T,f)$ for $\l_\cF$-almost every $f$. Since $d$ is uniformly bounded, by the dominated convergence theorem we have $L_n(\cA)\to\E_{\l_\cF}[d(\cF_\T,f)]$. Similarly, for $\cA_{\mathrm{base}}$, we have $L_n(\cA_{\mathrm{base}})\to\E_{\l_\cF}[d(\cF_{\mathrm{base}},f)]$, which is strictly positive by Assumption~\ref{assu:nondegen}. Taking the ratio of these limits proves that $r(\cA)$ exists and equals $r(\cF_\T;\cF,d)$.
\end{proof}

\begin{proof}[Proof of Proposition~\ref{prop:reachable}]
(i) Fix $f\in\cF$, $n\ge1$, a realization of $S_n(f)$, and $\o\in\O_\cA$. By Definition~\ref{def:reachable}, $\hat f_n\in\cF_\cA$, so $d(\hat f_n,f)\;\ge\;\inf_{h\in\cF_{\cA}}d(h,f)=d(\cF_{\cA},f).
$
Taking expectations over $(S_n(f),\o)$ and then over $f\sim\l_\cF$, we obtain (i).

(ii) By the risk consistency of $\cA_{\mathrm{base}}$ and uniform boundedness of $d$, we have  
$L_n(\cA_{\mathrm{base}})\to\E_{\l_\cF}[d(\cF_{\mathrm{base}},f)]>0$, 
where strict positivity follows from Assumption~\ref{assu:nondegen}. (ii) then follows from dividing the bound in (i) by $L_n(\cA_{\mathrm{base}})$ and taking the liminf.

(iii) With a singleton baseline, the baseline procedure always returns $f_{\mathrm{base}}$, so for every $n\ge1$,
$
L_n(\cA_{\mathrm{base}})=\E_{\l_\cF}[d(f_{\mathrm{base}},f)]=\E_{\l_\cF}[d(\cF_{\mathrm{base}},f)]>0,$  again by Assumption~\ref{assu:nondegen}. Dividing the bound in (i) by this common denominator yields (iii).
\end{proof}

\begin{proposition}[Asymptotic Normality and Standard Error]\label{prop:proc-clt}
Fix $n\ge1$, and suppose $\text{Var}\left(D_1-r_n(\cA)B_1\right)>0$. Then, as $M\to\infty$,
\[
\sqrt{M}\left(\widehat r_{n,M}(\cA)-r_n(\cA)\right)
\;\dto\;
N\left(0,\s_n^2\right),
\quad
\s_n^2
:=
\frac{\text{Var}\left(D_1-r_n(\cA)\,B_1\right)}{\left(\E[B_1]\right)^{2}},
\]
and the plug-in estimator
\[
\widehat\s_n^2
:=
\frac{\frac{1}{M-1}\sum_{m=1}^{M}\left(D_m-\widehat r_{n,M}(\cA)\,B_m\right)^{2}}{\left(\frac{1}{M}\sum_{m=1}^{M}B_m\right)^{2}}
\]
is consistent as $M\to\infty$.
\end{proposition}

\begin{proof}[Proof of Proposition~\ref{prop:proc-clt}]
Fix $n$. The simulation design makes the pairs $(D_m,B_m)$ i.i.d.\ across $m$, and uniform boundedness of $d$ ensures finite second moments. Write $\mu_D=\E[D_1]=L_n(\cA)$ and $\mu_B=\E[B_1]=L_n(\cA_{\mathrm{base}})>0$, where positivity is required by Definition~\ref{def:proc-rest}, and let $\bar D_M$ and $\bar B_M$ denote their sample means. The bivariate central limit theorem applies to $(\bar D_M,\bar B_M)$. Since $h(u,v)=u/v$ is continuously differentiable near $(\mu_D,\mu_B)$, with gradient $\mu_B^{-1}(1,-r_n(\cA))^\top$, the delta method gives
\[
\sqrt M\big(\widehat r_{n,M}(\cA)-r_n(\cA)\big)
=\frac{1}{\mu_B\sqrt M}\sum_{m=1}^M\big(D_m-r_n(\cA)B_m\big)+o_p(1)
\;\dto\;N(0,\s_n^2),
\]
where $\s_n^2=\text{Var}(D_1-r_n(\cA)B_1)/\mu_B^2$.

The law of large numbers gives $\bar B_M\pto\mu_B$ and $\widehat r_{n,M}(\cA)\pto\mu_D/\mu_B=r_n(\cA)$. Applying the law of large numbers to $D_m^2$, $D_mB_m$, and $B_m^2$, and using this consistency, yields
\[
\frac{1}{M}\sum_{m=1}^M\big(D_m-\widehat r_{n,M}(\cA)B_m\big)^2
\;\pto\;\E\big[(D_1-r_n(\cA)B_1)^2\big].
\]
The expectation on the right equals $\text{Var}(D_1-r_n(\cA)B_1)$ because $\E[D_1-r_n(\cA)B_1]=0$. Since $M/(M-1)\to1$ and $\bar B_M^2\pto\mu_B^2>0$, the stated variance estimator is consistent.
\end{proof}

\subsection{Gaussian-Process Evaluation Distributions}
\label{sec:gp.details}
Gaussian processes provide a convenient way to specify
distributions over real-valued functions. A Gaussian process $g\sim\mathcal{GP}(m,K)$ is characterized by its mean function $m(x)$ and covariance function $K(x,x')$, such that for any $x_1,\ldots,x_q\in\mathcal X$, 
$\bigl(g(x_1),\ldots,g(x_q)\bigr)^\top
\sim\mathcal N(\boldsymbol m,\boldsymbol K),
m_i=m(x_i),$ and $ K_{ij}=K(x_i,x_j).$  The covariance function governs dependence across inputs and the smoothness of sampled functions. Common choices include the squared exponential, Mat\'ern, $\gamma$-exponential, rational quadratic, and dot-product families \citep[Chapter~4]{williams2006gaussian}.
Our applications use the Mat\'ern-$3/2$ kernel
\[
K_{3/2}(x,x')
=
\sigma^2
\left(1+\frac{\sqrt{3}\,r}{\ell}\right)
\exp\left(-\frac{\sqrt{3}\,r}{\ell}\right),
\quad r=\lVert x-x'\rVert,
\]
where $\sigma^2$ is the variance and $\ell$ is the length scale. The resulting process is once mean-square differentiable.

The evaluation distribution must respect the restrictions defining the eligible set $\mathcal F$. For Gaussian processes, bounds and shape restrictions can be incorporated through functional transformations \citep{snelson2003warped}, integral or spline representations \citep{shively2009bayesian}, or constrained
finite-dimensional approximations \citep{maatouk2017gaussian}. Alternatively, restrictions can be imposed at selected points, as in the virtual derivative approach to monotonicity \citep{riihimaki2010gaussian}, although finite-point
restrictions need not hold throughout the domain.
\citet{swiler2020survey} survey these methods.

In our applications, the certainty-equivalent functions are constructed by adjusting Gaussian-process draws to satisfy bounds and monotonicity on sampled lotteries. For discrete choice, Gaussian processes generate utility components, with the specification ensuring that utility decreases in own price and market shares follow a random-utility model.
Online Appendix~\ref{sec:detail.example1} and  \ref{sec:detail.example2} provide the implementation details.

\subsection{Sensitivity to the Evaluation Distribution}

Restrictiveness depends on the evaluation distribution over eligible prediction rules. We examine sensitivity to this choice by reweighting the GP-based distribution within density-ratio bounds and, for the certainty-equivalent application, by generating functions from a prior adapted from a tabular foundation model. At training sample size \(n=100\), each probability-weighting parameter continues to contribute more flexibility than utility curvature under both exercises. Completion, full CPT, and Agent also retain their ordering from most to least restrictive. Reweighting also preserves the lower restrictiveness of NL and MXL relative to MNL for demand shares.

\subsubsection{Reweighting the GP-Based Evaluation Distribution}
\label{subsec:app1-cpt-sensitivity}

We follow the density-ratio approach to Bayesian prior sensitivity of \citet{WassermanKadane1992} and \citet{GewekePetrella1998}. Let $\lambda_0$ denote the baseline GP-based distribution of eligible functions. We consider distributions with density ratio $w=d\lambda/d\lambda_0$ satisfying $E_{\lambda_0}[w]=1$ and $\Gamma^{-1}\leq w\leq\Gamma$. The corresponding weights on the simulated functions satisfy
\begin{equation}
\mathcal W_{\Gamma,M}=\left\{w\in\mathbb R^M:
\frac{1}{M}\sum_{m=1}^M w_m=1,\quad
\Gamma^{-1}\leq w_m\leq\Gamma\ \text{for every }m\right\},
\quad\Gamma\geq1.
\label{eq:app1-cpt-sensitivity-main}
\end{equation}
For example, $\Gamma=2$ permits each function to receive between one-half and twice its baseline probability, subject to normalization.

For a given training sample size, let $D_{jmk}$ denote the evaluation MSE of procedure $j$ for function $m$ and training sample $k$, and write $\widetilde D_{jm}=K_j^{-1}\sum_{k=1}^{K_j}D_{jmk}$. The restrictiveness estimator requires only one training sample per function, whereas sensitivity analysis uses repeated samples to approximate conditional mean losses (Online Appendix~\ref{subsec:app1-sensitivity-online}). Let $\widetilde B_m$ denote the corresponding average loss of the baseline rule. For any admissible weights, a paired comparison of procedures $a$ and $b$ is
\begin{equation}
\widehat\Delta_{a,b}(w)=\widehat r_a(w)-\widehat r_b(w)
=\frac{\sum_m w_m(\widetilde D_{am}-\widetilde D_{bm})}{\sum_m w_m\widetilde B_m}.
\label{eq:app1-weighted-comparison}
\end{equation}
The lower and upper bounds are then the smallest and largest comparisons attainable over the admissible weights:
$\underline\Delta_{a,b}(\Gamma)
=\min_{w\in\mathcal W_{\Gamma,M}}\widehat\Delta_{a,b}(w),
\overline\Delta_{a,b}(\Gamma)
=\max_{w\in\mathcal W_{\Gamma,M}}\widehat\Delta_{a,b}(w).$
A positive lower bound preserves the ordering under every admissible reweighting. The critical value $\Gamma^*$ is the smallest $\Gamma$ for which the lower bound is non-positive.

We apply the paired-comparison formula in equation~(\ref{eq:app1-weighted-comparison}) to the CPT parameter contributions in Section~\ref{subsec:cpt-blackbox}.
For $q\in\{\alpha,\gamma,\delta\}$, let $r_{-q}$ denote the restrictiveness of CPT submodel with parameter $q$ fixed at one, and let $r_{\mathrm{full}}$ denote that of full CPT. The contribution is defined as $C_q=r_{-q}-r_{\mathrm{full}}$. Under reweighting, we estimate $C_q$ by $\widehat\Delta_{-q,\mathrm{full}}(w)$ and $C_q-C_s$ directly by $\widehat\Delta_{-q,-s}(w)$. 

Figure~\ref{fig:app1-cpt-sensitivity} shows that both probability-weighting parameters contribute more flexibility than utility curvature under all reweightings through $\Gamma=3$. The lower bounds on $C_\delta-C_\alpha$ and $C_\gamma-C_\alpha$ remain positive at $0.180$ and $0.022$, respectively. However, the ordering between the probability-weighting parameters can reverse once $\Gamma$ exceeds $2.409$. Thus, the greater contribution of probability weighting is robust, while the relative contributions of its two parameters are more sensitive.

\begin{figure}[t!]
\caption{Sensitivity of differences in CPT parameter contributions}
\label{fig:app1-cpt-sensitivity}
\centering
\includegraphics[width=\textwidth]{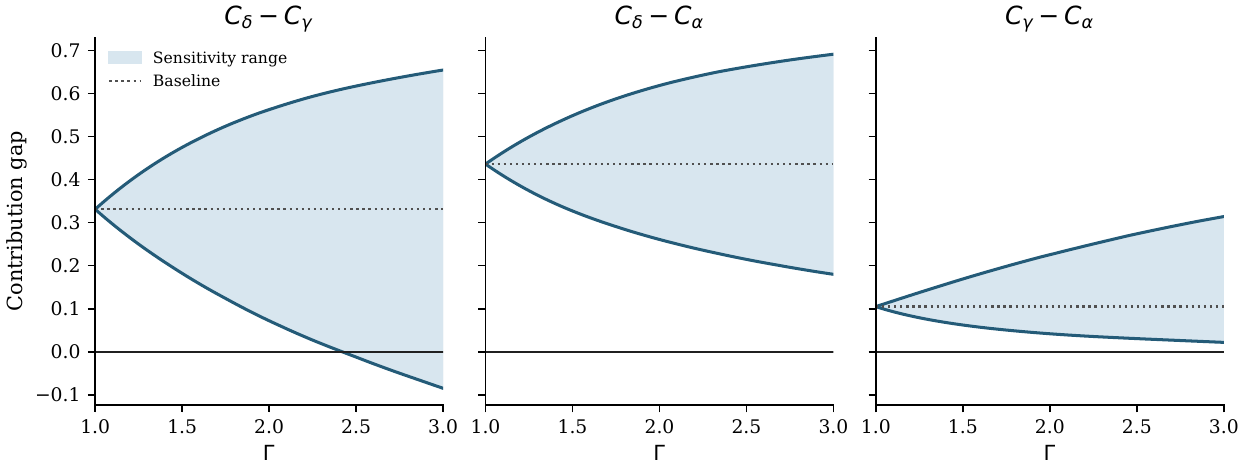}
{\footnotesize\raggedright\emph{Notes:} Bounds in equation~\eqref{eq:app1-cpt-sensitivity-main}. Dotted lines mark estimates under uniform weights; the horizontal line marks zero. A positive lower bound preserves the ordering for every admissible reweighting.\par}
\end{figure}

The  sensitivity analysis supports several central findings in both applications under reweighting through $\Gamma=3$ at $n=100$. In the certainty-equivalent application, Completion remains more restrictive than full CPT, which remains more restrictive than Agent; GP regression and the TFMs also remain less restrictive than the fitted CPT specifications. In the demand application, the distinction between fitting shares and recovering elasticities persists: NL and MXL retain better share fit but less accurate cross-price elasticities than MNL, and NL also remains less accurate for own-price elasticities. Agent and TabPFN retain their lower restrictiveness than the structural procedures for demand shares. Some other elasticity comparisons and comparisons at smaller training sample sizes are more sensitive, so robustness depends on both the discrepancy and sample size. Online Appendices~\ref{subsec:app1-sensitivity-online} and~\ref{subsec:app2-sensitivity-online} report the complete set of robustness results.

\FloatBarrier

\subsubsection{Evaluation Distribution Based on a TFM Prior}
\label{subsec:cpt-tfm}

As a second check, we adapt the synthetic-data prior used to pretrain TabICL v2 \citep{qu2026tabicl} to generate new certainty-equivalent functions. This prior combines randomly drawn neural networks, trees, and other functions, allowing a different range of shapes from the Mat\'ern construction. We impose prize bounds and monotonicity at the sampled lotteries, and match each adjusted latent draw's mean and standard deviation to those of the corresponding GP draw.

Table~\ref{tab:app1-tfm} compares restrictiveness under the two evaluation distributions at $n=100$, using $M=100$ replications. Each procedure retains the same fitting settings, prompt, and computation budget. Changes across distributions are paired by replication. Online Appendix~\ref{subsec:app1-tfm-prior} provides the construction, full sample-size profiles, and additional sensitivity checks.

\begin{table}[t!]
\caption{GP-based vs. TFM-based evaluation distributions ($n=100$)}
\label{tab:app1-tfm}
\centering
\small\singlespacing\setlength{\tabcolsep}{4pt}
\begin{tabular}{@{}lccc@{}}\toprule
 & GP design & TFM-based design & Paired change \\ \midrule
  CPT$(\alpha,\gamma,\delta)$ & 0.123 (0.017) & 0.126 (0.016) & 0.003 (0.011) \\
  CPT$(\gamma,\delta)$ & 0.148 (0.020) & 0.161 (0.022) & 0.013 (0.014) \\
  CPT$(\alpha,\delta)$ & 0.253 (0.033) & 0.307 (0.034) & 0.054 (0.019) \\
  CPT$(\delta)$ & 0.269 (0.034) & 0.331 (0.037) & 0.062 (0.020) \\
  CPT$(\alpha,\gamma)$ & 0.584 (0.032) & 0.498 (0.033) & $-$0.087 (0.013) \\
  CPT$(\gamma)$ & 0.737 (0.013) & 0.682 (0.018) & $-$0.055 (0.014) \\
  CPT$(\alpha)$ & 0.852 (0.030) & 0.825 (0.029) & $-$0.027 (0.012) \\
\midrule
  GP regression & 0.017 (0.002) & 0.031 (0.005) & 0.014 (0.004) \\
  TabPFN v2 & 0.014 (0.002) & 0.019 (0.004) & 0.005 (0.003) \\
  TabICL v2 & 0.015 (0.002) & 0.018 (0.003) & 0.003 (0.002) \\
  AutoML (FLAML) & 0.069 (0.009) & 0.059 (0.008) & $-$0.010 (0.007) \\
\midrule
  Completion & 0.328 (0.053) & 0.266 (0.036) & $-$0.062 (0.046) \\
  Completion (high thinking) & 0.087 (0.019) & 0.069 (0.014) & $-$0.018 (0.020) \\
  Agent & 0.025 (0.004) & 0.025 (0.005) & 0.000 (0.005) \\
\bottomrule
\end{tabular}
\par\smallskip
{\footnotesize\raggedright\emph{Notes:} $M=100$ replications. Online Appendix Table~\ref{tab:app1-tfm-profiles} reports the full profiles.\par}
\end{table}

The main comparisons are similar under the two distributions. Completion remains more restrictive than full CPT ($0.266$ versus $0.126$), while Agent is much less restrictive ($0.025$). Allowing high thinking again puts Completion between full CPT and Agent. The two-parameter $\mathrm{CPT}(\alpha,\gamma)$ remains substantially more restrictive, at $0.498$, while retaining its high empirical completeness. Completeness is unchanged because it is measured on the same experimental data. Across the fourteen procedures, the rank correlation between the two evaluation distribution designs is $0.978$. The CPT parameter comparisons also persist, although their magnitudes change. Both probability-weighting parameters contribute more flexibility than utility curvature, but the gap between the contributions of $\delta$ and $\gamma$ narrows from $0.331$ to $0.191$. Hence, the greater contribution of probability weighting is robust across these distributions, while the relative contributions of its two parameters are more sensitive.

The TFM-based distribution also separates GP regression from the TFMs more clearly. GP regression's restrictiveness rises by $0.014$ (paired s.e.\ $0.004$), while TabPFN and TabICL remain close to zero. Under the TFM-based design, they are less restrictive than GP regression by $0.012$ and $0.013$, respectively (with s.e.\ $0.003$).

\subsection{Prompt Design and Parameter Estimation}
\label{sec:app1-prompt-estimation}

Table~\ref{tab:app1-prompt-estimation} reports how Completion's use of the supplied observations for numerical estimation varies across prompts, alongside its restrictiveness at $n=100$. We classify a program as specifying numerical estimation if it calls a fitting routine that uses the supplied observations to estimate coefficients or select tuning parameters, including through optimization, least squares, or solving linear systems. Programs relying solely on fixed coefficients are excluded. Among programs specifying estimation, we distinguish those using all 100 observations from those using a subset.

\begin{table}[!htbp]
\centering
\caption{Prompt design, parameter estimation, and restrictiveness: Completion}
\label{tab:app1-prompt-estimation}
\small\singlespacing\setlength{\tabcolsep}{4pt}
\begin{tabular}{@{}lcccr@{}}\toprule
 & \multicolumn{3}{c}{\shortstack{Programs specifying\\numerical estimation}} & Restrictiveness \\
\cmidrule(lr){2-4}
Prompt & Full sample & Subset & Total & \\
\midrule
Neutral & 36 & 36 & 72 & 0.328 (0.053) \\
Econometrician role & 15 & 49 & 64 & 0.277 (0.037) \\
Economic-model instruction & 23 & 52 & 75 & 0.244 (0.040) \\
Data-scientist role & 72 & 27 & 99 & 0.183 (0.031) \\
Machine-learning instruction & 86 & 14 & 100 & 0.138 (0.018) \\
\bottomrule
\end{tabular}
\par\smallskip
{\footnotesize\raggedright\emph{Notes:} The estimation columns count programs among 100 synthetic draws at $n=100$. We identify the supplied training observations embedded in the returned code by matching their values to the training sample at six-decimal precision. Two Completion programs under the data-scientist role contain estimation code but fail at initialization because of coding errors; all other Completion programs counted as estimation execute successfully.\par}
\end{table}

The returned programs illustrate how prompts affect the use of supplied observations. For example, one program under the neutral prompt searches a parameter grid using only 30 observations, whereas a program under the machine-learning instruction estimates an RDU model with power utility and Prelec weighting using all 100 observations. The latter retains an economic preference specification despite the machine-learning instruction. 
\setcounter{section}{0}
\renewcommand{\thesection}{S.\arabic{section}}
\setcounter{subsection}{0}
\renewcommand{\thesubsection}{\thesection.\arabic{subsection}}

\clearpage
\begin{center}
{\Large\bfseries Online Appendix \par}
\vspace{1em}
\end{center}

\section{Details for Application 1 (Certainty Equivalents)}
\label{sec:detail.example1}

This section provides the implementation details and supplementary results for the certainty equivalent application.

\subsection{Evaluation and Inference}
\label{subsec:evaluation.example1}
For restrictiveness, we construct certainty-equivalent values jointly on training and evaluation lotteries over the support $\mathcal D_B$ in equation~\eqref{eq:cpt-rectangular-support}. Each point $x=(\bar z,\underline z,p)$ specifies a lottery. For replication $m$ and training-sample size $n\in\{50,100,200\}$, we generate training and evaluation locations using separate random scrambles of a three-dimensional Sobol sequence. This is a randomized quasi-Monte Carlo construction designed to spread points evenly over $[0,1]^3$. We rescale the coordinates to the prize and probability ranges in $\mathcal D_B$, discard points with $\bar z<\underline z$, and retain $n$ admissible points for training and 1,000 for evaluation.

Let $\mathcal X_m=\mathcal X_m^{\rm train}\cup\mathcal X_m^{\rm eval}$ denote the combined set of $n+1{,}000$ lottery locations, suppressing its dependence on $n$. Scramble seeds vary across replications. Within a replication, the training locations are nested across sample sizes and the evaluation locations are common, while the pseudo-truth is constructed separately for each $n$. We draw a latent Gaussian vector jointly on $\mathcal X_m$ as
$$
g_m\sim N\!\left(0,K_m+10^{-4}I\right),
\quad
(K_m)_{ij}
=10\left(1+\sqrt 3\|x_i-x_j\|\right)
\exp\!\left(-\sqrt 3\|x_i-x_j\|\right).
$$
For pairs of points ordered coordinatewise in $\mathcal X_m$, we repeatedly average latent values that violate the chosen monotonicity direction, and  stop when the Euclidean change over a sweep is below $10^{-6}$ or after 50 sweeps. Applying this operation in the increasing and decreasing directions gives $g_{m,+}$ and $g_{m,-}$. Let $e_+$ and $e_-$ denote their squared deviations from the original latent vector. We combine the adjusted vectors as
\begin{equation*}
\widetilde g_m=(1-w_m)g_{m,+}-w_mg_{m,-},\quad
w_m=\frac{e_+}{e_++e_-}.
\end{equation*}
The certainty-equivalent pseudo-truth is $$
f_m(x)=\underline z+(\bar z-\underline z)(1+\exp\{-\widetilde g_m(x)\})^{-1},
$$
which enforces $\underline z\leq f_m(x)\leq\bar z$ at every realized point. The pseudo-truth is realized jointly on this finite set of points. Only the training points and their labels are supplied to procedures, and the fitted rules are scored on the evaluation points.

Inference on restrictiveness follows Proposition~\ref{prop:proc-clt}. For draw $m$, write $D_{jm}$ for the MSE of procedure $j$ on the evaluation points and $B_m$ for the MSE of the expected-value rule on those same points. The estimator and its Monte Carlo SE are
\begin{equation}
\widehat r_j=\frac{\overline D_j}{\overline B},\quad
\widehat\psi_{jm}=\frac{D_{jm}-\widehat r_jB_m}{\overline B},\quad
\widehat{\operatorname{se}}(\widehat r_j)
=\sqrt{\frac{\widehat{\operatorname{Var}}(\widehat\psi_{jm})}{M}}.
\label{eq:app1-restrictiveness-se}
\end{equation}
For paired contrasts, we use the sample variance of $\widehat\psi_{jm}-\widehat\psi_{km}$ on common draws.

For the inference on completeness, let~\mbox{$G_i=\ell_{{\rm EV},i}-\ell_{j,i}$} and~\mbox{$H_i=\ell_{{\rm EV},i}-\ell_{{\rm flex},i}$} denote the held-out squared-error gains, where $\ell_{{\rm EV},i}$ is the expected-value baseline's loss and $\ell_{{\rm flex},i}$ is TabPFN's loss. The pooled estimator is
\begin{equation}
\widehat\kappa_j=\frac{\overline G}{\overline H},\quad
\widehat\phi_i=
\frac{(G_i-\overline G)-\widehat\kappa_j(H_i-\overline H)}{\overline H}.
\label{eq:app1-completeness-if}
\end{equation}
Standard errors cluster these response-level influence values by participant: if $S=179$, $N=4{,}459$, and $\widetilde\phi_s=(S/N)\sum_{i\in s}\widehat\phi_i$, then
\begin{equation}
\widehat{\operatorname{se}}(\widehat\kappa_j)
=\left\{\frac{1}{S(S-1)}
\sum_{s=1}^S(\widetilde\phi_s-\overline{\widetilde\phi})^2\right\}^{1/2}.
\end{equation}

\subsection{Fitting Procedures}
\label{subsec:procedures.example1}
We fit each of the seven distinct CPT classes by nonlinear least squares on the supplied labels, with $\alpha,\gamma\in[0.01,1]$ and $\delta\in[0.01,30]$. 
For binary lotteries, the disappointment-aversion (DA) model of \citet{Gul1991} admits the RDU representation
\[
    u(z)=z^\alpha,\quad
    w(p)=\frac{p}{1+(1-p)\eta},\quad \eta>-1.
\]
We denote by $\mathrm{DA}(\alpha,\eta)$ the specification with
both parameters free and by $\mathrm{DA}(\eta)$ the restriction
$\alpha=1$. Setting $\gamma=1$ in CPT yields the same probability weighting function under $\delta=(1+\eta)^{-1}$. The implemented range $\delta\in[0.01,30]$ corresponds to $\eta\in[-29/30,99]$, allowing both disappointment aversion
($\eta\geq0$) and elation-loving preferences ($\eta<0$). Thus, $\mathrm{DA}(\alpha,\eta)$ and $\mathrm{DA}(\eta)$ are equivalent to $\mathrm{CPT}(\alpha,\delta)$ and $\mathrm{CPT}(\delta)$, respectively, and are reported under their CPT labels.

For GP regression, we use the Mat\'ern-$3/2$ kernel with variance $10$ and length scale $1$ as a noise-free regressor and clip predictions to $[\underline z,\bar z]$. TabPFN v2 and TabICL v2 use fixed pretrained weights. AutoML uses FLAML, an open-source Python library that automatically selects
machine-learning models and tunes their hyperparameters, with a budget of 100 search iterations. AIDE uses the language model gpt-oss-120b to generate candidate machine-learning scripts and evaluates them using cross-validation, with its tree search limited to eight steps.

Completion uses \texttt{claude-sonnet-5-5} with \texttt{between\_tools} thinking, effort set to \texttt{low}, and at most 32,000 output tokens per response. Completion with high thinking and Agent use the same model and output limit with adaptive thinking and effort set to \texttt{high}. Returned programs may use numpy, math, and the Python standard library.

Both Completion configurations and Agent receive the task and prediction requirements in Section~\ref{subsec:cpt-procedures}. These are common to the restrictiveness and completeness exercises, with the supplied data changing across exercises. Each Completion configuration receives the task once, without tools or subsequent repair. Agent can execute Python with access to the supplied data and the libraries numpy, scipy, pandas, and scikit-learn. Each call starts a new interpreter, with access to data and code saved by earlier calls. Agent has a budget of 20 tool rounds, 120 seconds per call, and two hours in total; returned tool output is limited to 8,000 characters per call. Agent cannot access other data, use the network, or start other processes. If it reaches its tool or time budget, it is asked for a final answer without further tool use.

We extract the last Python code block from the final response and execute it in a new interpreter, allowing 600 seconds for the returned predictor. For all LLM procedures, the returned program may retain the supplied training data and estimate parameters when initialized or first queried. In the empirical lottery fits, Completion estimates parameters when the returned programs are run, whereas Agent estimates them during tool use and reports fitted constants; both minimize squared certainty-equivalent errors. Some synthetic Agent programs also calculate regression or kernel coefficients when their final prediction programs are run.

Valid predictions are clipped to the prize interval, and invalid predictions receive the expected-value rule. A program that uses a prohibited library, fails to execute, or supplies no code receives the expected-value rule at every evaluation point.

Table~\ref{tab:app1-main} reports the estimates underlying the
main comparison of economic models and procedures in Figure~\ref{fig:app1-plane}.

\begin{table}[t!]
  \caption{Restrictiveness and completeness at $n=100$:
    certainty equivalents}
  \label{tab:app1-main}
  \centering
  \small
  \setlength{\tabcolsep}{3pt}
  \begin{tabular}{@{}lcc@{}}\hline
   & Restrictiveness & Completeness \\ \hline
  \multicolumn{3}{@{}l}{\emph{Economic model classes
    (least squares)}} \\
  CPT$(\alpha,\gamma,\delta)$
    & 0.123 (0.017) & 0.976 (0.009) \\
  CPT$(\gamma,\delta)$
    & 0.148 (0.020) & 0.951 (0.015) \\
  CPT$(\alpha,\delta)$ $=$ DA$(\alpha,\eta)$
    & 0.253 (0.033) & 0.266 (0.058) \\
  CPT$(\delta)$ $=$ DA$(\eta)$
    & 0.269 (0.034) & 0.266 (0.058) \\
  CPT$(\alpha,\gamma)$
    & 0.584 (0.032) & 0.976 (0.009) \\
  CPT$(\gamma)$
    & 0.737 (0.013) & 0.685 (0.059) \\
  CPT$(\alpha)$
    & 0.852 (0.030) & 0.249 (0.055) \\ \hline
  \multicolumn{3}{@{}l}{\emph{Numerical black boxes}} \\
  GP regression
    & 0.017 (0.002) & 1.002 (0.002) \\
  TabPFN v2
    & 0.014 (0.002) & 1 (reference) \\
  TabICL v2
    & 0.015 (0.002) & 1.002 (0.001) \\
  AutoML (FLAML)
    & 0.069 (0.009) & 0.922 (0.012) \\
  AIDE (LLM tree search)
    & 0.044 (0.010) & 0.742 (0.034) \\ \hline
  \multicolumn{3}{@{}l}{\emph{Claude Sonnet 5.5,
    shared neutral task}} \\
  Completion
    & 0.328 (0.053) & 0.988 (0.007) \\
  Agent
    & 0.025 (0.004) & 0.989 (0.007) \\ \hline
  \end{tabular}
  \par\smallskip
  {\footnotesize\raggedright
  \emph{Notes:} Standard errors are reported in parentheses. \par}
\end{table}

Table~\ref{tab:app1-thinking} compares Completion in the
main comparison, Completion with high thinking, and Agent at
$n=100$. Completion with high thinking is less restrictive than the main configuration but more restrictive than Agent. The three attain essentially identical completeness for participants' lottery responses, so the main difference lies in restrictiveness.

\begin{table}[t!]
  \caption{Completion configurations and Agent: certainty equivalents at $n=100$}
  \label{tab:app1-thinking}
  \centering
  \small
  \setlength{\tabcolsep}{6pt}%
  \begin{tabular}{@{}lcc@{}}\hline
    & Restrictiveness & Completeness \\ \hline
    Completion (main)
      & 0.328 (0.053) & 0.988 (0.007) \\
    Completion (high thinking)
      & 0.087 (0.019) & 0.988 (0.004) \\
    Agent
      & 0.025 (0.004) & 0.989 (0.007) \\ \hline
  \end{tabular}
  \par\smallskip
  {\footnotesize\raggedright
  \emph{Notes:} Standard errors are reported in parentheses.
  \par}
\end{table}

\subsection{Specification Choices}
\label{subsec:specifications.example1}
Figure~\ref{fig:app1-neutral-model-choices} in the main paper reports the specifications selected in every synthetic exercise at $n=100$ and every empirical fit of Completion and Agent. Here we classify the prediction rules using the source code, saved estimates, and execution records.

For restrictiveness, we classify fits to simulated certainty equivalents by the form of the prediction rule. The \emph{RDU} category consists of the rules in equation~\eqref{eq:app1-rdu}, with a common utility function and probability-weighting function across lotteries. The CPT and DA specifications and the alternative utility and weighting functions described below are parametric cases of this representation.

The other categories do not impose a single RDU representation. \emph{Other parametric CE} comprises direct nonlinear certainty-equivalent surfaces, logistic specifications for the normalized certainty equivalent, outcome-dependent weighting or curvature, and averages of certainty equivalents. \emph{Feature regression} predicts the certainty equivalent, its normalized position in the prize interval, or an expected-value correction using a finite feature expansion linear in its coefficients, including polynomial and ridge regressions. \emph{GP/kernel} uses Gaussian-process or kernel prediction weights and may include an additive finite-feature mean. \emph{Hybrid} combines distinct broad families.

For completeness, all rules fitted to the experimental lottery responses use RDU. The panels distinguish \emph{power} utility, $u(z)=z^a$, and \emph{shifted power}, $u(z)=(z+c)^a$ up to an additive constant. Their probability weights are either two-parameter Prelec, $w(p)=\exp[-b(-\log p)^g]$, or two-parameter odds, $w(p)=d p^g/[d p^g+(1-p)^g]$; parameters are specific to each specification. These specifications concern gains and contain no loss-aversion parameter. Completion uses Prelec weights in all ten empirical fits. Agent uses Prelec and odds weights in nine and one fits, respectively.

\emph{Unresolved model} denotes a specification that cannot be determined from the available records. \emph{Execution failure} denotes a program that fails to initialize, omits the required prediction function, or cannot produce predictions at any evaluation point.

\FloatBarrier

\subsection{Model-Version Comparisons}\label{subsec:app1-version-comparison}

For the model-version comparison, Completion uses disabled thinking with low effort for Sonnet 5 and Haiku 5.5, between-tools thinking with low effort for Sonnet 5.5, and reasoning effort none for Luna and Sol. Completion with Opus 5.5 uses adaptive thinking at low effort because this model does not support disabling thinking. Seven of its 100 synthetic responses contain thinking; none of its ten empirical fits does. Agent uses the high-effort setting of each language model. All runs use the common task, permitted libraries, prediction requirements, and maximum budgets. Table~\ref{tab:app1-model-versions} reports the neutral-prompt results for Completion and Agent across model versions at $n=100$. 

\begin{table}[t!]
\centering
\caption{Model versions under the neutral prompt: restrictiveness at $n=100$}
\label{tab:app1-model-versions}
\small\setlength{\tabcolsep}{4pt}
\begin{tabular}{@{}lccc@{}}\toprule
 & Restrictiveness & Completeness & Failures \\
\midrule
\multicolumn{4}{@{}l}{\emph{Completion}} \\
Claude Sonnet 5.5 & 0.328 & 0.988 & 0 \\
 & (0.053) & (0.007) & \\
Claude Sonnet 5 & 0.450 & 0.248 & 6 \\
 & (0.049) & (0.033) & \\
Claude Haiku 5.5 & 0.167 & 0.989 & 3 \\
 & (0.038) & (0.011) & \\
GPT-6 Luna & 0.379 & $-$1.386 & 12 \\
 & (0.063) & (0.330) & \\
GPT-6 Sol & 0.232 & 0.887 & 5 \\
 & (0.067) & (0.025) & \\
Claude Opus 5.5$^\dagger$ & 0.112 & 0.986 & 0 \\
 & (0.022) & (0.008) & \\
\midrule
\multicolumn{4}{@{}l}{\emph{Agent}} \\
Claude Sonnet 5.5 & 0.025 & 0.989 & 0 \\
 & (0.004) & (0.007) & \\
Claude Sonnet 5 & 0.178 & 0.985 & 1 \\
 & (0.041) & (0.009) & \\
Claude Haiku 5.5 & 0.026 & 0.984 & 0 \\
 & (0.006) & (0.006) & \\
GPT-6 Luna & 0.163 & 1.002 & 10 \\
 & (0.041) & (0.002) & \\
GPT-6 Sol & 0.065 & 1.002 & 3 \\
 & (0.023) & (0.002) & \\
Claude Opus 5.5$^\dagger$ & 0.016 & 1.002 & 0 \\
 & (0.002) & (0.005) & \\
\midrule
\end{tabular}
\par\smallskip
{\footnotesize\raggedright\emph{Notes:} Standard errors are reported in parentheses. Failures count programs replaced by the expected-value rule at every evaluation point among the 100 draws. $^\dagger$Completion with Opus 5.5 uses adaptive thinking at low effort. Eight completeness folds for Completion with Sonnet 5 fail; one fold for Completion with GPT-6 Luna fails. No other completeness fold fails.\par}
\end{table}

To examine how Completion uses the supplied observations across model versions, we record whether its returned programs perform numerical estimation, including optimization, least squares, solving linear systems, and selection of tuning parameters. We count only calculations used to produce predictions and record smoothers with fixed tuning parameters separately. Estimation counts include programs that subsequently fail. Table~\ref{tab:app1-completion-estimation} reports the corresponding frequencies.

\begin{table}[t!]
\centering
\caption{Numerical estimation across language model versions (Completion)}
\label{tab:app1-completion-estimation}
\small\setlength{\tabcolsep}{5pt}
\begin{tabular}{@{}lccc@{}}\toprule
Model & \shortstack{Programs with\\estimation} & \shortstack{Estimation programs\\that fail} & \shortstack{All program\\failures} \\
\midrule
Claude Sonnet 5.5 & 72 & 0 & 0 \\
Claude Sonnet 5 & 75 & 6 & 6 \\
Claude Haiku 5.5 & 99 & 2 & 3 \\
GPT-6 Luna & 95 & 7 & 12 \\
GPT-6 Sol & 91 & 5 & 5 \\
Claude Opus 5.5$^\dagger$ & 100 & 0 & 0 \\
\bottomrule\end{tabular}
\par\smallskip
{\footnotesize\raggedright\emph{Notes:} Counts among the 100 synthetic draws at $n=100$, under the neutral task prompt.  Two Completion programs with GPT-6 Sol use fixed-bandwidth kernel smoothing and are recorded separately from parameter estimation. The failure columns report programs replaced by the expected-value rule. $^\dagger$Completion with Opus 5.5 uses adaptive thinking at low effort.\par}
\end{table}

For Completion with Sonnet 5.5, restrictiveness is lower for programs that estimate parameters from the supplied observations. The 72 programs with numerical estimation have restrictiveness $0.206$, compared with $0.659$ for the 28 programs with fixed coefficients, using the ratio estimator within each group. This comparison reflects both estimation and the specifications selected by Completion. With Sonnet 5, Completion uses estimation code in 75 draws, often within a parsimonious utility and probability-weighting specification.
\par

The six failed Completion programs with Sonnet 5 in the restrictiveness exercise require \texttt{scipy.optimize}, which is excluded from returned predictors. The eight failed empirical folds also use prohibited imports.

Execution errors limit Agent's opportunities to estimate and compare specifications. At $n=100$, Agent with Sonnet 5 exhausts the round budget in 44 draws, compared with none with Sonnet 5.5. With Sonnet 5, Agent has 400 failed calls among 1,547 executions, including 380 undefined-variable errors; with Sonnet 5.5, Agent has 44 failed calls among 903 executions. Most Agent programs with Sonnet 5 still initialize. Agent with GPT-6 Luna has ten final-program failures, eight caused by errors in reading the saved data. Agent with GPT-6 Sol has three, all caused by errors in reading or decoding the saved data.

\subsection{Prompt Comparisons}
\label{subsec:app1-prompt-design}

Table~\ref{tab:app1-prompt-designs} reports the results for the five prompt designs summarized in Table~\ref{tab:app1-prompt-sentences}. All comparisons use 100 simulated draws at $n=100$ and the same ten empirical folds. 

\begin{table}[htbp!]
\centering
\caption{Prompt designs on Claude Sonnet 5.5: restrictiveness at $n=100$}
\label{tab:app1-prompt-designs}
\small\setlength{\tabcolsep}{4pt}
\begin{tabular}{@{}lccc@{}}\toprule
 & Restrictiveness & Completeness & Failures \\
\midrule
\multicolumn{4}{@{}l}{\emph{Completion}} \\
Neutral & 0.328 & 0.988 & 0 \\
 & (0.053) & (0.007) & \\
Econometrician role & 0.277 & 0.988 & 0 \\
 & (0.037) & (0.007) & \\
Economic-model instruction & 0.244 & 0.988 & 0 \\
 & (0.040) & (0.007) & \\
Data-scientist role & 0.183 & 0.990 & 2 \\
 & (0.031) & (0.007) & \\
Machine-learning instruction & 0.138 & 0.988 & 0 \\
 & (0.018) & (0.008) & \\
\midrule
\multicolumn{4}{@{}l}{\emph{Agent}} \\
Neutral & 0.025 & 0.989 & 0 \\
 & (0.004) & (0.007) & \\
Econometrician role & 0.024 & 0.994 & 0 \\
 & (0.003) & (0.007) & \\
Economic-model instruction & 0.043 & 0.990 & 0 \\
 & (0.006) & (0.007) & \\
Data-scientist role & 0.030 & 0.985 & 0 \\
 & (0.008) & (0.007) & \\
Machine-learning instruction & 0.025 & 0.985 & 0 \\
 & (0.004) & (0.007) & \\
\midrule
\end{tabular}
\par\smallskip
{\footnotesize\raggedright\emph{Notes:} Standard errors are reported in parentheses. Failures count programs replaced by the expected-value rule at every evaluation point among the 100 draws. The economic-model instruction accompanies the econometrician role; the machine-learning instruction accompanies the data-scientist role. No completeness fold fails.\par}
\end{table}

\FloatBarrier

\begingroup
\subsection{Reweighting the Evaluation Distribution}
\label{subsec:app1-sensitivity-online}

For the GP-based evaluation distribution, we report each comparison separately at $n=50,100,200$, using $M=100$ certainty-equivalent functions at each sample size.

\paragraph{Conditional losses and within-function averaging.}
Ordinary restrictiveness averages over functions and training samples.
Under the conditions of Proposition~\ref{prop:proc-clt}, one independently generated training sample per independently drawn function therefore suffices for consistent estimation of the mean losses as $M$ increases.
Sensitivity analysis instead optimizes over reweightings of functions while holding fixed the conditional distribution of training samples and procedure randomness. Let $\mathcal W_\Gamma$ be the class of measurable weights satisfying $\mathbb E_{\lambda_0}[w(f)]=1$ and $\Gamma^{-1}\le w(f)\le\Gamma$ almost surely. The conditional losses, population comparisons, and sensitivity bounds are
\begin{equation}
\begin{aligned}
\ell_{j,n}(f)
&=\mathbb E_{S_n,\omega}\!\left[
d\bigl(\mathcal A_{j,n}(S_n(f),\omega),f\bigr)\mid f
\right],\\
\Delta_{a,b,n}(w)
&=\frac{\mathbb E_{\lambda_0}\!\left[
w(f)\{\ell_{a,n}(f)-\ell_{b,n}(f)\}
\right]}
{\mathbb E_{\lambda_0}[w(f)\ell_{0,n}(f)]},\\
\underline\Delta_{a,b,n}(\Gamma)
&=\inf_{w\in\mathcal W_\Gamma}\Delta_{a,b,n}(w),\\
\overline\Delta_{a,b,n}(\Gamma)
&=\sup_{w\in\mathcal W_\Gamma}\Delta_{a,b,n}(w),
\end{aligned}
\label{eq:sensitivity-conditional-risk}
\end{equation}
where $\ell_{0,n}$ is the conditional baseline loss.
These bounds concern conditional mean losses, not individual training realizations. Optimizing a separate weight for each simulated function using noisy losses can systematically favor unusually high or low realizations, so increasing $M$ alone need not eliminate the resulting error in the bounds. Averaging losses over repeated training samples within each function reduces this error before the weights are optimized.

Each pseudo-truth is specified only by its certainty-equivalent values on the finite set $\mathcal X_m=\mathcal X_m^{\rm train}\cup\mathcal X_m^{\rm eval}$. To label additional training samples at new lottery locations, we interpolate these simulated certainty equivalents. For each lottery with prizes $\bar z_i$ and $\underline z_i$ and certainty equivalent $y_i$, we normalize its value to $u_i=(y_i-\underline z_i)/(\bar z_i-\underline z_i)$ and take the logit, $g_i=\log\{u_i/(1-u_i)\}$. Using the Mat\'ern-$3/2$ kernel specified in Online Appendix~\ref{subsec:evaluation.example1}, we obtain coefficients $a$ by solving \[ Ka=g-\bar g\mathbf 1, \qquad K_{ij}=k(x_i,x_j), \] where $\bar g$ is the average of the $g_i$. At a new lottery $x$, we evaluate $\bar g+\sum_i a_i k(x,x_i)$, apply the inverse logit, and rescale to its prize interval. Independent scrambled-Sobol training samples are labelled by interpolation, while evaluation MSE is calculated using the simulated certainty equivalents on the fixed points $\mathcal X_m^{\rm eval}$.

The CPT specifications, GP regression, TabPFN, TabICL, and AutoML use twenty training samples per function. Completion, Agent, and AIDE use one, because each additional sample requires new LLM calls and program generation, making repetition substantially more costly.

Throughout this section and Online Appendix~\ref{subsec:app2-sensitivity-online}, Panel A reports the ratio of paired mean MSE differences to the common mean baseline MSE under uniform weights; Panel B reports its sensitivity interval at $\Gamma=2$, and Panel C reports the critical value $\Gamma^*$. The notation $>3$ denotes a positive lower bound through $\Gamma=3$; $\Gamma^*=1$ means that the displayed difference is already zero or negative under uniform weights ($\Gamma=1$).

\paragraph{Contributions of CPT parameters.}
Table~\ref{tab:cpt-sample-sensitivity} extends the contributions and paired differences in Appendix~\ref{subsec:app1-cpt-sensitivity} to all three training sample sizes. Both probability-weighting parameters contribute more flexibility than utility curvature through $\Gamma=3$ at every $n$. Each contribution also remains positive throughout this range. The ordering between the probability-weighting parameters is more sensitive, as in the $n=100$ calculation.

\begin{table}[t!]

\caption{Sensitivity of CPT parameter contributions across sample sizes}
\label{tab:cpt-sample-sensitivity}
\centering
\small\setlength{\tabcolsep}{3pt}
\begin{tabular}{@{}p{0.38\linewidth}ccc@{}}\hline
Comparison & $n=50$ & $n=100$ & $n=200$ \\ \hline
\multicolumn{4}{@{}l}{\emph{Panel A: Differences under uniform weights}}\\
$C_{\alpha}$ & 0.024 & 0.026 & 0.027 \\
$C_{\gamma}$ & 0.132 & 0.131 & 0.131 \\
$C_{\delta}$ & 0.460 & 0.462 & 0.460 \\
$C_{\delta}-C_{\gamma}$ & 0.328 & 0.331 & 0.329 \\
$C_{\delta}-C_{\alpha}$ & 0.436 & 0.436 & 0.434 \\
$C_{\gamma}-C_{\alpha}$ & 0.108 & 0.105 & 0.105 \\
\hline
\multicolumn{4}{@{}l}{\emph{Panel B: Sensitivity intervals at $\Gamma=2$}}\\
$C_{\alpha}$ & [0.009, 0.054] & [0.010, 0.056] & [0.010, 0.058] \\
$C_{\gamma}$ & [0.059, 0.264] & [0.058, 0.262] & [0.059, 0.263] \\
$C_{\delta}$ & [0.283, 0.633] & [0.286, 0.634] & [0.285, 0.632] \\
$C_{\delta}-C_{\gamma}$ & [0.073, 0.563] & [0.071, 0.563] & [0.070, 0.561] \\
$C_{\delta}-C_{\alpha}$ & [0.260, 0.618] & [0.261, 0.619] & [0.258, 0.617] \\
$C_{\gamma}-C_{\alpha}$ & [0.043, 0.230] & [0.042, 0.227] & [0.041, 0.225] \\
\hline
\multicolumn{4}{@{}l}{\emph{Panel C: Critical values $\Gamma^*$}}\\
$C_{\alpha}$ & $>3$ & $>3$ & $>3$ \\
$C_{\gamma}$ & $>3$ & $>3$ & $>3$ \\
$C_{\delta}$ & $>3$ & $>3$ & $>3$ \\
$C_{\delta}-C_{\gamma}$ & $2.436$ & $2.409$ & $2.408$ \\
$C_{\delta}-C_{\alpha}$ & $>3$ & $>3$ & $>3$ \\
$C_{\gamma}-C_{\alpha}$ & $>3$ & $>3$ & $>3$ \\
\hline
\end{tabular}
\end{table}

\paragraph{Comparisons across procedures.}

At $n=100$, the ordering of Completion, full CPT, and Agent, from most to least restrictive, holds through $\Gamma=3$ (Table~\ref{tab:app1-procedure-sensitivity}). It also holds at $n=200$. At $n=50$, Completion remains more restrictive than full CPT throughout this range, while the full-CPT--Agent and Completion--Agent comparisons reach zero at critical values $1.601$ and $1.992$, respectively. GP regression and the TFMs remain less restrictive than all seven fitted CPT specifications through $\Gamma=3$ at every sample size (Table~\ref{tab:lottery-learning-sensitivity}). AutoML and AIDE also remain less restrictive than full CPT through $\Gamma=3$ at $n=100$. At $n=50$, the full-CPT--AutoML difference is already negative under uniform weights; the full-CPT--AIDE ordering can reverse within the investigated range at $n=50$ and $200$.

\begin{table}[t!]

\caption{Sensitivity of comparisons among Completion, CPT, and Agent}
\label{tab:app1-procedure-sensitivity}
\centering
\small\setlength{\tabcolsep}{3pt}
\begin{tabular}{@{}p{0.38\linewidth}ccc@{}}\hline
Comparison & $n=50$ & $n=100$ & $n=200$ \\ \hline
\multicolumn{4}{@{}l}{\emph{Panel A: Differences under uniform weights}}\\
Completion minus $\mathrm{CPT}(\alpha,\gamma,\delta)$ & 0.043 & 0.207 & 0.269 \\
$\mathrm{CPT}(\alpha,\gamma,\delta)$ minus Agent & 0.051 & 0.096 & 0.105 \\
Completion minus Agent & 0.094 & 0.303 & 0.374 \\
\hline
\multicolumn{4}{@{}l}{\emph{Panel B: Sensitivity intervals at $\Gamma=2$}}\\
Completion minus $\mathrm{CPT}(\alpha,\gamma,\delta)$ & [0.012, 0.103] & [0.082, 0.428] & [0.117, 0.506] \\
$\mathrm{CPT}(\alpha,\gamma,\delta)$ minus Agent & [-0.024, 0.153] & [0.041, 0.192] & [0.048, 0.202] \\
Completion minus Agent & [-0.0005, 0.220] & [0.133, 0.551] & [0.177, 0.661] \\
\hline
\multicolumn{4}{@{}l}{\emph{Panel C: Critical values $\Gamma^*$}}\\
Completion minus $\mathrm{CPT}(\alpha,\gamma,\delta)$ & $>3$ & $>3$ & $>3$ \\
$\mathrm{CPT}(\alpha,\gamma,\delta)$ minus Agent & $1.601$ & $>3$ & $>3$ \\
Completion minus Agent & $1.992$ & $>3$ & $>3$ \\
\hline
\end{tabular}
\end{table}

\clearpage
\begingroup
\small\singlespacing\setlength{\tabcolsep}{3pt}
\renewcommand{\arraystretch}{0.90}
\setlength{\LTpre}{8pt}
\setlength{\LTpost}{8pt}
\begin{longtable}{@{}p{0.38\linewidth}ccc@{}}
\caption{Sensitivity of CPT and other prediction procedures across sample sizes}
\label{tab:lottery-learning-sensitivity}
\label{tab:lottery-learning-sensitivity-bounds}
\label{tab:lottery-learning-sensitivity-critical}
\\
\hline
\addlinespace[2pt]
Comparison & $n=50$ & $n=100$ & $n=200$ \\[2pt]
\hline
\endfirsthead
\caption[]{Sensitivity of CPT and other prediction procedures across sample sizes (continued)}\\
\hline
\addlinespace[2pt]
Comparison & $n=50$ & $n=100$ & $n=200$ \\[2pt]
\hline
\endhead
\hline
\endfoot
\hline
\endlastfoot
\addlinespace[2pt]
\multicolumn{4}{@{}l}{\emph{Panel A: Differences under uniform weights}}\\*[2pt]
$\mathrm{CPT}(\alpha,\gamma,\delta)$ minus GP regression & 0.108 & 0.110 & 0.112 \\
$\mathrm{CPT}(\alpha,\gamma,\delta)$ minus TabPFN & 0.110 & 0.112 & 0.114 \\
$\mathrm{CPT}(\alpha,\gamma,\delta)$ minus TabICL & 0.109 & 0.112 & 0.114 \\
$\mathrm{CPT}(\alpha,\gamma)$ minus GP regression & 0.568 & 0.572 & 0.572 \\
$\mathrm{CPT}(\alpha,\gamma)$ minus TabPFN & 0.571 & 0.574 & 0.574 \\
$\mathrm{CPT}(\alpha,\gamma)$ minus TabICL & 0.569 & 0.574 & 0.574 \\
$\mathrm{CPT}(\gamma,\delta)$ minus GP regression & 0.132 & 0.135 & 0.139 \\
$\mathrm{CPT}(\gamma,\delta)$ minus TabPFN & 0.135 & 0.138 & 0.141 \\
$\mathrm{CPT}(\gamma,\delta)$ minus TabICL & 0.133 & 0.138 & 0.141 \\
$\mathrm{CPT}(\alpha,\delta)$ minus GP regression & 0.240 & 0.241 & 0.244 \\
$\mathrm{CPT}(\alpha,\delta)$ minus TabPFN & 0.242 & 0.243 & 0.245 \\
$\mathrm{CPT}(\alpha,\delta)$ minus TabICL & 0.241 & 0.243 & 0.245 \\
$\mathrm{CPT}(\alpha)$ minus GP regression & 0.835 & 0.840 & 0.842 \\
$\mathrm{CPT}(\alpha)$ minus TabPFN & 0.838 & 0.843 & 0.844 \\
$\mathrm{CPT}(\alpha)$ minus TabICL & 0.837 & 0.843 & 0.844 \\
$\mathrm{CPT}(\gamma)$ minus GP regression & 0.722 & 0.725 & 0.726 \\
$\mathrm{CPT}(\gamma)$ minus TabPFN & 0.724 & 0.727 & 0.728 \\
$\mathrm{CPT}(\gamma)$ minus TabICL & 0.723 & 0.728 & 0.728 \\
$\mathrm{CPT}(\delta)$ minus GP regression & 0.254 & 0.257 & 0.261 \\
$\mathrm{CPT}(\delta)$ minus TabPFN & 0.256 & 0.259 & 0.263 \\
$\mathrm{CPT}(\delta)$ minus TabICL & 0.255 & 0.259 & 0.263 \\
$\mathrm{CPT}(\alpha,\gamma,\delta)$ minus AutoML & -0.004 & 0.057 & 0.085 \\
$\mathrm{CPT}(\alpha,\gamma,\delta)$ minus AIDE & 0.049 & 0.077 & 0.063 \\
\hline
\addlinespace[2pt]
\multicolumn{4}{@{}l}{\emph{Panel B: Sensitivity intervals at $\Gamma=2$}}\\*[2pt]
$\mathrm{CPT}(\alpha,\gamma,\delta)$ minus GP regression & [0.049, 0.204] & [0.050, 0.209] & [0.051, 0.213] \\
$\mathrm{CPT}(\alpha,\gamma,\delta)$ minus TabPFN & [0.051, 0.209] & [0.051, 0.213] & [0.052, 0.217] \\
$\mathrm{CPT}(\alpha,\gamma,\delta)$ minus TabICL & [0.050, 0.207] & [0.051, 0.213] & [0.052, 0.217] \\
$\mathrm{CPT}(\alpha,\gamma)$ minus GP regression & [0.392, 0.705] & [0.397, 0.705] & [0.397, 0.706] \\
$\mathrm{CPT}(\alpha,\gamma)$ minus TabPFN & [0.395, 0.707] & [0.399, 0.707] & [0.399, 0.707] \\
$\mathrm{CPT}(\alpha,\gamma)$ minus TabICL & [0.393, 0.706] & [0.399, 0.707] & [0.399, 0.707] \\
$\mathrm{CPT}(\gamma,\delta)$ minus GP regression & [0.061, 0.247] & [0.062, 0.254] & [0.064, 0.261] \\
$\mathrm{CPT}(\gamma,\delta)$ minus TabPFN & [0.062, 0.253] & [0.063, 0.258] & [0.065, 0.264] \\
$\mathrm{CPT}(\gamma,\delta)$ minus TabICL & [0.061, 0.251] & [0.063, 0.259] & [0.064, 0.264] \\
$\mathrm{CPT}(\alpha,\delta)$ minus GP regression & [0.114, 0.439] & [0.115, 0.444] & [0.116, 0.449] \\
$\mathrm{CPT}(\alpha,\delta)$ minus TabPFN & [0.115, 0.443] & [0.116, 0.448] & [0.117, 0.452] \\
$\mathrm{CPT}(\alpha,\delta)$ minus TabICL & [0.115, 0.441] & [0.116, 0.448] & [0.117, 0.452] \\
$\mathrm{CPT}(\alpha)$ minus GP regression & [0.660, 0.935] & [0.666, 0.939] & [0.668, 0.940] \\
$\mathrm{CPT}(\alpha)$ minus TabPFN & [0.663, 0.937] & [0.668, 0.940] & [0.670, 0.941] \\
$\mathrm{CPT}(\alpha)$ minus TabICL & [0.661, 0.936] & [0.668, 0.940] & [0.670, 0.941] \\
$\mathrm{CPT}(\gamma)$ minus GP regression & [0.623, 0.784] & [0.629, 0.786] & [0.633, 0.788] \\
$\mathrm{CPT}(\gamma)$ minus TabPFN & [0.627, 0.786] & [0.633, 0.788] & [0.635, 0.789] \\
$\mathrm{CPT}(\gamma)$ minus TabICL & [0.626, 0.785] & [0.633, 0.788] & [0.635, 0.789] \\
\pagebreak
\addlinespace[2pt]
\multicolumn{4}{@{}l}{\emph{Panel B: Sensitivity intervals at $\Gamma=2$ (continued)}}\\*[2pt]
$\mathrm{CPT}(\delta)$ minus GP regression & [0.121, 0.456] & [0.123, 0.465] & [0.125, 0.472] \\
$\mathrm{CPT}(\delta)$ minus TabPFN & [0.122, 0.461] & [0.124, 0.469] & [0.126, 0.475] \\
$\mathrm{CPT}(\delta)$ minus TabICL & [0.121, 0.459] & [0.124, 0.469] & [0.126, 0.475] \\
$\mathrm{CPT}(\alpha,\gamma,\delta)$ minus AutoML & [-0.062, 0.048] & [0.019, 0.129] & [0.036, 0.171] \\
$\mathrm{CPT}(\alpha,\gamma,\delta)$ minus AIDE & [-0.024, 0.146] & [0.026, 0.164] & [0.003, 0.158] \\
\hline
\addlinespace[2pt]
\multicolumn{4}{@{}l}{\emph{Panel C: Critical values $\Gamma^*$}}\\*[2pt]
$\mathrm{CPT}(\alpha,\gamma,\delta)$ minus GP regression & $>3$ & $>3$ & $>3$ \\
$\mathrm{CPT}(\alpha,\gamma,\delta)$ minus TabPFN & $>3$ & $>3$ & $>3$ \\
$\mathrm{CPT}(\alpha,\gamma,\delta)$ minus TabICL & $>3$ & $>3$ & $>3$ \\
$\mathrm{CPT}(\alpha,\gamma)$ minus GP regression & $>3$ & $>3$ & $>3$ \\
$\mathrm{CPT}(\alpha,\gamma)$ minus TabPFN & $>3$ & $>3$ & $>3$ \\
$\mathrm{CPT}(\alpha,\gamma)$ minus TabICL & $>3$ & $>3$ & $>3$ \\
$\mathrm{CPT}(\gamma,\delta)$ minus GP regression & $>3$ & $>3$ & $>3$ \\
$\mathrm{CPT}(\gamma,\delta)$ minus TabPFN & $>3$ & $>3$ & $>3$ \\
$\mathrm{CPT}(\gamma,\delta)$ minus TabICL & $>3$ & $>3$ & $>3$ \\
$\mathrm{CPT}(\alpha,\delta)$ minus GP regression & $>3$ & $>3$ & $>3$ \\
$\mathrm{CPT}(\alpha,\delta)$ minus TabPFN & $>3$ & $>3$ & $>3$ \\
$\mathrm{CPT}(\alpha,\delta)$ minus TabICL & $>3$ & $>3$ & $>3$ \\
$\mathrm{CPT}(\alpha)$ minus GP regression & $>3$ & $>3$ & $>3$ \\
$\mathrm{CPT}(\alpha)$ minus TabPFN & $>3$ & $>3$ & $>3$ \\
$\mathrm{CPT}(\alpha)$ minus TabICL & $>3$ & $>3$ & $>3$ \\
$\mathrm{CPT}(\gamma)$ minus GP regression & $>3$ & $>3$ & $>3$ \\
$\mathrm{CPT}(\gamma)$ minus TabPFN & $>3$ & $>3$ & $>3$ \\
$\mathrm{CPT}(\gamma)$ minus TabICL & $>3$ & $>3$ & $>3$ \\
$\mathrm{CPT}(\delta)$ minus GP regression & $>3$ & $>3$ & $>3$ \\
$\mathrm{CPT}(\delta)$ minus TabPFN & $>3$ & $>3$ & $>3$ \\
$\mathrm{CPT}(\delta)$ minus TabICL & $>3$ & $>3$ & $>3$ \\
$\mathrm{CPT}(\alpha,\gamma,\delta)$ minus AutoML & $1$ & $>3$ & $>3$ \\
$\mathrm{CPT}(\alpha,\gamma,\delta)$ minus AIDE & $1.589$ & $>3$ & $2.076$ \\
\end{longtable}
\endgroup

\FloatBarrier

\paragraph{Within-function averaging.}

Table~\ref{tab:app1-inner-point-comparison} compares restrictiveness computed with one and twenty training samples per function. Within-function averaging has little effect on estimated restrictiveness. Across all 33 procedure--sample-size estimates in this application, the median absolute change is 0.0022 and the largest is 0.0060. The proportional changes for GP regression and the two TFMs at $n=100$ are larger---about one-third to two-fifths---because their restrictiveness estimates are close to zero. 

Table~\ref{tab:app1-inner-bound-comparison} compares sensitivity bounds and critical values under $K=1$ and $20$. The directional robustness conclusions at $\Gamma=3$ are unchanged: the lower bound has the same sign in all 84 directional comparisons across the three sample sizes for procedures evaluated at both $K$ values. The bounds also change little in magnitude. At $\Gamma=2$, the absolute changes in the lower and upper sensitivity bounds from $K=1$ to $K=20$ have a median of 0.0027 and a maximum of 0.0100.

\begin{table}[t!]

\caption{Restrictiveness with one and twenty training samples per function}
\label{tab:app1-inner-point-comparison}
\centering
\small\setlength{\tabcolsep}{4pt}
\begin{tabular}{@{}lrrr@{}}\hline
Procedure & $K=1$ & $K=20$ & Change \\ \hline
$\mathrm{CPT}(\alpha,\gamma,\delta)$ & 0.1232 & 0.1210 & -0.0022 \\
$\mathrm{CPT}(\alpha,\gamma)$ & 0.5843 & 0.5831 & -0.0012 \\
$\mathrm{CPT}(\gamma,\delta)$ & 0.1483 & 0.1467 & -0.0017 \\
$\mathrm{CPT}(\alpha,\delta)$ & 0.2532 & 0.2519 & -0.0013 \\
$\mathrm{CPT}(\alpha)$ & 0.8519 & 0.8516 & -0.0003 \\
$\mathrm{CPT}(\gamma)$ & 0.7371 & 0.7362 & -0.0009 \\
$\mathrm{CPT}(\delta)$ & 0.2692 & 0.2680 & -0.0012 \\
GP regression & 0.0167 & 0.0112 & -0.0055 \\
TabPFN & 0.0144 & 0.0088 & -0.0057 \\
TabICL & 0.0147 & 0.0087 & -0.0060 \\
AutoML & 0.0690 & 0.0635 & -0.0055 \\
\hline\end{tabular}
\par\smallskip
{\footnotesize\raggedright\emph{Notes:} $n=100$.\par}
\end{table}

\begin{table}[t!]

\caption{Effect of within-function averaging on sensitivity bounds}
\label{tab:app1-inner-bound-comparison}
\centering
\small\setlength{\tabcolsep}{3pt}
\begin{tabular}{@{}p{0.31\linewidth}rcccc@{}}\hline
& & \multicolumn{2}{c}{$K=1$} & \multicolumn{2}{c}{$K=20$}\\
\cmidrule(lr){3-4}\cmidrule(lr){5-6}
Comparison & $n$ & $\Gamma=2$ & $\Gamma^*$ & $\Gamma=2$ & $\Gamma^*$\\ \hline
$C_{\delta}-C_{\gamma}$ & 100 & [0.072, 0.562] & $2.420$ & [0.071, 0.563] & $2.409$ \\
$C_{\delta}-C_{\alpha}$ & 100 & [0.261, 0.618] & $>3$ & [0.261, 0.619] & $>3$ \\
$C_{\gamma}-C_{\alpha}$ & 100 & [0.042, 0.225] & $>3$ & [0.042, 0.227] & $>3$ \\
$\mathrm{CPT}(\alpha,\gamma,\delta)$ minus AutoML & 50 & [-0.072, 0.057] & $1$ & [-0.062, 0.048] & $1$ \\
\hline\end{tabular}
\par\smallskip
{\footnotesize\raggedright\emph{Notes:} The first three rows compare CPT contribution differences at $n=100$; the last reports the largest absolute change in a lower or upper sensitivity bound at $\Gamma=2$ over all three sample sizes.\par}
\end{table}

\endgroup

\subsection{Evaluation Distribution Based on the TabICL Prior}
\label{subsec:app1-tfm-prior}

In this subsection, we describe the alternative evaluation distribution in Appendix~\ref{subsec:cpt-tfm} based on a generative prior used to train a TFM model. Specifically, we use the TabICL v2 pretraining prior released in version 2.2.0 of the \texttt{tabicl} package, with the default configuration of \texttt{tabicl.prior.graph\_lib}. The prior generates a directed acyclic graph with 2--32 nodes, each of which consists of a randomly selected function from eight families: neural networks, oblivious-tree ensembles, nearest-center discretizations, linear and quadratic maps, random-feature Gaussian processes, soft cluster assignments, and their products. A randomly selected node generates the target (outcome) variable. We adapt this construction by replacing the base inputs of one root ancestor of the target node, selected uniformly at random, with standardized lottery covariates.

For each sampled graph, the latent eligible prediction rule is the conditional mean $h(x)=\E[y\mid x]$. We approximate $h$ by Monte Carlo using 256--8,192 samples per point, where the Monte Carlo sampling stops once the Monte Carlo RMSE across points is below $2\%$ of the standard deviation of $h$ across those points (or a computational limit is reached). A candidate prediction rule is treated as valid if its values are finite and nonconstant, its explained-variance share is at least $0.05$, and its relative Monte Carlo error is at most $5\%$. 

We apply the increasing/decreasing adjustment and blending construction in Online Appendix~\ref{subsec:evaluation.example1} to each valid candidate rule. Let $\widetilde h_m$ be the resulting monotone latent values and let $\rho$ be their correlation with the original candidate. We keep a candidate rule if $\rho\geq0.8$ (so that the function's overall shape is not drastically annihilated by the monotonization transformation). The acceptance rates for this screening procedure are $17\%$, $14\%$, and $16\%$ at $n=50,100,200$, respectively. These restrictions and adjustments make the evaluation distribution a monotone adaptation of the pretraining prior. The TFM pretraining prior does not specify a scale for lottery valuations. We transform each accepted TFM-based prediction rule to share the same mean $\mu_m$ and standard deviation $s_m$ of the corresponding GP-based prediction rule.

Table~\ref{tab:app1-tfm-profiles} reports the seven CPT specifications and four statistical learners at each sample size, together with the three neutral-task LLM configurations at $n=100$. All fitting settings, prompts, and computation budgets for the reported procedures are held fixed.

\begin{table}[t!]
\caption{Restrictiveness profiles under the TFM-based evaluation distribution}
\label{tab:app1-tfm-profiles}
\centering
\small\singlespacing\setlength{\tabcolsep}{4pt}
\begin{tabular}{@{}lccccc@{}}\toprule
 & GP-based   & \multicolumn{3}{c}{TFM-based } & Change \\
\cmidrule(lr){2-2}\cmidrule(lr){3-5}\cmidrule(lr){6-6}
 & $n=100$ & $n=50$ & $n=100$ & $n=200$ & at $n=100$ \\ \midrule
\multicolumn{6}{@{}l}{\emph{Economic models}} \\
  CPT$(\alpha,\gamma,\delta)$ & 0.123 & 0.129 & 0.126 & 0.143 & 0.003 \\
   & (0.017) & & (0.016) & & (0.011) \\
  CPT$(\gamma,\delta)$ & 0.148 & 0.171 & 0.161 & 0.193 & 0.013 \\
   & (0.020) & & (0.022) & & (0.014) \\
  CPT$(\alpha,\delta)$ & 0.253 & 0.324 & 0.307 & 0.343 & 0.054 \\
   & (0.033) & & (0.034) & & (0.019) \\
  CPT$(\delta)$ & 0.269 & 0.348 & 0.331 & 0.376 & 0.062 \\
   & (0.034) & & (0.037) & & (0.020) \\
  CPT$(\alpha,\gamma)$ & 0.584 & 0.485 & 0.498 & 0.498 & $-$0.087 \\
   & (0.032) & & (0.033) & & (0.013) \\
  CPT$(\gamma)$ & 0.737 & 0.677 & 0.682 & 0.674 & $-$0.055 \\
   & (0.013) & & (0.018) & & (0.014) \\
  CPT$(\alpha)$ & 0.852 & 0.815 & 0.825 & 0.831 & $-$0.027 \\
   & (0.030) & & (0.029) & & (0.012) \\
\midrule
\multicolumn{6}{@{}l}{\emph{Statistical learners}} \\
  GP regression & 0.017 & 0.041 & 0.031 & 0.028 & 0.014 \\
   & (0.002) & & (0.005) & & (0.004) \\
  TabPFN v2 & 0.014 & 0.025 & 0.019 & 0.016 & 0.005 \\
   & (0.002) & & (0.004) & & (0.003) \\
  TabICL v2 & 0.015 & 0.025 & 0.018 & 0.014 & 0.003 \\
   & (0.002) & & (0.003) & & (0.002) \\
  AutoML (FLAML) & 0.069 & 0.096 & 0.059 & 0.034 & $-$0.010 \\
   & (0.009) & & (0.008) & & (0.007) \\
\midrule
\multicolumn{6}{@{}l}{\emph{Claude Sonnet 5.5}} \\
  Completion & 0.328 & --- & 0.266 & --- & $-$0.062 \\
   & (0.053) & & (0.036) & & (0.046) \\
  Completion (high thinking) & 0.087 & --- & 0.069 & --- & $-$0.018 \\
   & (0.019) & & (0.014) & & (0.020) \\
  Agent & 0.025 & --- & 0.025 & --- & 0.000 \\
   & (0.004) & & (0.005) & & (0.005) \\
\bottomrule
\end{tabular}
\par\smallskip
{\footnotesize\raggedright\emph{Notes:} $M=100$ replications per sample size. Parentheses report Monte Carlo standard errors for the $n=100$ estimates and their changes. All LLM configurations use Claude Sonnet 5.5.\par}
\end{table}

\clearpage
\section{Details for Application 2 (Discrete Choice)}
\label{sec:detail.example2}

\subsection{Evaluation and Inference}
\label{subsec:evaluation.example2}
A market $X=(x_1,\ldots,x_{24})$ contains 24 products, with $x_j=(p_j,m_j,s_j)$ denoting price, mushy status, and sugar content. For each category $m\in\{0,1\}$, let $\mathcal H_m$ be the convex hull of the observed price--sugar pairs in the calibration data. We draw products independently according to
\begin{equation*}
m_j\sim\mathrm{Bernoulli}(1/3),
\quad
(p_j,s_j)\mid m_j\sim\mathrm{Uniform}(\mathcal H_{m_j}),
\quad j=1,\ldots,24.
\end{equation*}
These draws define the market distribution $P_X$. To sample uniformly within a hull, we triangulate it, select a triangle with probability proportional to its area, and sample uniformly within that triangle. Price and sugar are standardized using their pooled calibration means and standard deviations.

For each replication, we independently draw one demand system that applies to all training and evaluation markets. Its preferences are generated by four independent zero-mean Gaussian fields: $H_C(z^s,m)$ and $H_I(z^s,m,\nu)$ describe common and individual utility-level components, while $G_C(v,z^s,m)$ and $G_I(v,z^s,m,\nu)$ describe common and individual price-sensitivity indices. Here $z^s$ is standardized sugar content, $v$ is a standardized price argument, and $\nu\sim N(0,I_4)$ is consumer type.

For each raw field $F$, let $u$ collect its continuous arguments. Its target covariance is
\begin{equation*}
\begin{gathered}
k_F\bigl((u,m),(u',m')\bigr)
=B_{mm'}(1+\sqrt{3}r_F)\exp(-\sqrt{3}r_F),\\
r_F^2
=\sum_h\frac{(u_h-u_h')^2}{\ell_{F,h}^2},
\quad
B=\begin{pmatrix}1&0.5\\0.5&1\end{pmatrix}.
\end{gathered}
\end{equation*}
We approximate each field using 64 random Fourier features. The price and sugar length scales, $\ell_p=0.66$ and $\ell_s=1.25$, are calibrated to make target correlation $1/2$ at the median nonzero standardized change in the corresponding coordinate, holding other inputs fixed. The calibration uses within-product market pairs for price and product pairs for sugar. Each consumer-type coordinate has length scale one.

We center and standardize each raw field on a fixed reference grid combining the calibration product--market observations with 128 simulated consumer types. Denote the resulting fields by $\widetilde H_C$, $\widetilde H_I$, $\widetilde G_C$, and $\widetilde G_I$. The functions entering the utility specification in the main paper are
\begin{equation*}
\begin{aligned}
L(z^s,m,\nu)
&=\frac{\widetilde H_C(z^s,m)
       +\widetilde H_I(z^s,m,\nu)}{s_L},\\
\Lambda(v,z^s,m,\nu)
&=\frac{\widetilde G_C(v,z^s,m)
       +\widetilde G_I(v,z^s,m,\nu)}{s_\Lambda},
\end{aligned}
\end{equation*}
where $s_L$ and $s_\Lambda$ are the standard deviations of the respective sums on the reference grid. Thus, $L$ adds the common and individual utility-level components, while $\Lambda$ adds the common and individual price-index components. Both sums are rescaled to unit variance.

We transform $\Lambda$ into a positive price sensitivity,
\begin{equation}
a(v,z^s,m,\nu)
=\frac{\log\bigl(1+\exp\Lambda(v,z^s,m,\nu)\bigr)}{c},
\end{equation}
where $c$ makes the average of $a$ on the reference grid equal to one. Systematic utility is then
\begin{equation}
V_j(\nu)
=-\log 24+L(z_j^s,m_j,\nu)
-\int_0^{z_j^p}a(v,z_j^s,m_j,\nu)\,dv.
\end{equation}
The level component enters utility directly, while the integral of the positive sensitivity determines the price contribution. Normalizing outside utility to $V_0(\nu)=0$, population shares are
\begin{equation}
f_j(X)
=\mathbb E_{\nu\sim N(0,I_4)}
\left[
\frac{\exp(V_j(\nu))}
{1+\sum_{k=1}^{24}\exp(V_k(\nu))}
\right],
\quad j=0,\ldots,24.
\end{equation}
We approximate population shares by averaging choice probabilities over 512 simulated consumer types, separate from the 128 types used for normalization. Both sets pair scrambled-Sobol normal draws with their negatives and remain fixed across replications. The price integral in each consumer's utility is approximated using 64 Gauss--Legendre nodes.

To calculate the induced restrictiveness for elasticities, we approximate the price elasticities using symmetric finite differences in log prices:
\[
\mathcal E_{jk}^{(h)}[g](X)
=
\frac{\log g_j(X^{+k})-\log g_j(X^{-k})}{2h_k},
\]
where $X^{+k}$ and $X^{-k}$ replace product $k$'s price $p_k$ with $p_k e^{h_k}$ and $p_k e^{-h_k}$, respectively, holding all other product characteristics fixed.
To keep the perturbed prices within the support, let
$L_k$ and $U_k$ denote the endpoints of the feasible price interval at product $k$'s sugar content, obtained from the convex hull's half-space representation. We set
\[
h_k
=
\min\left\{
0.01,\,
\frac{1}{2}\log\frac{U_k}{p_k},\,
\frac{1}{2}\log\frac{p_k}{L_k}
\right\}.
\]
The same finite differences are applied to the fitted
share functions and the pseudo-true demand system.

For completeness, we use the cross-validation design described in Section~\ref{sec:app-dc}. We average squared errors over the 24 inside products' shares within each market and then over the two markets within each city. The estimator and SE of completeness are computed as follows. Let $L^j_{rkc}$ be the held-out loss of procedure $j$ for city $c$ in fold $\mathcal C_{rk}$ of repetition $r$. With $R=K=5$ and $G=47$, define
\begin{equation}
\widehat R_j=\frac{1}{RK}\sum_{r=1}^R\sum_{k=1}^K
\frac{1}{|\mathcal C_{rk}|}\sum_{c\in\mathcal C_{rk}}L^j_{rkc}.
\label{eq:app2-heldout-risk}
\end{equation}
Write $B$ for the baseline and $F$ for the TabPFN reference, and let $a^j_{rkc}=L^B_{rkc}-L^j_{rkc}$, $d_{rkc}=L^B_{rkc}-L^F_{rkc}$, and $\widehat D=\widehat R_B-\widehat R_F$. We report
\begin{equation}
\begin{aligned}
\widehat\kappa_j&=\frac{\widehat R_B-\widehat R_j}{\widehat D},\\
\widehat{\operatorname{se}}(\widehat\kappa_j)
&=\left\{\frac{\sum_{r=1}^R\sum_{k=1}^K\sum_{c\in\mathcal C_{rk}}
\bigl(a^j_{rkc}-\widehat\kappa_jd_{rkc}\bigr)^2}
{RG^2\widehat D^2}\right\}^{1/2}.
\end{aligned}
\label{eq:app2-completeness-se}
\end{equation}

\subsection{Fitting Procedures}
For the structural models, the coefficient bounds are $\beta\in[-50,50]^4$. Nesting parameters lie in $[0,0.99]$, and Cholesky entries lie in $[0,100]$ on the diagonal and $[-100,100]$ otherwise. Starting values include the parameters of the baseline logit model and a Berry log-share regression. For NL and MXL, we also use the fitted MNL with fixed perturbations. Empty nests in NL contribute no inclusive value. MXL integrates the random coefficients over 512 nodes. Estimation is limited to 400 objective evaluations.

Completion in the main comparison, Completion with high thinking, and Agent use the model and thinking settings from Application~1, with an output limit of 64,000 tokens per response and a 300-second timeout for each tool call. Returned programs have access to the supplied training data and may estimate parameters when initialized or first queried. Agent may also use data or code produced during fitting. Some synthetic Agent programs calculate regression or kernel coefficients when their returned programs are run. In the returned programs, we find that Agent estimates its selected specifications by minimizing share cross-entropy or squared log-share ratios, and often compares candidates by cross-validation across markets.

The TabPFN procedure uses fixed version-2 regression weights. Each regressor receives the 72 product characteristics in a market and predicts one product's log share relative to the outside option. Invalid share predictions receive the fitted common-intercept logit baseline for shares and price elasticities in both exercises. If a returned program fails to execute, the baseline is used at every evaluation point.

Table~\ref{tab:app2-thinking} compares the two Completion configurations with Agent at $n=100$. As in Application 1, Completion with high thinking falls between the main Completion configuration and Agent on all three restrictiveness measures. It also achieves higher completeness than the main configuration, bringing prediction loss close to the baseline. Agent remains substantially more complete.

\begin{table}[!htbp]
  \caption{Completion configurations and Agent: discrete choice at $n=100$}
  \label{tab:app2-thinking}
  \centering
  \small\setlength{\tabcolsep}{5pt}%
  \begin{tabular}{@{}lcccc@{}}\hline
   & \multicolumn{3}{c}{Restrictiveness} & Completeness \\
   \cmidrule(lr){2-4}\cmidrule(lr){5-5}
   & Shares & Own-price & Cross-price & $\kappa$ \\ \hline\shortstack[l]{Completion (main)}
      & 0.871 (0.012) & 0.653 (0.025)
      & 0.884 (0.008) & $-$0.802 (0.246) \\
    \shortstack[l]{Completion (high thinking)}
      & 0.634 (0.025) & 0.471 (0.022)
      & 0.740 (0.021) & 0.006 (0.092) \\
    Agent & 0.026 (0.010) & 0.082 (0.009) & 0.106 (0.019) & 1.149 (0.044) \\ \hline
  \end{tabular}
  \par\smallskip
  {\footnotesize\raggedright
  \emph{Notes:} The Own-price and Cross-price columns report restrictiveness for own-price and cross-price elasticities, respectively. Standard errors are reported in parentheses.
  \par}
\end{table}

\subsection{Specification Choices}
\label{subsec:specifications.example2}
Figure~\ref{fig:app2-neutral-model-choices} reports the specifications selected in every synthetic exercise at $n=100$ and every empirical fit of Completion and Agent. We identify the specification used to generate predictions from the source code, saved estimates, and execution records.

For restrictiveness, we classify fits to simulated market shares by their utility specifications and implied substitution patterns. \emph{Linear/interaction logit} uses multinomial-logit probabilities, including an outside option, with utilities linear in the observed characteristics and, where selected, their interactions, product fixed effects, or linear market-level terms. \emph{Heterogeneous logit} averages conditional-logit choice probabilities over finite latent classes, random tastes on characteristics, or simulated product-quality shocks. 

\emph{Flexible-utility logit} uses logit normalization with utility indices modeled by higher-order polynomials, splines, neural networks, Gaussian processes, kernel predictors, or other nonlinear functions of characteristics. These indices equal log share ratios relative to the outside option. With estimated indices $\widehat g_j(x)$, normalization gives \[ \widehat s_j(x) =\frac{\exp(\widehat g_j(x))} {1+\sum_{\ell=1}^{24}\exp(\widehat g_\ell(x))}. \] \emph{Constant shares} denotes a market-share vector that does not vary with the market's characteristics, either fixed in advance or calculated from the supplied shares. In the restrictiveness panels, \emph{Other} contains one Completion logit whose probabilities are shrunk toward equal shares. For Agent, it contains one nested logit and four fits whose specification cannot be determined from the available records. Nested logit uses inclusive-value or equivalent power aggregation to distinguish within-nest substitution from substitution toward other nests or the outside option. 

For completeness, the same categories describe the rules fitted to cereal market shares. The 23 Agent linear/interaction logits include 22 common-price-slope models and one characteristic-dependent price slope, all with product effects. The two Agent heterogeneous logits integrate a nondegenerate normal random taste for mushy. 

Each fitted rule is counted once. A rule that averages conditional-logit probabilities over latent heterogeneity is classified as heterogeneous logit, even if its conditional utilities are nonlinear; a nested logit without latent heterogeneity enters \emph{Other}. \emph{Execution failure} denotes a program that fails to initialize, omits the required prediction function, or cannot produce predictions at any evaluation point. The figure includes all fitted rules, and empirical counts combine fits across the overlapping folds.

\begingroup
\subsection{Reweighting the Evaluation Distribution}
\label{subsec:app2-sensitivity-online}

We apply the function weights in Appendix~\ref{subsec:app1-cpt-sensitivity} separately at $n=50,100,200$ to $M=200$ demand systems. As explained in equation~\eqref{eq:sensitivity-conditional-risk}, the optimization concerns conditional mean losses for each system. MNL, NL, and MXL use twenty training samples per system; Completion, Agent, and TabPFN use one. All procedures share the same evaluation markets. The common denominator for shares averages the losses of twenty fitted common-intercept baselines.

Tables~\ref{tab:app2-procedure-sensitivity}, \ref{tab:demand-own-sensitivity}, and~\ref{tab:demand-cross-sensitivity} report the baseline differences, sensitivity intervals, and critical values at all three sample sizes.

\paragraph{Market shares.}

At $n=100$, NL and MXL remain less restrictive than MNL for shares through $\Gamma=3$; the same holds at $n=50$ and $200$. At all three sample sizes, MXL remains less restrictive than NL until
$\Gamma$ reaches a critical value close to 3. Agent and TabPFN remain less restrictive than the structural procedures through $\Gamma=3$ at $n=100$ and $200$. Those comparisons are more sensitive at $n=50$; the MXL--TabPFN difference is already negative under uniform weights. Completion remains more restrictive than every other procedure through $\Gamma=3$ at each $n$.

\paragraph{Own-price elasticities.}

NL remains more restrictive than MNL for own-price elasticities through $\Gamma=3$ at all three sample sizes. MXL and MNL have similar baseline own-price estimates; their difference interval at $\Gamma=2$ contains zero at every $n$. NL is more restrictive than MXL under uniform weights at every $n$, but their $\Gamma=2$ intervals include zero. TabPFN remains less restrictive than the structural procedures throughout the range considered. Agent has the same property at $n=100$ and $200$, but each structural-model comparison can reverse at $n=50$. Completion's greater restrictiveness than NL can reverse before $\Gamma=2$ at each $n$.

\paragraph{Cross-price elasticities.}

NL and MXL remain more restrictive than MNL for cross-price elasticities through $\Gamma=3$ at all three sample sizes. MXL is more restrictive than NL under uniform weights at $n=100$ and $200$, but less restrictive at $n=50$. The MXL--NL interval at $\Gamma=2$ includes zero at every $n$. The comparisons between structural procedures and Agent or TabPFN are more sensitive at $n=50$; the MNL--TabPFN difference is already negative under uniform weights. At $n=100$, Agent's lower restrictiveness than MNL can reverse at $\Gamma^*\approx2.950$; the remaining structural comparisons with Agent and TabPFN stay positive through three. At $n=200$, all of these comparisons remain positive through three. Completion remains more restrictive than every other procedure throughout the range at each $n$.

\begin{table}[p]

\caption{Sensitivity of share comparisons across sample sizes}
\label{tab:app2-procedure-sensitivity}
\label{tab:app2-procedure-sensitivity-bounds}
\label{tab:app2-procedure-sensitivity-critical}
\centering
\small\setlength{\tabcolsep}{3pt}
\renewcommand{\arraystretch}{0.90}
\begin{tabular}{@{}p{0.38\linewidth}ccc@{}}
\hline
\addlinespace[2pt]
Comparison & $n=50$ & $n=100$ & $n=200$ \\[2pt]
\hline
\addlinespace[2pt]
\multicolumn{4}{@{}l}{\emph{Panel A: Differences under uniform weights}}\\[2pt]
MNL minus NL & 0.009 & 0.010 & 0.010 \\
MNL minus MXL & 0.043 & 0.043 & 0.042 \\
NL minus MXL & 0.033 & 0.033 & 0.032 \\
MNL minus TabPFN & 0.023 & 0.137 & 0.166 \\
MNL minus Agent & 0.103 & 0.156 & 0.158 \\
NL minus TabPFN & 0.014 & 0.127 & 0.156 \\
NL minus Agent & 0.094 & 0.146 & 0.148 \\
MXL minus TabPFN & -0.020 & 0.094 & 0.123 \\
MXL minus Agent & 0.061 & 0.113 & 0.115 \\
Completion minus TabPFN & 0.718 & 0.825 & 0.892 \\
Completion minus Agent & 0.798 & 0.845 & 0.884 \\
Completion minus MNL & 0.695 & 0.689 & 0.726 \\
Completion minus NL & 0.705 & 0.699 & 0.736 \\
Completion minus MXL & 0.738 & 0.731 & 0.768 \\
\hline
\addlinespace[2pt]
\multicolumn{4}{@{}l}{\emph{Panel B: Sensitivity intervals at $\Gamma=2$}}\\[2pt]
MNL minus NL & [0.004, 0.018] & [0.005, 0.018] & [0.005, 0.018] \\
MNL minus MXL & [0.019, 0.084] & [0.019, 0.083] & [0.019, 0.083] \\
NL minus MXL & [0.009, 0.073] & [0.009, 0.072] & [0.008, 0.072] \\
MNL minus TabPFN & [-0.028, 0.082] & [0.086, 0.205] & [0.112, 0.239] \\
MNL minus Agent & [0.008, 0.198] & [0.091, 0.235] & [0.093, 0.235] \\
NL minus TabPFN & [-0.038, 0.072] & [0.077, 0.193] & [0.103, 0.227] \\
NL minus Agent & [-0.00007, 0.187] & [0.082, 0.224] & [0.085, 0.224] \\
MXL minus TabPFN & [-0.075, 0.035] & [0.053, 0.149] & [0.081, 0.180] \\
MXL minus Agent & [-0.029, 0.139] & [0.058, 0.177] & [0.059, 0.179] \\
Completion minus TabPFN & [0.607, 0.834] & [0.730, 0.921] & [0.775, 1.038] \\
Completion minus Agent & [0.657, 0.916] & [0.740, 0.944] & [0.755, 1.034] \\
Completion minus MNL & [0.573, 0.814] & [0.564, 0.802] & [0.589, 0.886] \\
Completion minus NL & [0.582, 0.823] & [0.575, 0.810] & [0.599, 0.895] \\
Completion minus MXL & [0.623, 0.848] & [0.619, 0.837] & [0.642, 0.924] \\
\hline
\addlinespace[2pt]
\multicolumn{4}{@{}l}{\emph{Panel C: Critical values $\Gamma^*$}}\\[2pt]
MNL minus NL & $>3$ & $>3$ & $>3$ \\
MNL minus MXL & $>3$ & $>3$ & $>3$ \\
NL minus MXL & $2.991$ & $2.947$ & $2.910$ \\
MNL minus TabPFN & $1.354$ & $>3$ & $>3$ \\
MNL minus Agent & $2.114$ & $>3$ & $>3$ \\
NL minus TabPFN & $1.195$ & $>3$ & $>3$ \\
NL minus Agent & $1.999$ & $>3$ & $>3$ \\
MXL minus TabPFN & $1$ & $>3$ & $>3$ \\
MXL minus Agent & $1.639$ & $>3$ & $>3$ \\
Completion minus TabPFN & $>3$ & $>3$ & $>3$ \\
Completion minus Agent & $>3$ & $>3$ & $>3$ \\
Completion minus MNL & $>3$ & $>3$ & $>3$ \\
Completion minus NL & $>3$ & $>3$ & $>3$ \\
Completion minus MXL & $>3$ & $>3$ & $>3$ \\
\hline
\end{tabular}
\end{table}

\begin{table}[p]

\caption{Sensitivity of own-price elasticity comparisons across sample sizes}
\label{tab:demand-own-sensitivity}
\label{tab:demand-own-sensitivity-bounds}
\label{tab:demand-own-sensitivity-critical}
\centering
\small\setlength{\tabcolsep}{3pt}
\renewcommand{\arraystretch}{0.90}
\begin{tabular}{@{}p{0.38\linewidth}ccc@{}}
\hline
\addlinespace[2pt]
Comparison & $n=50$ & $n=100$ & $n=200$ \\[2pt]
\hline
\addlinespace[2pt]
\multicolumn{4}{@{}l}{\emph{Panel A: Differences under uniform weights}}\\[2pt]
NL minus MNL & 0.116 & 0.106 & 0.104 \\
MXL minus MNL & 0.017 & 0.016 & 0.015 \\
NL minus MXL & 0.098 & 0.090 & 0.088 \\
MNL minus TabPFN & 0.165 & 0.224 & 0.257 \\
MNL minus Agent & 0.120 & 0.217 & 0.203 \\
NL minus TabPFN & 0.281 & 0.330 & 0.361 \\
NL minus Agent & 0.235 & 0.324 & 0.306 \\
MXL minus TabPFN & 0.182 & 0.240 & 0.272 \\
MXL minus Agent & 0.137 & 0.233 & 0.218 \\
Completion minus TabPFN & 0.513 & 0.577 & 0.618 \\
Completion minus Agent & 0.468 & 0.571 & 0.563 \\
Completion minus MNL & 0.348 & 0.353 & 0.361 \\
Completion minus NL & 0.233 & 0.247 & 0.257 \\
Completion minus MXL & 0.331 & 0.337 & 0.345 \\
\hline
\addlinespace[2pt]
\multicolumn{4}{@{}l}{\emph{Panel B: Sensitivity intervals at $\Gamma=2$}}\\[2pt]
NL minus MNL & [0.045, 0.239] & [0.041, 0.221] & [0.039, 0.216] \\
MXL minus MNL & [-0.062, 0.092] & [-0.065, 0.093] & [-0.067, 0.095] \\
NL minus MXL & [-0.028, 0.280] & [-0.034, 0.264] & [-0.036, 0.261] \\
MNL minus TabPFN & [0.098, 0.263] & [0.156, 0.323] & [0.189, 0.356] \\
MNL minus Agent & [-0.016, 0.253] & [0.115, 0.333] & [0.097, 0.320] \\
NL minus TabPFN & [0.158, 0.481] & [0.210, 0.523] & [0.241, 0.551] \\
NL minus Agent & [0.060, 0.451] & [0.186, 0.518] & [0.165, 0.503] \\
MXL minus TabPFN & [0.116, 0.267] & [0.173, 0.326] & [0.207, 0.359] \\
MXL minus Agent & [0.018, 0.246] & [0.140, 0.332] & [0.122, 0.316] \\
Completion minus TabPFN & [0.368, 0.669] & [0.426, 0.749] & [0.457, 0.802] \\
Completion minus Agent & [0.281, 0.653] & [0.406, 0.757] & [0.391, 0.747] \\
Completion minus MNL & [0.165, 0.521] & [0.172, 0.533] & [0.166, 0.562] \\
Completion minus NL & [-0.022, 0.447] & [-0.005, 0.466] & [-0.002, 0.490] \\
Completion minus MXL & [0.180, 0.487] & [0.164, 0.515] & [0.169, 0.537] \\
\hline
\addlinespace[2pt]
\multicolumn{4}{@{}l}{\emph{Panel C: Critical values $\Gamma^*$}}\\[2pt]
NL minus MNL & $>3$ & $>3$ & $>3$ \\
MXL minus MNL & $1.175$ & $1.161$ & $1.151$ \\
NL minus MXL & $1.688$ & $1.632$ & $1.610$ \\
MNL minus TabPFN & $>3$ & $>3$ & $>3$ \\
MNL minus Agent & $1.845$ & $>3$ & $>3$ \\
NL minus TabPFN & $>3$ & $>3$ & $>3$ \\
NL minus Agent & $2.601$ & $>3$ & $>3$ \\
MXL minus TabPFN & $>3$ & $>3$ & $>3$ \\
MXL minus Agent & $2.203$ & $>3$ & $>3$ \\
Completion minus TabPFN & $>3$ & $>3$ & $>3$ \\
Completion minus Agent & $>3$ & $>3$ & $>3$ \\
Completion minus MNL & $>3$ & $>3$ & $>3$ \\
Completion minus NL & $1.888$ & $1.973$ & $1.989$ \\
Completion minus MXL & $>3$ & $>3$ & $>3$ \\
\hline
\end{tabular}
\end{table}

\begin{table}[p]

\caption{Sensitivity of cross-price elasticity comparisons across sample sizes}
\label{tab:demand-cross-sensitivity}
\label{tab:demand-cross-sensitivity-bounds}
\label{tab:demand-cross-sensitivity-critical}
\centering
\small\setlength{\tabcolsep}{3pt}
\renewcommand{\arraystretch}{0.90}
\begin{tabular}{@{}p{0.38\linewidth}ccc@{}}
\hline
\addlinespace[2pt]
Comparison & $n=50$ & $n=100$ & $n=200$ \\[2pt]
\hline
\addlinespace[2pt]
\multicolumn{4}{@{}l}{\emph{Panel A: Differences under uniform weights}}\\[2pt]
NL minus MNL & 0.150 & 0.114 & 0.105 \\
MXL minus MNL & 0.146 & 0.147 & 0.145 \\
MXL minus NL & -0.003 & 0.032 & 0.040 \\
MNL minus TabPFN & -0.063 & 0.095 & 0.167 \\
MNL minus Agent & 0.071 & 0.161 & 0.176 \\
NL minus TabPFN & 0.087 & 0.210 & 0.272 \\
NL minus Agent & 0.220 & 0.275 & 0.281 \\
MXL minus TabPFN & 0.083 & 0.242 & 0.311 \\
MXL minus Agent & 0.217 & 0.308 & 0.321 \\
Completion minus TabPFN & 0.556 & 0.712 & 0.780 \\
Completion minus Agent & 0.689 & 0.778 & 0.790 \\
Completion minus MNL & 0.619 & 0.617 & 0.613 \\
Completion minus NL & 0.469 & 0.502 & 0.508 \\
Completion minus MXL & 0.472 & 0.470 & 0.469 \\
\hline
\addlinespace[2pt]
\multicolumn{4}{@{}l}{\emph{Panel B: Sensitivity intervals at $\Gamma=2$}}\\[2pt]
NL minus MNL & [0.071, 0.291] & [0.052, 0.224] & [0.046, 0.209] \\
MXL minus MNL & [0.077, 0.232] & [0.074, 0.234] & [0.071, 0.234] \\
MXL minus NL & [-0.153, 0.123] & [-0.092, 0.151] & [-0.081, 0.158] \\
MNL minus TabPFN & [-0.128, -0.00002] & [0.041, 0.154] & [0.115, 0.226] \\
MNL minus Agent & [-0.063, 0.177] & [0.066, 0.244] & [0.096, 0.251] \\
NL minus TabPFN & [-0.028, 0.253] & [0.113, 0.347] & [0.177, 0.406] \\
NL minus Agent & [0.048, 0.416] & [0.149, 0.427] & [0.175, 0.412] \\
MXL minus TabPFN & [0.003, 0.167] & [0.169, 0.324] & [0.231, 0.402] \\
MXL minus Agent & [0.069, 0.349] & [0.186, 0.425] & [0.215, 0.431] \\
Completion minus TabPFN & [0.458, 0.643] & [0.629, 0.780] & [0.698, 0.848] \\
Completion minus Agent & [0.536, 0.809] & [0.659, 0.867] & [0.682, 0.871] \\
Completion minus MNL & [0.525, 0.696] & [0.518, 0.697] & [0.513, 0.699] \\
Completion minus NL & [0.277, 0.605] & [0.330, 0.628] & [0.355, 0.633] \\
Completion minus MXL & [0.353, 0.580] & [0.347, 0.579] & [0.340, 0.578] \\
\hline
\addlinespace[2pt]
\multicolumn{4}{@{}l}{\emph{Panel C: Critical values $\Gamma^*$}}\\[2pt]
NL minus MNL & $>3$ & $>3$ & $>3$ \\
MXL minus MNL & $>3$ & $>3$ & $>3$ \\
MXL minus NL & $1$ & $1.204$ & $1.264$ \\
MNL minus TabPFN & $1$ & $>3$ & $>3$ \\
MNL minus Agent & $1.469$ & $2.950$ & $>3$ \\
NL minus TabPFN & $1.650$ & $>3$ & $>3$ \\
NL minus Agent & $2.438$ & $>3$ & $>3$ \\
MXL minus TabPFN & $2.052$ & $>3$ & $>3$ \\
MXL minus Agent & $2.720$ & $>3$ & $>3$ \\
Completion minus TabPFN & $>3$ & $>3$ & $>3$ \\
Completion minus Agent & $>3$ & $>3$ & $>3$ \\
Completion minus MNL & $>3$ & $>3$ & $>3$ \\
Completion minus NL & $>3$ & $>3$ & $>3$ \\
Completion minus MXL & $>3$ & $>3$ & $>3$ \\
\hline
\end{tabular}
\end{table}

\FloatBarrier

\paragraph{Within-function averaging.}

Table~\ref{tab:app2-inner-point-comparison} shows that increasing the number of training samples per demand system from one to twenty changes restrictiveness estimates little: across all 27 estimates, the median absolute change is 0.0008 and the largest is 0.0242. Table~\ref{tab:app2-inner-bound-comparison} compares the corresponding sensitivity bounds and critical values. The qualitative robustness conclusions at $\Gamma=3$ are unchanged in all 27 directional structural comparisons. At $\Gamma=2$, absolute changes in the lower and upper bounds have a median of 0.0018 and a maximum of 0.0416.

\begin{table}[!htbp]
\caption{Restrictiveness with one and twenty training samples per demand system}
\label{tab:app2-inner-point-comparison}
\centering
\small\setlength{\tabcolsep}{4pt}
\begin{tabular}{@{}lrrr@{}}\hline
Procedure & $K=1$ & $K=20$ & Change \\ \hline
\multicolumn{4}{@{}l}{\emph{Market shares}}\\
MNL & 0.1826 & 0.1825 & -0.0001 \\
NL & 0.1726 & 0.1726 & 0.0000 \\
MXL & 0.1401 & 0.1399 & -0.0002 \\
\multicolumn{4}{@{}l}{\emph{Own-price elasticities}}\\
MNL & 0.2987 & 0.2995 & 0.0008 \\
NL & 0.4086 & 0.4059 & -0.0027 \\
MXL & 0.3103 & 0.3156 & 0.0053 \\
\multicolumn{4}{@{}l}{\emph{Cross-price elasticities}}\\
MNL & 0.2671 & 0.2669 & -0.0001 \\
NL & 0.3918 & 0.3813 & -0.0105 \\
MXL & 0.4160 & 0.4135 & -0.0025 \\
\hline\end{tabular}
\par\smallskip
{\footnotesize\raggedright\emph{Notes:} $n=100$.\par}
\end{table}

\begin{table}[!htbp]
\caption{Effect of within-function averaging on sensitivity bounds}
\label{tab:app2-inner-bound-comparison}
\centering
\small\setlength{\tabcolsep}{3pt}
\begin{tabular}{@{}p{0.35\linewidth}ccccc@{}}\hline
& & \multicolumn{2}{c}{$K=1$} & \multicolumn{2}{c}{$K=20$}\\
Comparison & $n$ & $\Gamma=2$ & $\Gamma^*$ & $\Gamma=2$ & $\Gamma^*$\\ \hline
MNL minus NL, shares & 100 & [0.005, 0.019] & $>3$ & [0.005, 0.018] & $>3$ \\
MNL minus MXL, shares & 100 & [0.019, 0.083] & $>3$ & [0.019, 0.083] & $>3$ \\
NL minus MXL, shares & 100 & [0.008, 0.072] & $2.909$ & [0.009, 0.072] & $2.947$ \\
NL minus MNL, own-price & 100 & [0.041, 0.232] & $>3$ & [0.041, 0.221] & $>3$ \\
MXL minus MNL, own-price & 100 & [-0.073, 0.093] & $1.111$ & [-0.065, 0.093] & $1.161$ \\
NL minus MXL, own-price & 100 & [-0.027, 0.278] & $1.707$ & [-0.034, 0.264] & $1.632$ \\
NL minus MNL, cross-price & 100 & [0.051, 0.261] & $>3$ & [0.052, 0.224] & $>3$ \\
MXL minus MNL, cross-price & 100 & [0.070, 0.247] & $>3$ & [0.074, 0.234] & $>3$ \\
MXL minus NL, cross-price & 100 & [-0.127, 0.157] & $1.132$ & [-0.092, 0.151] & $1.204$ \\
\hline\end{tabular}
\end{table}

\FloatBarrier
\endgroup

\section{Noise-Free versus Noisy Learning Curves}\label{sec:oa-gp-learning-curve}

We compare the noise-free average-case learning curve obtained from \eqref{eq:learning-curve} under squared-error discrepancy with learning curves studied for Gaussian-process (GP) regression. The comparison helps separate approximation error from the uncertainty caused by noisy observations. As in Section~\ref{subsec:procedures}, the learning curve is defined for a prediction procedure; a specified model class is needed only when interpreting its limit as model-class restrictiveness.

GP regression typically uses squared loss $\ell(y,y') = (y-y')^2$, so $d(f,g)$ is an $L^2(P_X)$ prediction error. \citet{legratietgarnier2015} consider a zero-mean GP prior with covariance kernel $k$, covariates $X_i \sim \mu$ on $\mathbb{R}^d$, and observations $Z(x_i) + \varepsilon_i$. The observation-noise variance grows with the sample size according to $\mathrm{Var}(\varepsilon_i) = n\tau$. They study the integrated mean squared error (IMSE) of the best linear unbiased predictor (BLUP),
	\[
	\mathrm{IMSE}_n
	= \int \sigma_n^2(x)\,d\mu(x),
	\]
	where $\sigma_n^2(x)$ is the posterior MSE of the predictor for the latent function
	value $Z(x)$. Their main result shows that, for a broad class of kernels,
	\begin{equation}
		\mathrm{IMSE}_n \;\xrightarrow{p}\;
		\sum_{p \geq 0} \frac{\tau \lambda_p}{\tau + \lambda_p},
		\text{ as } n\to\infty,
	\end{equation}
	where $(\lambda_p,\phi_p)$ are the eigenvalues and eigenfunctions of the covariance
	operator associated with $k$ and the design measure $\mu$.\footnote{See Section~3 of \citet{legratietgarnier2015} for a general statement and
		examples.}

To relate this setting to our notation, let $\mathcal{F}$ be the set of sample paths of the GP, let $\lambda_{\mathcal{F}}$ be the GP prior, and take $P_X = \mu$ and $\ell(y,y')=(y-y')^2$. The discrepancy is then the squared $L^2(\mu)$ distance between functions. Averaging the BLUP's prediction error over both $f \sim \lambda_{\mathcal{F}}$ and samples gives the same type of average-case learning curve. With noisy observations, however, this is a noisy analogue of the curve in \eqref{eq:learning-curve}, which uses noise-free labels.

For a noise-free procedure $\mathcal{A}$, suppose instead that the conditions of Proposition~\ref{prop:learning-curve} hold for a specified comparison class $\mathcal{F}_\Theta$. Then $\lim_n L_n(\mathcal{A})$ equals $\mathbb{E}_{f \sim \lambda_{\mathcal{F}}}[d(\mathcal{F}_\Theta,f)]$, the class's average approximation error relative to $(\mathcal{F},\lambda_{\mathcal{F}},d)$. Risk-consistency is what permits this interpretation, though it is not required to define the procedural learning curve itself.

The distinction concerns what remains in the limit, as well as the associated interpretation. With noise-free labels $Y_i = f(X_i)$ and risk-consistency, estimation error vanishes, leaving the average population approximation error. In \citet{legratietgarnier2015}, observation noise grows proportionally to $n$. For $\tau>0$ and a nondegenerate kernel, their limit remains positive even under the correctly specified GP prior. It measures the uncertainty that persists under this noise scaling, rather than a restriction on which functions the model can approximate.
\end{document}